\documentclass[12pt]{article}

\usepackage[utf8]{inputenc}

\usepackage[sectionbib]{natbib}
\usepackage{array,epsfig,rotating}
\usepackage{color}
\usepackage[dvipsnames]{xcolor}
\usepackage{hyperref}
\hypersetup{
  colorlinks,
  linkcolor={red!50!black},
  citecolor={blue!70!black},
  urlcolor={blue}
}

\usepackage{tabularx}
\usepackage{amsmath}
\usepackage{amssymb}
\usepackage{amsfonts}
\usepackage{amsthm}
\usepackage{float}
\usepackage{verbatim}
\usepackage{booktabs}
\usepackage{multirow}
\usepackage{caption}
\usepackage{subcaption}
\usepackage[toc,page]{appendix}
\usepackage[margin=1in]{geometry}
\usepackage{authblk}
\usepackage{bbm}

\newcommand{\blind}{1}

\newcommand{\darkred}[1]{\textcolor{red!70!black}{#1}}
\newcommand{\darkblue}[1]{\textcolor{blue!70!black}{#1}}

\definecolor{myblue}{rgb}{0.0265,    0.6137,    0.8135}

\def\bar{\widebar}
\def\tilde{\widetilde}
\def\bm{\mathbf}

\usepackage{mycommands}
\usepackage{authblk}

\newtheorem{theorem}{Theorem}

\newtheorem{lemma}{Lemma}

\newtheorem{proposition}{Proposition}
\newtheorem{definition}{Definition}
\newtheorem{example}{Example}
\newtheorem{remark}{Remark}

\newcommand{\aaa}{\boldsymbol \alpha}
\newcommand{\bo}{\boldsymbol }

\newcommand{\cc}{\boldsymbol c}

\newcommand{\zz}{\boldsymbol z}

\newcommand{\II}{\mathbf I}

\newcommand{\LL}{\mathbf L}
\newcommand{\PP}{\mathbf P}

\newcommand{\LLambda}{\bo{\Lambda}}

\newcommand{\one}{\mathbf 1}
\newcommand{\zero}{\mathbf 0}
\newcommand{\mca}{\mathcal A}
\newcommand{\diag}{\text{\normalfont{diag}}}
\newcommand{\vect}{\text{\normalfont{vec}}}
\newcommand{\rank}{\text{\normalfont{rank}}}

\newcommand{\gom}{\text{\normalfont{GoM}}}

\def\tilde{\widetilde}
\newcommand{\bcolon}{\boldsymbol{:}}
\newcommand{\bcdot}{\boldsymbol\cdot}

\DeclareMathOperator*{\argmax}{arg\,max}

\allowdisplaybreaks

\usepackage{tikz}
\usetikzlibrary{shapes, calc, shapes, arrows, positioning}

\tikzstyle{qedge}=[->,thick]

\tikzstyle{neuron}=[draw,circle,minimum size=25pt,inner sep=0pt, fill=black!10]

\tikzstyle{hidden}=[draw,circle,minimum size=25pt,inner sep=0pt, fill=white]

\usetikzlibrary{arrows.meta}
\tikzset{>={Latex[width=3mm,length=2mm]}}

\tikzstyle{arr}=[->, thick, black]

\usetikzlibrary{decorations.text}

\usetikzlibrary{shadows,shadings,shapes.symbols}
\tikzset{
    double color fill/.code 2 args={
        \pgfdeclareverticalshading[%
            tikz@axis@top,tikz@axis@middle,tikz@axis@bottom%
        ]{diagonalfill}{100bp}{%
            color(0bp)=(tikz@axis@bottom);
            color(50bp)=(tikz@axis@bottom);
            color(50bp)=(tikz@axis@middle);
            color(50bp)=(tikz@axis@top);
            color(100bp)=(tikz@axis@top)
        }
        \tikzset{shade, left color=#1, right color=#2, shading=diagonalfill}
    }
}

\def\spacingset#1{\renewcommand{\baselinestretch}%
{#1}\small\normalsize} \spacingset{1}
\spacingset{1}

\begin{document}
\if1\blind
{
\title{Dimension-Grouped Mixed Membership Models for Multivariate Categorical Data}
\date{}

\author[1]{Yuqi Gu}
\author[2]{Elena A. Erosheva}
\author[3]{Gongjun Xu}
\author[4]{David B. Dunson}

\affil[1]{Department of Statistics, Columbia University}
\affil[2]{Department of Statistics, School of Social Work, and the Center for Statistics and the Social Sciences, University of Washington}
\affil[3]{Department of Statistics, University of Michigan}
\affil[4]{Department of Statistical Science, Duke University}

\maketitle
} \fi

\if0\blind
{
  \bigskip
  \bigskip
  \bigskip
  \begin{center}
    {\LARGE
    Dimension-Grouped Mixed Membership Models for Multivariate Categorical Data
    \par}
\end{center}
  \medskip
} \fi

\begin{abstract}
Mixed Membership Models (MMMs) are a popular family of latent structure models for complex multivariate data.  Instead of forcing each subject to belong to a single cluster, MMMs incorporate a vector of subject-specific weights characterizing partial membership across clusters.  With this flexibility come challenges in uniquely identifying, estimating, and interpreting the parameters. In this article, we propose a new class of \textit{Dimension-Grouped} MMMs (Gro-M$^3$s) for multivariate categorical data, which improve parsimony and interpretability.
In Gro-M$^3$s, observed variables are partitioned into groups such that the latent membership is constant for variables within a group but can differ across groups.
Traditional latent class models are obtained when all variables are in one group, while traditional MMMs are obtained when each variable is in its own group.
The new model corresponds to a novel decomposition of probability tensors.
Theoretically, we derive transparent identifiability conditions for both the unknown grouping structure and model parameters in general settings.
Methodologically, we propose a Bayesian approach for Dirichlet Gro-M$^3$s to inferring the variable grouping structure and estimating model parameters. Simulation results demonstrate good computational performance and empirically confirm the identifiability results. 
We illustrate the new methodology through 
{applications to a functional disability survey dataset and a personality test dataset}.
\end{abstract}


\section{Introduction}\label{sec-intro}

Mixed membership models (MMMs) are a popular family of latent structure models for complex multivariate data. Building on classical latent class and finite mixture models \citep[][]{mclachlan2000}, which assign each subject to a single cluster, MMMs include a vector of probability weights characterizing partial membership.
MMMs have seen many applications in a wide variety of fields, including social science surveys \citep{erosheva2007aoas}, topic modeling and text mining \citep{blei2003lda},
population genetics and bioinformatics \citep{pritchard2000genetics, saddiki2015glad},
biological and social networks \citep{airoldi2008mmsbm}, collaborative filtering \citep{mackey2010mixed}, and data privacy \citep{manrique2012jasa}; see \cite{airoldi2014handbook} for more examples.

Although MMMs are conceptually  appealing and very flexible, with the rich modeling capacity come challenges in identifying, accurately estimating, and interpreting the parameters.
MMMs have been popular in many applications, yet key theoretical issues remain understudied.
The handbook of \cite{airoldi2014handbook} emphasized theoretical difficulties of MMMs ranging from non-identifiability to multi-modality of the likelihood.  Finite mixture models have related challenges, and the additional complexity of the \textit{individual-level} mixed membership incurs extra difficulties.
A particularly important case is MMMs for multivariate categorical data, such as survey response 
 \citep{woodbury1978gom, erosheva2007aoas, manrique2012jasa}.  In this setting, MMMs provide an attractive alternative to the latent class model of  
 \cite{goodman1974}.
 However, little is known about what is fundamentally identifiable and learnable from observed data under such models.

 Identifiability is a key property of a statistical model, meaning that the model parameters can be uniquely obtained from the observables.
An identifiable model is a prerequisite for reproducible statistical inferences and reliable applications.
Indeed, interpreting parameters estimated from an unidentifiable model is meaningless, and may lead to misleading conclusions in practice. 
It is thus important to study the identifiability of MMMs and to provide theoretical support to back up the conceptual appeal. Even better would be to expand the MMM framework to allow 
variations that aid interpretability and identifiability.
With this motivation, and focused on mixed membership modeling of multivariate categorical data, this paper makes the following key contributions.

We propose a new class of models for multivariate categorical data, which retains the same flexibility offered by MMMs, while favoring greater parsimony and interpretability.
The key innovation is to allow the $p$-dimensional latent membership vector to belong to $G$ $(G\ll p)$ groups; memberships are the same for different variables within a group but can differ across groups.
We deem the new model the \textit{Dimension-Grouped Mixed Membership Model} (Gro-M$^3$).
Gro-M$^3$ improves interpretability by allowing the potentially high-dimensional observed variables to belong to a small number of meaningful groups. 
Theoretically, we show that both the continuous model parameters, and the discrete variable grouping structure, can be identified from the data for models in the Gro-M$^3$ class under transparent conditions on how the variables are grouped.
{This challenging identifiability issue is addressed by carefully leveraging the dimension-grouping structure to write the model as certain structured tensor products, and then invoking Kruskal's fundamental theorem on the uniqueness of three-way tensor decompositions \citep{kruskal1977three, allman2009}.}

To illustrate the methodological usefulness of the proposed class of models, we consider a special case in which each subject's mixed membership proportion vector follows a Dirichlet distribution. This is among the most popular modeling assumptions underlying various MMMs \citep{blei2003lda, erosheva2007aoas, manrique2012jasa, zhao2018gmm}.
For such a Dirichlet Gro-M$^3$, we employ a Bayesian inference procedure and develop a Metropolis-Hastings-within-Gibbs algorithm for posterior computation. 
The algorithm has excellent computational performance.
Simulation results demonstrate this approach can accurately learn the identifiable quantities of the model, including both the variable-grouping structure and the continuous model parameters. This also empirically confirms the model identifiability result.

The rest of this paper is organized as follows. Section \ref{sec-model} reviews existing mixed membership models, introduces the proposed Gro-M$^3$, and provides an interesting probabilistic tensor decomposition perspective of the models.
Section \ref{sec-id} is devoted to the study of the identifiability of the new model.
Section \ref{sec-bayes} focuses on the Dirichlet distribution induced Gro-M$^3$ and proposes a Bayesian inference procedure.
Section \ref{sec-simu} includes simulation studies and Section \ref{sec-real} applies the new model to reanalyze the NLTCS disability survey data.
Section \ref{sec-discuss} provides discussions.

\section{Dimension-Grouped Mixed Membership Models}\label{sec-model}
\subsection{Existing Mixed Membership Models}
In this subsection, we briefly review the existing MMM literature to give our proposal appropriate context.
Let $K$ be the number of extreme latent profiles.
Denote the $K$-dimensional probability simplex by $\Delta^{K-1} = \{(\pi_1,\ldots, \pi_K): \pi_k\geq 0 \text{ for all } k,\; \sum_{k=1}^K \pi_k = 1\}$.
Each subject $i$ has an individual proportion vector $\bo\pi_i = (\pi_{i,1}, \ldots, \pi_{i,K})\in\Delta^{K-1}$, which indicates the degrees to which subject $i$ is a member of the $K$ extreme profiles.
The general mixed membership models summarized in \cite{airoldi2014handbook} have the following distribution,
\begin{align}\label{eq-mmm}
	p\left(\left\{y^{(r)}_{i,1},\ldots,y_{i,p}^{(r)}\right\}_{r=1}^R\right)
	=&~\int_{\Delta^{K-1}} \prod_{j=1}^p \prod_{r=1}^R \left( \sum_{k=1}^K \pi_{i,k} f(y_{i,j}^{(r)}\mid \bo\lambda_{j,k}) \right) dD_{\aaa}(\bo \pi_i),
\end{align}
where $\bo \pi_i = (\pi_{i,1},\ldots,\pi_{i,K})$ follows the distribution $D_{\aaa}$ and is integrated out; the $\aaa$ refers to some generic population parameters depending on the specific model.
The hierarchical Bayesian representation for the model in \eqref{eq-mmm} can be written as follows.
\begin{align*}
y_{ij}^{(1)}, \ldots, y_{ij}^{(R)} \mid z_{ij}=k & \stackrel{\text{i.i.d.}}{\sim}  \text{Categorical}([d_j];~\bo\lambda_{j,k}),~~ j\in[p];\\ 
z_{i1}, \ldots, z_{ip} \mid \bo\pi_i &\stackrel{\text{i.i.d.}}{\sim}  \text{Categorical}([K];~\bo\pi_i),~~ i\in[n];\\
\bo\pi_1,\ldots,\bo\pi_n 
&\stackrel{\text{i.i.d.}}{\sim} D_{\aaa}.
\end{align*}
{where ``i.i.d.'' is short for ``independent and identically distributed''.}
The number $p$ in \eqref{eq-mmm} is the number of ``characteristics'', and $R$ is the number of ``replications'' per characteristic. 
As shown in \eqref{eq-mmm}, for each characteristic $j$, there are a corresponding set of $K$ conditional distributions indexed by parameter vectors $\{\bo\lambda_{j,k}:\, k=1,\ldots,K\}$.
Many different mixed membership models are special cases of the general setup \eqref{eq-mmm}.
For example, the popular Latent Dirichlet Allocation (LDA) \citep{blei2003lda, blei2012acm, anandkumar2014tensor} for topic modeling takes a document $i$ as a subject, and assumes there is only $p=1$ distinct characteristic (one single set of $K$ topics which are distributions over the word  vocabulary) with $R > 1$ replications (a document $i$ contains $R$ words which are conditionally i.i.d. given $\bo\pi_i$); LDA further specifies $D_{\aaa}(\bo \pi_i)$ to be the Dirichlet distribution with parameters $\bo\alpha=(\alpha_1,\ldots,\alpha_K)$.

Focusing on MMMs for  multivariate categorical data, there are generally many characteristics with $p \gg 1$ and one replication of each characteristic with $R=1$ in \eqref{eq-mmm}.
Each variable $y_{i,j}\in\{1,\ldots, d_j\}$ takes one of  $d_j$ unordered categories.
For each subject $i$, the observables $\bo y_i = (y_{i,1}, \ldots, y_{i,p})^\top$ are a vector of $p$ categorical variables. 
MMMs for such data are traditionally called Grade of Membership models (GoMs) \citep{woodbury1978gom}.
GoMs have been extensively used in applications, such as disability survey data \citep{erosheva2007aoas}, scholarly publication data \citep{erosheva2004pnas}, and data disclosure risk and privacy \citep{manrique2012jasa}.
GoMs are also useful for psychological measurements where data are Likert scale responses to psychology survey items, and educational assessments where data are students' correct/wrong answers to test questions \citep[e.g.][]{shang2021aoas}.

In GoMs, the conditional distribution $f(y_{i,j}\mid \bo\lambda_{j,k})$ in \eqref{eq-mmm} can be written as 
$\mathbb P(y_{i,j} \mid \bo\lambda_{j,k}) = \prod_{c_j=1}^{d_j} \lambda_{j,c_j,k}^{\mathbb I(y_{i,j}=c_j)}$.
Hence, the probability mass function of $\bo y_i$ in a GoM is
\begin{align}\label{eq-gompmf}
	p^{\text{GoM}}\left(y_{i,1},\ldots,y_{i,p}\mid \bo\Lambda, \aaa\right)
	&= \int_{\Delta^{K-1}}
	\prod_{j=1}^p \left[ \sum_{k=1}^K \pi_{i,k}  \prod_{c_j=1}^{d_j} \lambda_{j, c_j, k}^{\mathbb I(y_{i,j}=c_j)} \right] dD_{\aaa}(\bo\pi_i).
\end{align}
The hierarchical Bayesian representation for the model in \eqref{eq-gompmf} can be written as follows.
\begin{align*}
y_{ij}\mid z_{ij}=k & \stackrel{\text{i.i.d.}}{\sim}  \text{Categorical}([d_j];~\bo\lambda_{j,k}),~~ j\in[p];\\ 
z_{i1}, \ldots, z_{ip} \mid\bo\pi_i &\stackrel{\text{i.i.d.}}{\sim}  \text{Categorical}([K];~\bo\pi_i),~~ i\in[n];\\
\bo\pi_1,\ldots,\bo\pi_n 
&\stackrel{\text{i.i.d.}}{\sim} D_{\aaa}.
\end{align*}
See a graphical model representation of the GoM with sample size $n$ in Figure \ref{fig-graph}(b), where individual latent indicator variables $(z_{i,1}, \ldots, z_{i,p})\in[K]^p$ are introduced to better describe the data generative process. 

\begin{figure}[h!]\centering

\resizebox{0.47\textwidth}{!}{
    \begin{tikzpicture}[scale=1.8]
    \node (v1)[neuron] at (0, 0) {$y_{i,1}$};
    \node (v2)[neuron] at (0.8, 0) {$y_{i,2}$};
    \node (v3)[neuron] at (1.6, 0) {$\cdots$};
    \node (v4)[neuron] at (2.4, 0) {$\cdots$};
    \node (v5)[neuron] at (3.2, 0) {$\cdots$};
    \node (v6)[neuron] at (4, 0)   {$y_{i,p}$};

    \node (h0)[hidden] at (2, 1.8) {$z_i$};
    \node (h0_text)[] at (3, 1.82) {$z_i\in[K]$};
    
    \node (hout)[] at (2, 3.6) {$\bo\nu$};
    
    \draw[qedge] (hout) -- (h0);

    \draw[qedge] (h0) -- (v1) node [midway,above=-0.12cm,sloped] {};
    
    \draw[qedge] (h0) -- (v2) node [midway,above=-0.12cm,sloped] {};
    
    \draw[qedge] (h0) -- (v3) node [midway,above=-0.12cm,sloped] {};
    
    \draw[qedge] (h0) -- (v4) node [midway,above=-0.12cm,sloped] {};
    
    \draw[qedge] (h0) -- (v5) node [midway,above=-0.12cm,sloped] {};
    
    \draw[qedge] (h0) -- (v6) node [midway,above=-0.12cm,sloped] {}; 

    \draw[rounded corners, thick] (-0.4, -0.5) rectangle (4.4, 3);
    \node (nn) at (-0.1, 2.7) {\large $n$};
    
    \node (lambda1)[] at (0, -0.8) {$\mathbf{\Lambda}_{1,:,:}$};
    \node (lambda2)[] at (0.8, -0.8) {$\mathbf{\Lambda}_{2,:,:}$};
    \node (lambda3)[] at (1.6, -0.8) {$\cdots$};
    \node (lambda4)[] at (2.4, -0.8) {$\cdots$};
    \node (lambda5)[] at (3.2, -0.8) {$\cdots$};
    \node (lambda6)[] at (4, -0.8) {$\mathbf{\Lambda}_{p,:,:}$};
    
    \draw[qedge] (lambda1) -- (v1);
    \draw[qedge] (lambda2) -- (v2);
    \draw[qedge] (lambda3) -- (v3);
    \draw[qedge] (lambda4) -- (v4);
    \draw[qedge] (lambda5) -- (v5);
    \draw[qedge] (lambda6) -- (v6);
    
    \node at (2, -1.2) {(a) Latent Class Model};
    \node at (2, -1.5) {(Probabilistic CP Decomposition)};
\end{tikzpicture}
}
\hfill
\resizebox{0.47\textwidth}{!}{
    \begin{tikzpicture}[scale=1.8]
    \node (v1)[neuron] at (0, 0) {$y_{i,1}$};
    \node (v2)[neuron] at (0.8, 0) {$y_{i,2}$};
    \node (v3)[neuron] at (1.6, 0) {$\cdots$};
    \node (v4)[neuron] at (2.4, 0) {$\cdots$};
    \node (v5)[neuron] at (3.2, 0) {$\cdots$};
    \node (v6)[neuron] at (4, 0)   {$y_{i,p}$};

    \node (h1)[hidden] at (0, 1.2) {$z_{i,1}$};
    \node (h2)[hidden] at (0.8, 1.2) {$\cdots$};
    \node (h3)[hidden] at (1.6, 1.2) {$\cdots$};
    \node (h4)[hidden] at (2.4, 1.2) {$\cdots$};
    \node (h5)[hidden] at (3.2, 1.2) {$\cdots$};
    \node (h6)[hidden] at (4, 1.2)   {$z_{i,p}$};
    
    \node (h0)[hidden] at (2, 2.4) {$\bo\pi_i$};
    \node (h0_text)[] at (3, 2.45) {$\bo\pi_i\in \Delta^{K-1}$};
    
    \node (hout)[] at (2, 3.6) {$\aaa$};
    
    \draw[qedge] (hout) -- (h0);

    \draw[qedge] (h0) -- (h1); 
    \draw[qedge] (h0) -- (h2);
    \draw[qedge] (h0) -- (h3);
    \draw[qedge] (h0) -- (h4); 
    \draw[qedge] (h0) -- (h5);
    \draw[qedge] (h0) -- (h6);

    \draw[qedge] (h1) -- (v1) node [midway,above=-0.12cm,sloped] {}; 
    
    \draw[qedge] (h2) -- (v2) node [midway,above=-0.12cm,sloped] {};  
    
    \draw[qedge] (h3) -- (v3) node [midway,above=-0.12cm,sloped] {}; 
    
    \draw[qedge] (h4) -- (v4) node [midway,above=-0.12cm,sloped] {}; 
    
    \draw[qedge] (h5) -- (v5) node [midway,above=-0.12cm,sloped] {}; 
    
    \draw[qedge] (h6) -- (v6) node [midway,above=-0.12cm,sloped] {}; 

    \draw[rounded corners, thick] (-0.4, -0.5) rectangle (4.4, 3);
    \node (nn) at (-0.1, 2.7) {\large $n$};
    
    \node (lambda1)[] at (0, -0.8) {$\mathbf{\Lambda}_{1,:,:}$};
    \node (lambda2)[] at (0.8, -0.8) {$\mathbf{\Lambda}_{2,:,:}$};
    \node (lambda3)[] at (1.6, -0.8) {$\cdots$};
    \node (lambda4)[] at (2.4, -0.8) {$\cdots$};
    \node (lambda5)[] at (3.2, -0.8) {$\cdots$};
    \node (lambda6)[] at (4, -0.8) {$\mathbf{\Lambda}_{p,:,:}$};
    
    \draw[qedge] (lambda1) -- (v1);
    \draw[qedge] (lambda2) -- (v2);
    \draw[qedge] (lambda3) -- (v3);
    \draw[qedge] (lambda4) -- (v4);
    \draw[qedge] (lambda5) -- (v5);
    \draw[qedge] (lambda6) -- (v6);
    
    \node at (2, -1.2) {(b) Grade of Membership Model, $z_{i,j}\in[K]$};
    \node at (2, -1.5) {(Probabilistic Tucker Decomposition)};
\end{tikzpicture}
}

\bigskip
\resizebox{0.47\textwidth}{!}{
    \begin{tikzpicture}[scale=1.8]
    \draw [dotted, thick, rounded corners, red!70!black] (-0.3, -0.4) rectangle (4.3, 1.55);
    
    \node (nn) at (-0.1, 2.7) {\large $n$};
        
    \node (v1)[neuron] at (0, 0) {$y_{i,1}$};
    \node (v2)[neuron] at (0.8, 0) {$y_{i,2}$};
    \node (v3)[neuron] at (1.6, 0) {$\cdots$};
    \node (v4)[neuron] at (2.4, 0) {$\cdots$};
    \node (v5)[neuron] at (3.2, 0) {$\cdots$};
    \node (v6)[neuron] at (4, 0)   {$y_{i,p}$};
    
    \node (h1)[hidden] at (1.0, 1.2) {$z_{i,1}$};
    \node (h2)[hidden] at (2.0, 1.2) {$\cdots$};
    \node (h3)[hidden] at (3.0, 1.2) {$z_{i,G}$};

    \node (h0)[hidden] at (2, 2.4) {$\bo\eta_i$};
    \node (h0_text)[] at (3.2, 2.45) {$f(\bo\eta_i)\in \Delta^{K-1}$};
    
    \node (gmat) at (0.6, 3.6) {\darkred{$\LL$}};
    \draw[qedge] (gmat) -- (0.6, 1.55);
    
    \node (h_mu) at (1.2, 3.6) {$\bo\mu$};
    \node (h_sigma) at (2.8, 3.6) {$\mathbf{\Sigma}$};
    
    \draw[qedge] (h_mu) -- (h0);
    \draw[qedge] (h_sigma) -- (h0);

    \draw[qedge] (h0) -- (h1); 
    \draw[qedge] (h0) -- (h2);
    \draw[qedge] (h0) -- (h3);

    \draw[qedge, red!70!black, thick] (h1) -- (v1) node [midway,above=-0.12cm,sloped] {}; 
    
    \draw[qedge, red!70!black, thick] (h1) -- (v2) node [midway,above=-0.12cm,sloped] {};
    
    \draw[qedge, red!70!black, thick] (h2) -- (v3) node [midway,above=-0.12cm,sloped] {};
    
    \draw[qedge, red!70!black, thick] (h2) -- (v4) node [midway,above=-0.12cm,sloped] {}; 
    
    \draw[qedge, red!70!black, thick] (h3) -- (v5) node [midway,above=-0.12cm,sloped] {}; 
    
    \draw[qedge, red!70!black, thick] (h3) -- (v6) node [midway,above=-0.12cm,sloped] {}; 
    
    \draw[rounded corners, thick] (-0.4, -0.5) rectangle (4.4, 3);

    \node (lambda1)[] at (0, -0.8) {$\mathbf{\Lambda}_{1,:,:}$};
    \node (lambda2)[] at (0.8, -0.8) {$\mathbf{\Lambda}_{2,:,:}$};
    \node (lambda3)[] at (1.6, -0.8) {$\cdots$};
    \node (lambda4)[] at (2.4, -0.8) {$\cdots$};
    \node (lambda5)[] at (3.2, -0.8) {$\cdots$};
    \node (lambda6)[] at (4, -0.8) {$\mathbf{\Lambda}_{p,:,:}$};
    
    \draw[qedge] (lambda1) -- (v1);
    \draw[qedge] (lambda2) -- (v2);
    \draw[qedge] (lambda3) -- (v3);
    \draw[qedge] (lambda4) -- (v4);
    \draw[qedge] (lambda5) -- (v5);
    \draw[qedge] (lambda6) -- (v6);

    \node at (2, -1.2) {(c) (New) Gro-M$^3$, $f(\bo\eta_i)$ logit normal};
    \node at (2, -1.5) {(Probabilistic Hybrid Decomposition)};
\end{tikzpicture}
}
\hfill
\resizebox{0.47\textwidth}{!}{
    \begin{tikzpicture}[scale=1.8]
    \draw [dotted, thick, rounded corners, red!70!black] (-0.3, -0.4) rectangle (4.3, 1.55);
    \node (nn) at (-0.1, 2.7) {\large $n$};
        
    \node (v1)[neuron] at (0, 0) {$y_{i,1}$};
    \node (v2)[neuron] at (0.8, 0) {$y_{i,2}$};
    \node (v3)[neuron] at (1.6, 0) {$\cdots$};
    \node (v4)[neuron] at (2.4, 0) {$\cdots$};
    \node (v5)[neuron] at (3.2, 0) {$\cdots$};
    \node (v6)[neuron] at (4, 0)   {$y_{i,p}$};
    
    \node (h1)[hidden] at (1.0, 1.2) {$z_{i,1}$};
    \node (h2)[hidden] at (2.0, 1.2) {$\cdots$};
    \node (h3)[hidden] at (3.0, 1.2) {$z_{i,G}$};
    
    \node (h0)[hidden] at (2, 2.4) {$\bo\pi_i$};
    \node (h0_text)[] at (3, 2.45) {$\bo\pi_i\in \Delta^{K-1}$};

    \node (gmat) at (0.6, 3.6) {\darkred{$\LL$}};
    \draw[qedge] (gmat) -- (0.6, 1.55);
    
    \node (hout)[] at (2, 3.6) {$\aaa$};
    
    \draw[qedge] (hout) -- (h0);

    \draw[qedge] (h0) -- (h1); 
    \draw[qedge] (h0) -- (h2);
    \draw[qedge] (h0) -- (h3);

    \draw[qedge, red!70!black, thick] (h1) -- (v1) node [midway,above=-0.12cm,sloped] {}; 
    
    \draw[qedge, red!70!black, thick] (h1) -- (v2) node [midway,above=-0.12cm,sloped] {};
    
    \draw[qedge, red!70!black, thick] (h2) -- (v3) node [midway,above=-0.12cm,sloped] {};
    
    \draw[qedge, red!70!black, thick] (h2) -- (v4) node [midway,above=-0.12cm,sloped] {}; 
    
    \draw[qedge, red!70!black, thick] (h3) -- (v5) node [midway,above=-0.12cm,sloped] {}; 
    
    \draw[qedge, red!70!black, thick] (h3) -- (v6) node [midway,above=-0.12cm,sloped] {}; 

    \draw[rounded corners, thick] (-0.4, -0.5) rectangle (4.4, 3);
    
    \node (lambda1)[] at (0, -0.8) {$\mathbf{\Lambda}_{1,:,:}$};
    \node (lambda2)[] at (0.8, -0.8) {$\mathbf{\Lambda}_{2,:,:}$};
    \node (lambda3)[] at (1.6, -0.8) {$\cdots$};
    \node (lambda4)[] at (2.4, -0.8) {$\cdots$};
    \node (lambda5)[] at (3.2, -0.8) {$\cdots$};
    \node (lambda6)[] at (4, -0.8) {$\mathbf{\Lambda}_{p,:,:}$};
    
    \draw[qedge] (lambda1) -- (v1);
    \draw[qedge] (lambda2) -- (v2);
    \draw[qedge] (lambda3) -- (v3);
    \draw[qedge] (lambda4) -- (v4);
    \draw[qedge] (lambda5) -- (v5);
    \draw[qedge] (lambda6) -- (v6);
    
    \node at (2, -1.2) {(d) (New) Gro-M$^3$, $\bo\pi_i$ Dirichlet};
    \node at (2, -1.5) {(Probabilistic Hybrid Decomposition)};
\end{tikzpicture}
}

\caption{Graphical model representations of LCMs in (a), GoMs in (b), and the proposed family of Gro-M$^3$s with two examples in (c), (d). Shaded nodes $\{y_{i,j}\}$ are observed variables, white nodes are latent variables, quantities outside each solid box are population parameters.
In (c) and (d), the dotted red box is the key dimension-grouping structure, where the red edges from $\{z_{i,g}\}$ to $\{y_{i,j}\}$ correspond to entries of ``1'' in the grouping matrix $\LL$.}
\label{fig-graph}
\end{figure}
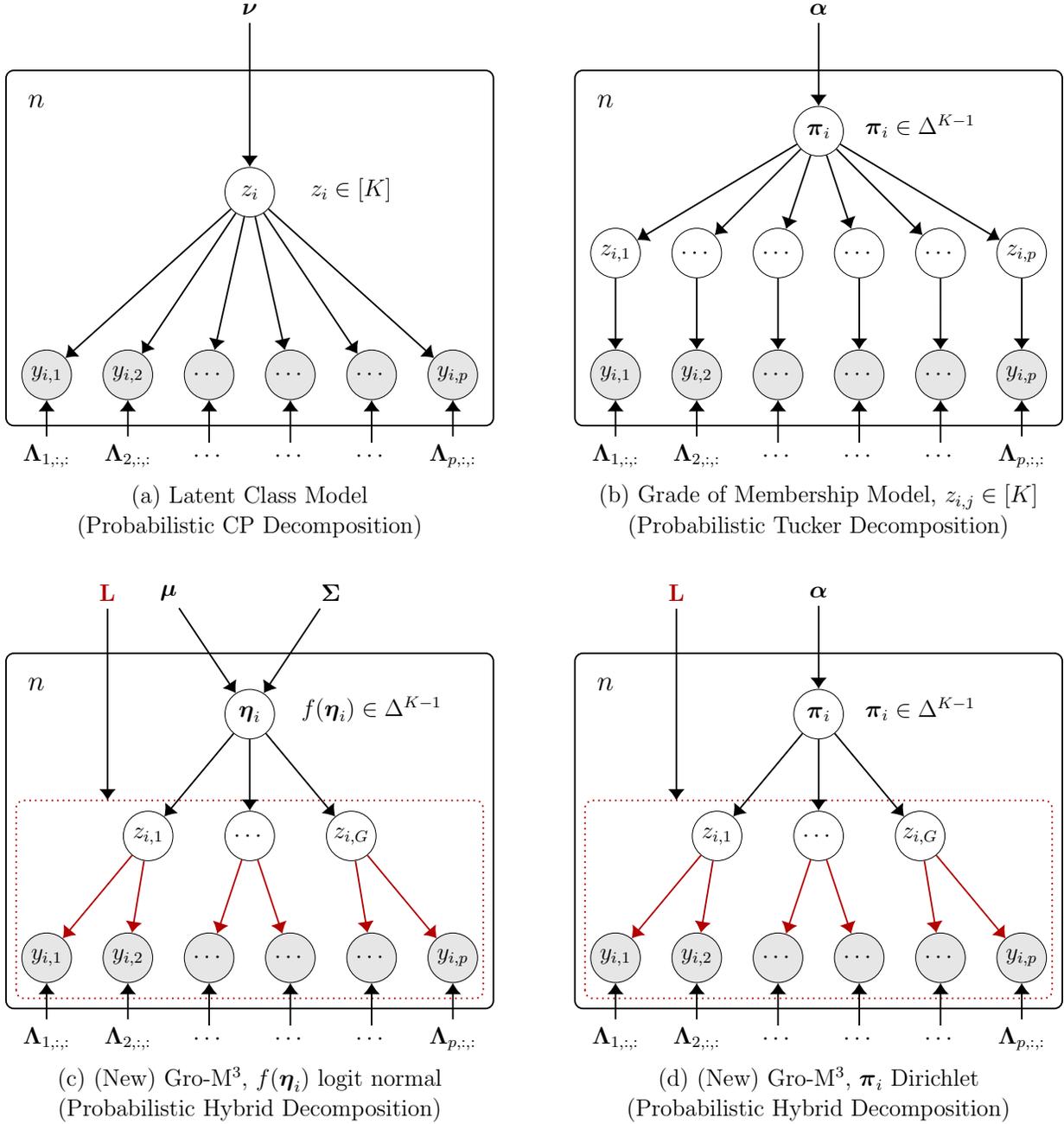

We emphasize that the case with $p>1$ and $R=1$ is fundamentally different from the topic models with $p=1$ and $R>1$, with the former typically involving many more parameters. 
{
This is because the ``bag-of-words'' assumption in the topic model with $R>1$ disregards word order in a document and assumes all words in a document are \emph{exchangeable}. 
In contrast, our mixed-membership model for multivariate categorical data {does not assume} a subject's responses to the $p$ items in a survey/questionnaire are exchangeable. 
In other words, given a subject's mixed membership vector $\bo\pi_i$, his/her responses to the $p$ items are {independent \emph{but not} identically distributed} (because they follow categorical distributions governed by $p$ different sets of parameters $\{\bo \lambda_{j,k} \in \mathbb R^{d}:~ k\in[K\}$ for $j=1,\ldots,p$); whereas in a topic model, given a document's latent topic proportion vector $\bo \pi_i$, the $p$ words in the document are {independent \emph{and} identically distributed},  following the categorical distribution with the same set of parameters $\{\bo\lambda_{k} \in \mathbb R^V:~ k\in[K]\}$ (here $V$ denotes the vocabulary size).
In this sense, the GoM model has greater modeling flexibility than topic models and are more suitable for modeling item response data, where it is inappropriate to assume that the items in the survey/questionnaire are exchangeable or share the same set of parameters.}
This fact is made clear also in Figure \ref{fig-graph}(b), where for each $j\in[p]$ there is a population quantity, the parameter node $\bo\Lambda_{j,:,:}$ (also denoted by $\bo\Lambda_j$ for simplicity), that governs its distribution.
{Thus identifiability is a much greater challenge for GoM models.
To our best knowledge, the identifiability issue of the grade-of-membership (GoM) models for item response data considered in \cite{woodbury1978gom} and \cite{erosheva2007aoas} has not been rigorously investigated so far. 
Motivated by the difficulty of identifying GoM in its original setting due to the large parameter complexity, we next propose a new modeling grouping component to enhance identifiability.
Our resulting model still does not make any exchangeability assumption of the items, but rather leverages the variable grouping structure to reduce model complexity. 
}

\subsection{New Modeling Component: the Variable Grouping Structure}

Generalizing Grade of Membership models for multivariate categorical data, we propose a new structure that groups the $p$ observed variables in the following sense: any subject's latent membership 
is the same for variables within a group but can differ across groups.
To represent the key structure of how the $p$ variables are partitioned into $G$ groups, we introduce a notation of the \textit{grouping matrix} $\LL = (\ell_{j,g})$. The $\LL$ is a $p\times G$ matrix with binary entries, with rows indexed by the $p$ variables and columns by the $G$ groups.
Each row $j$ of $\LL$ has exactly one entry of ``1'' indicating group assignment. In particular, 
\begin{align}
\LL = (\ell_{j,g})_{p\times G},\qquad
    \ell_{j,g}=
    \begin{cases}
    1, & \text{if the $j$th variable belongs to the $g$th group};\\
    0, & \text{otherwise}.
    \end{cases}
\end{align}
Our key specification is the following generative process {in the form of a hierarchical Bayesian representation},
\begin{flalign}\notag
\text{Gro-M$^3$:}
&&\left\{y_{i,j}\right\}_{\ell_{j,g}=1}  \mid z_{i,g}=k 
    ~& \stackrel{\text{ind.}}{\sim} ~ \text{Categorical}\left([d_j];\,\left(\lambda_{j, 1, k}, \cdots, \lambda_{j, d_j, k}\right)\right),\quad g\in[G]; &
    \\ \label{eq-model}
&&
	z_{i,1},\ldots,z_{i,G}\mid \bo \pi_i 
	 ~& \stackrel{\text{i.i.d}}{\sim} ~ \text{Categorical}([K];\, \bo \pi_i); \\ \notag
&&\bo \pi_1,\ldots,\bo\pi_n ~&\stackrel{\text{i.i.d.}}{\sim} D_{\aaa}.
\end{flalign}
{where ``ind.'' is short for ``independent'', meaning that conditional on $z_{i,g}=k$, subject $i$'s observed responses to items in group $g$ are independently generated.}
Hence, given the population parameters $(\LL, \bo\Lambda, \aaa)$, the probability distribution of $\bo y_i$ can be written as
\begin{align*}
     p^{\text{Gro-M$^3$}}\left(y_{i,1},\ldots,y_{i,p}\mid \LL, \bo\Lambda, \aaa\right)
	= &~
    \int_{\Delta^{K-1}} \prod_{g=1}^G \left[\sum_{k=1}^K \pi_{i,k} \prod_{j:\, \ell_{j,g} = 1} \prod_{c_j=1}^{d_j} \lambda_{j, c_j, k}^{\mathbb I(y_{i,j}=c_j)}  \right] dD_{\aaa}(\bo \pi_i).
\end{align*}
For a sample with $n$ subjects, assume the observed responses $\bo{y}_1, \ldots, \bo y_n$ are independent and identically distributed according to the above model.

We visualize the proposed model as a probabilistic graphical model to highlight connections to and differences from existing latent structure models for multivariate categorical data.
In Figure \ref{fig-graph}, we show the graphical model representations of two popular latent structure models for multivariate categorical data in (a) and (b), and for the newly proposed Gro-M$^3$ in (c) and (d). The $\bo\Lambda_{j}$ for $j\in[p]$ denotes a $d_j\times K$ matrix with entries $\lambda_{j,c_j,k}$. Each column of $\bo\Lambda_{j}$ characterizes a conditional probability distribution of variable $y_j$ given a particular extreme latent profile.  We emphasize that the variable grouping is done at the level of the latent allocation variables $z$, and 
that the $\mathbf{\Lambda}_{j}$ parameters are still free without constraints just as they are in traditional LCMs or GoMs.
{From the visualizations in Figure \ref{fig-graph} we can also easily distinguish our proposed model from another popular family of methods, the co-clustering methods \citep{dhillon2003information, govaert2013co}. Co-clustering usually refers to simultaneously clustering the subjects and clustering the variables, where subjects within a cluster exhibit similar behaviors and variables within a cluster also share similar characteristics. Our Gro-M$^3$, however, does not restrict the $p$ variables to have similar characteristics within groups, but rather allows them to have entirely free parameters $\LLambda_{1}, \ldots, \LLambda_p$. The ``dimension-grouping'' happens at the latent level by constraining the latent allocations behind the $p$ variables to be grouped into $G$ statuses. Such groupings give rise to a novel probabilistic hybrid tensor decomposition visualized in Figure \ref{fig-graph}(c)--(d); see the next Section \ref{sec-tensor} for details.}

Other than facilitating model identifiability (see Section \ref{sec-id}), our dimension-grouping modeling assumption is also motivated by real-world applications. 
{In general, our new model Gro-M$^3$ with the variable grouping component can apply to any multivariate categorical data to simultaneously model individual mixed membership and capture variable similarity. 
For example, Gro-M$^3$ can be applied to  survey/questionnaire response data in social sciences, where it is not only of interest to model subjects' partial membership to several extreme latent profiles, but also of interest to identify blocks of items which share similar measurement goals within each block. 
We next provide numerical evidence to demonstrate the merit of the variable grouping modeling component.
For a dataset simulated from Gro-M$^3$ (in the setting as the later Table \ref{tab-acc-K3}) and also the real-world IPIP personality test dataset (analyzed in the later Section \ref{sec-real}), we calculate the sample Cramer's V between item pairs.
Cramer's V is a classical measure of association between two categorical variables, which gives a value between 0 and 1, with larger values indicating stronger association.
Figure \ref{fig-crv4} presents the plots of the sample Cramer's V matrix for the simulated data and the real IPIP data.
This figure shows that the pairwise item dependence for the Gro-M$^3$-simulated data looks quite similar to the real-world personality test data. Indeed, after fitting the Gro-M$^3$ to this IPIP personality test dataset, the estimated model-based Cramer's V shown in Figure \ref{fig-crv4}(c) nicely and more clearly recovers the item block structure.}
We conjecture that many real world datasets in other applied domains exhibit similar grouped dependence.

\begin{figure}[h!]
    \centering
    \begin{subfigure}[b]{0.36\textwidth}
    \includegraphics[width=\textwidth]{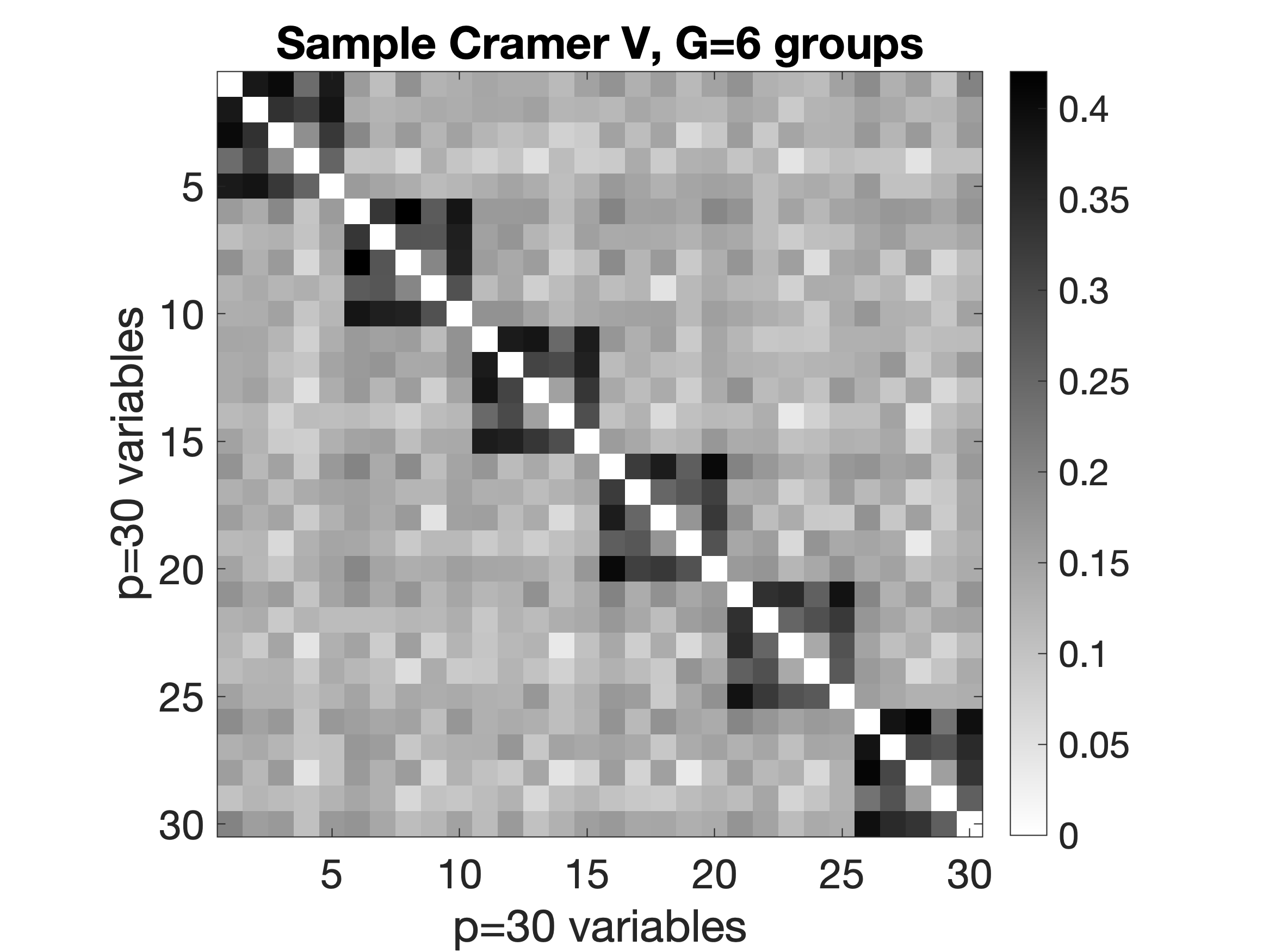}
    \caption{Simulated data.}
    \end{subfigure}
    \begin{subfigure}[b]{0.29\textwidth}
    \includegraphics[width=\textwidth]{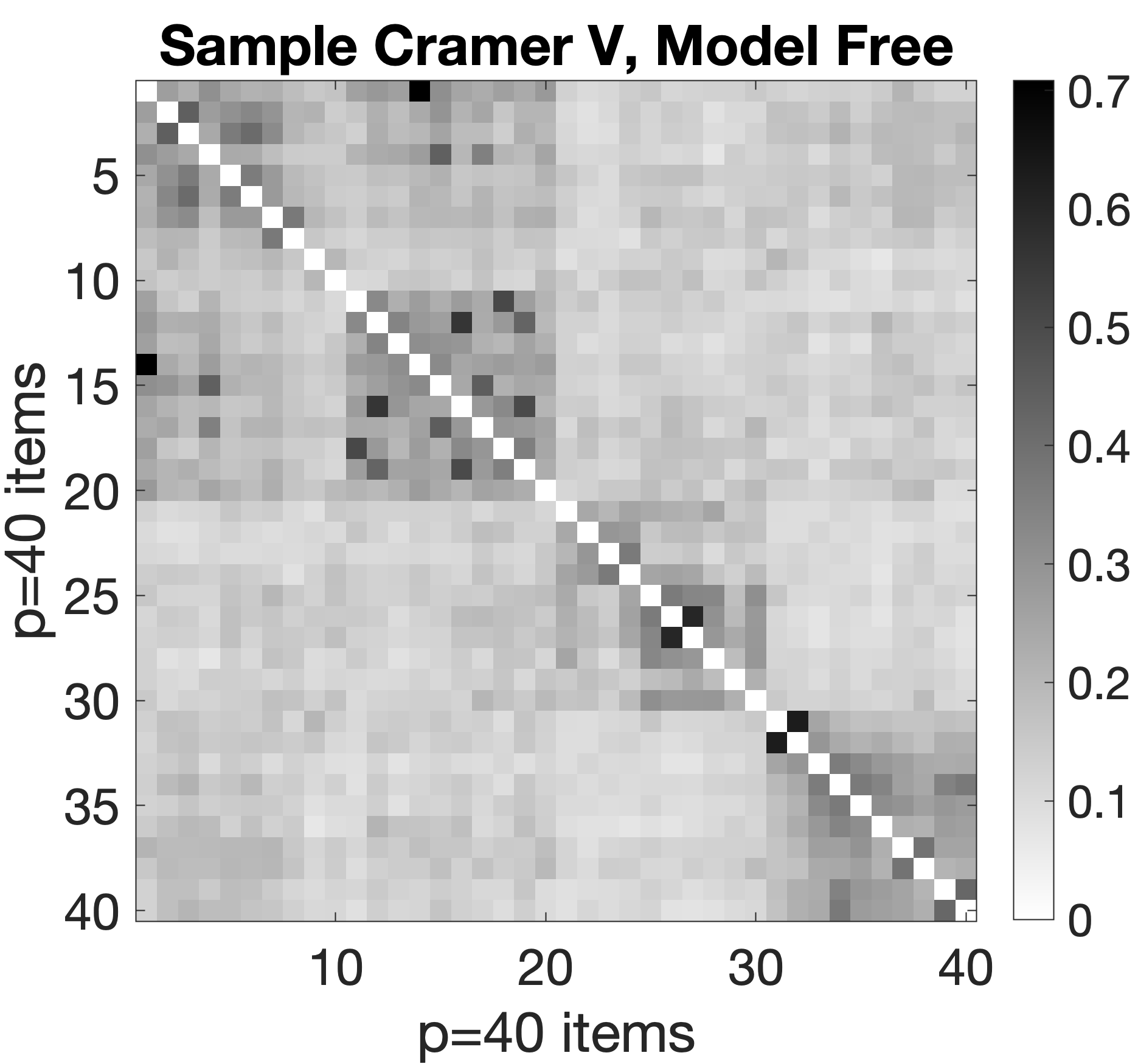}
    \caption{IPIP data, sample CRV.}
    \end{subfigure}
    \quad
    \begin{subfigure}[b]{0.29\textwidth}
    \includegraphics[width=\textwidth]{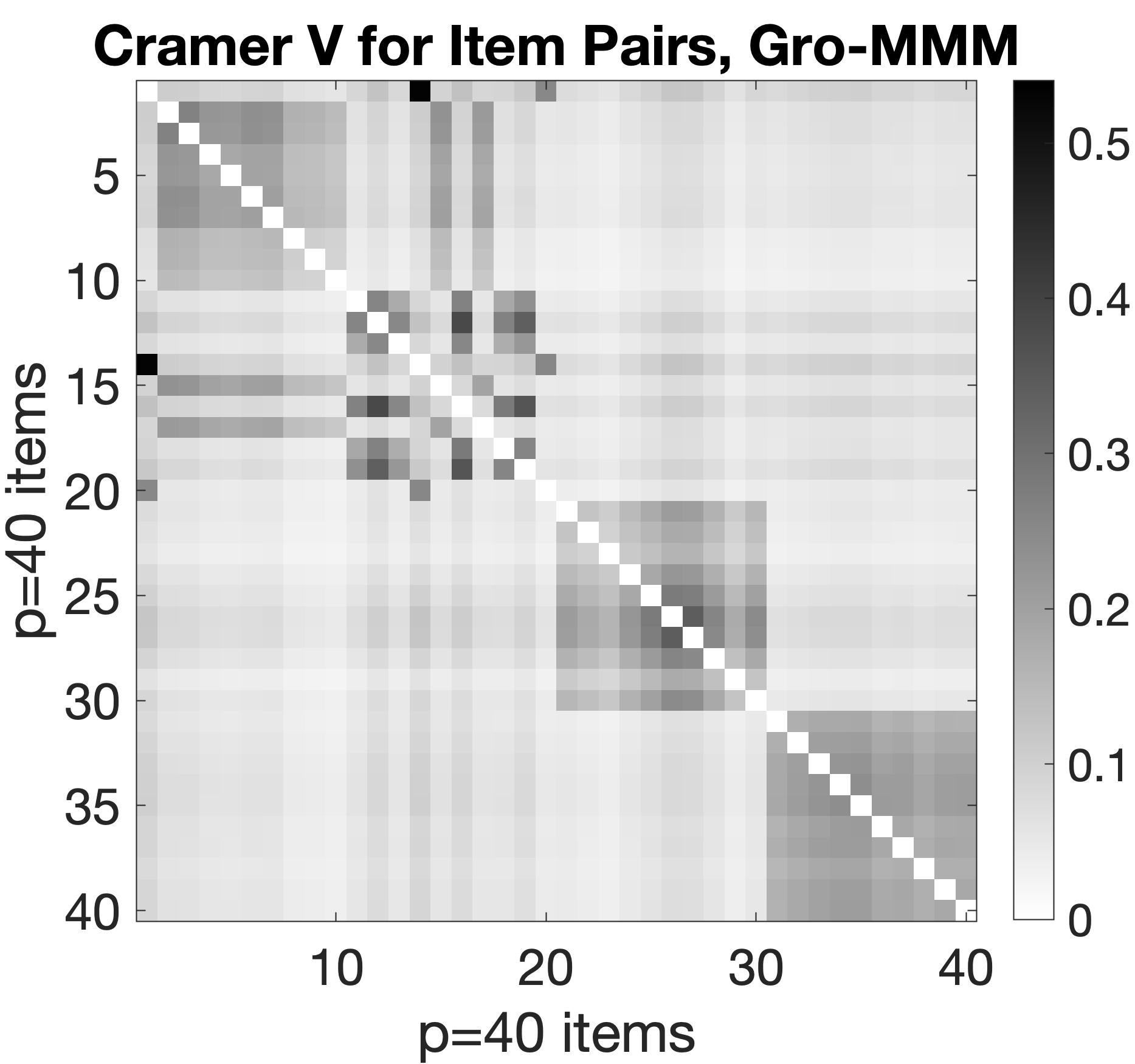}
    \caption{IPIP data, Gro-M$^3$ CRV.}
    \end{subfigure}
    
    \caption{\darkblue{(a): Sample Cramer's V (abbreviated as CRV) for a simulated dataset. (b): Sample Cramer's V for the IPIP data. (c) Gro-M$^3$ based Cramer's V for the IPIP data.}}
    \label{fig-crv4}
\end{figure}

\subsection{Probabilistic Tensor Decomposition Perspective}\label{sec-tensor}
The Gro-M$^3$ class has interesting connections to popular tensor decompositions.
For a subject $i$, the observed vector $\bo y_i$ resides in a contingency table with $\prod_{j=1}^p d_j$ cells. Since the MMMs for multivariate categorical data (both traditional GoM and the newly proposed Gro-M$^3$) induce a probability of $\bo y_i$ being in each of these cells, such probabilities $\{p\left(y_{i,1}=c_1,\ldots,y_{i,p}=c_p\mid -\right);\, c_j\in[d_j]~\text{for each}~j\in[p]\}$ can be arranged as a $p$-way $d_1 \times d_2 \times \cdots \times d_p$ array. 
This array is a tensor with $p$ modes and we denote it by $\PP$; \cite{koldabader2009} provided a review of tensors.
The tensor $\PP$ has all the entries nonnegative and they sum up to one, so we call it a \textit{probability tensor}. 
We next describe in detail the tensor decomposition perspective of our model; such a perspective will turn out to be useful in the study of identifiability. 

The probability mass function of $\bo y_i$ under the traditional GoM model can be written as follows by exchanging the order of product and summation,
\begin{align}\notag
	&~ p^{\text{GoM}}\left(y_{i,1}=c_1,\ldots,y_{i,p}=c_p\mid \bo\Lambda, \aaa\right)
    = \int_{\Delta^{K-1}}
	\prod_{j=1}^p \left[ \sum_{k=1}^K \pi_{i, k}  \lambda_{j, c_j, k}
    \right] dD_{\aaa}(\bo \pi_i)
    \\ \label{eq-gomtensor}
	= &~
    \sum_{k_1=1}^K \cdots \sum_{k_p=1}^K \prod_{j=1}^p
	\lambda_{j, c_j, k_j}
   \underbrace{\int_{\Delta^{K-1}}\pi_{i, k_1} \cdots \pi_{i, k_p} dD_{\aaa}(\bo\pi_i)}_{=:\; \phi^{\gom}_{k_1,\ldots, k_p}}.
\end{align}
Then $\mathbf{\Phi}^{\gom}:=\left(\phi^{\gom}_{k_1,\ldots, k_p};\; k_j\in[K]\right)$ forms a tensor with $p$ modes, and each mode has dimension $K$. Further, this tensor $\mathbf{\Phi}$ is a probability tensor, because $\phi_{k_1,\ldots, k_p} \geq 0$ and it is not hard to see that its entries sum up to one. Viewed from a tensor decomposition perspective, this is the popular Tucker decomposition \citep{tucker1966}; more specifically this is the nonnegative and probabilistic version of the Tucker decomposition.
The $\mathbf{\Phi}^{\gom}$ represents the Tucker tensor core, and the product of  $\{\lambda_{j, c_j, k}\}$ form the Tucker tensor arms.

It is useful to compare  our modeling assumption to that of the the Latent Class Model \citep[LCM;][]{goodman1974}, which follows the graphical model shown in Figure \ref{fig-graph}(a).
The LCM is essentially a finite mixture model assuming each subject $i$ belongs to a single cluster. 
The distribution of $\bo y_i$ under an LCM is
\begin{align}\label{eq-lcm}
    p^{\text{LC}}\left(y_{i,1}=c_1,\ldots,y_{i,p}=c_p\mid \bo\Lambda, \bo\nu\right)
    &=
    \sum_{k=1}^K \mathbb P(z_i = k) \prod_{j=1}^p \mathbb P(y_{i,j} \mid z_i = k)
    =
    \sum_{k=1}^K \nu_k \prod_{j=1}^p \lambda_{j, c_j, k}.
\end{align}
Based on the above definition, each subject $i$ has a single variable $z_i\in[K]$ indicating which latent class it belongs to, rather than a mixed membership proportion vector $\bo\pi_i$.
Denoting $\bo\nu^{\text{LC}} = (\nu_k;\, k\in[K])$, then \eqref{eq-lcm} corresponds to the popular CP decomposition of tensors \citep{hitchcock1927}, where the CP rank is at most $K$.

Finally, consider our proposed Gro-M$^3$,
\begin{align}\notag
	&~ p^{\text{Gro-M$^3$}}\left(y_{i,1},\ldots,y_{i,p}\mid \LL, \bo\Lambda, \aaa\right)
	= 
    \int_{\Delta^{K-1}} \prod_{g=1}^G \left[\sum_{k=1}^K \pi_{i,k} \prod_{j:\, \ell_{j,g} = 1} f(y_{i,j}\mid \lambda_{j,c_j,k}) \right] dD_{\aaa}(\bo \pi_i)
	\\ \label{eq-ctucker}
	= &~
	\sum_{k_1=1}^K \cdots \sum_{k_G=1}^K 
	\prod_{g=1}^G \prod_{j:\, \ell_{j,g}=1}
	f(y_{i,j}\mid \lambda_{j,c_j,k_g})
   \underbrace{\int_{\Delta^{K-1}}\pi_{i, k_1} \cdots \pi_{i, k_G} dD_{\aaa}(\bo \pi_i)}_{=:\; \phi^{\text{Gro-M$^3$}}_{k_1,\ldots, k_G}},
\end{align}
{where $f(y_{i,j} |\lambda_{j,c_j,k})$ generally denotes the conditional distribution of variable $y_{i,j}$ given parameter $\lambda_{j,c_j,k}$. In our Gro-M$^3$, $\lambda_{j,c_j,k}$ specifically refer to the  categorical distribution parameters for the random variable $y_{i,j}$; that is, $\lambda_{j,c_j,k} = \mathbb P(y_{i,j}=c_j\mid z_{i,j}=k)$ denotes the probability of responding $c_j$ to item $j$ given that the subject's realization of the latent profile for item $j$ is the $k$th extreme latent profile.}
In this case,  $\mathbf{\Phi}^{\text{Gro-M$^3$}}:=\left(\phi^{\text{Gro-M$^3$}}_{k_1,\ldots, k_G};\; k_g\in[K]\right)$ forms a tensor with $G$ modes, and each mode has dimension $K$. There still is $\sum_{k_1=1}^K \cdots \sum_{k_G=1}^K \allowbreak  \phi_{k_1,\ldots, k_G}^{\text{Gro-M$^3$}} \allowbreak = 1$.
This reduces the size of the core tensor in the classical Tucker decomposition because $G<p$. The 
Gro-M$^3$ incorporates aspects of both the CP and Tucker decompositions, providing a 
\textit{probabilistic hybrid decomposition} of probability tensors.  The CP is obtained when all variables are in the same group, while the Tucker is obtained when each variable is in its own group; see Figure \ref{fig-graph} for a clear illustration of this fact.

Gro-M$^3$ is conceptually related to the collapsed Tucker decomposition (c-Tucker) of 
\cite{johndrow2017}, though they did not model mixed memberships, used a very different model for the core tensor $\bo\Phi$, and did not consider identifiability.
Nonetheless and interestingly, our identifiability results can be applied to establish identifiability of c-Tucker decomposition (see Remark \ref{rmk-ctucker} in Section \ref{sec-bayes}).
Another work related to our dimension-grouping assumption is \cite{anandkumar2015overcomp}, which considered the case of overcomplete topic modeling with the number of topics exceeding the vocabulary size. 
For such scenarios, the authors proposed a ``persistent topic model'' which assumes the latent topic assignment persists locally through multiple words, and established identifiability. Our dimension-grouped mixed membership assumption is similar in spirit to this topic persistence assumption. However, the setting we consider here for general multivariate categorical data has the multi-characteristic single-replication nature ($p>1$ and $R=1$); as mentioned before, this is fundamentally different from topic models with a single characteristic and multiple replications ($p=1$ and $R>1$).


\section{Identifiability of Dimension-Grouped MMMs}\label{sec-id}

Identifiability is an important property of a statistical model, generally meaning that model parameters can be uniquely recovered from the observables. 
Identifiability serves as a fundamental prerequisite for valid statistical estimation and inference.
The study of identifiability, however, can be challenging for complicated models and especially latent variable models, including the Gro-M$^3$s considered here. 
In subsections \ref{sec-strid} and \ref{sec-genid}, we propose easily checkable and practically useful identifiability conditions for Gro-M$^3$s by carefully inspecting the inherent algebraic structures.
Specifically, we will exploit the variable groupings to write the model as certain highly structured mixed tensor products, and then prove identifiability by invoking Kruskal's theorem on the uniqueness of tensor decompositions \citep{kruskal1977three}. 
We point out that such proof procedures share a similar spirit to those in \cite{allman2009}, but the complicated  structure of the new Gro-M$^3$s requires some special care to handle. 
{We provide a high-level summary of our proof approach. First, we write the probability mass function of the observed $p$-dimensional multivariate categorical vector as a probabilistic tensor with $p$ modes. Second, we unfold this tensor into a $G$-way tensor with each mode corresponding to a variable group. Third, we further concatenate the transformed tensor and leverage Kruskal's Theorem on the uniqueness of three-way tensor decomposition to establish the identifiability of the model parameters under our proposed Gro-M$^3$.}
Our theoretical developments provide a solid foundation for performing estimation of the latent quantities and drawing valid conclusions from data.


\subsection{Strict Identifiability Conditions}\label{sec-strid}

For LDA and closely related topic models, there is a rich literature
 investigating identifiability under different assumptions \citep{anandkumar2012lda, arora2012, nguyen2015, wang2019ejs}.
Typically, when there is only one characteristic ($p=1$), $R \geq 2$ is necessary for identifiability; see Example 2 in \cite{wang2019ejs}.
However, there has been limited consideration of identifiability of  mixed membership models with multiple characteristics and one replication, i.e., $p>1$ and $R=1$. 
GoM models belong to this category, as does the proposed Gro-M$^3$s, with GoM being a special case of Gro-M$^3$s.

We consider the general setup in \eqref{eq-mmm}, where $\bo\Phi$ denotes the $G$-mode tensor core induced by any distribution $D(\bo\pi_i)$ on the probability simplex $\Delta^{K-1}$.
The following definition formally defines the concept of strict identifiability for the proposed model.

\begin{definition}[Strict Identifiability of Gro-M$^3$s]\label{def-strid}
    A parameter space $\bo\Theta$ of a Gro-M$\,^3$ is said to be strictly identifiable, if for any valid set of parameters $(\LL, \bo\Lambda,\bo\Phi) \in \bo\Theta$, the following equations hold if and only if $(\LL, \bo\Lambda,\bo\Phi)$ and the alternative $(\overline\LL, \overline{\bo\Lambda}, \overline{\bo\Phi})$ are identical up to permutations of the $K$ extreme latent profiles and permutations of the $G$ variable groups,
	\begin{align}\label{eq-id}
		\mathbb P(\bo y = \bo c\mid \LL, \bo\Lambda,\bo\Phi)
		=
		\mathbb P(\bo y = \bo c\mid \overline\LL, \overline{\bo\Lambda}, \overline{\bo\Phi}),
		\quad
		\forall \bo c\in \times_{j=1}^p [d_j].
	\end{align}
\end{definition}

Definition \ref{def-strid} gives the strongest possible notion of identifiability of the considered population quantities $(\LL, \bo\Lambda,\bo\Phi)$ in the model. 
In particular, the strict identifiability notion in Definition \ref{def-strid} requires identification of \textit{both} the continuous parameters $\bo\Lambda$ and $\bo\Phi$, \textit{and} the discrete latent grouping structure of variables in $\LL$.
The following theorem proposes sufficient conditions for the strict identifiability of Gro-M$^3$s.

\begin{theorem}\label{thm-attr1}
  
Under a Gro-M$\,^3$, the following two identifiability conclusions hold.
\begin{itemize}
    \item[(a)] Suppose each column of $\LL$ contains at least three entries of ``1''s, and the corresponding conditional probability table $\bo\Lambda_{j} = (\lambda_{j,c_j,k})_{d_j\times K}$ for each of these three $j$ has  full column rank. Then the $\bo\Lambda$ and $\bo\Phi$ are strictly identifiable.

   \item[(b)] In addition to the conditions in (a), if $\bo\Lambda$ satisfies that for each $j\in[p]$, not all the column vectors of $\bo\Lambda_j$ are identical, then $\LL$ is also identifiable.
\end{itemize}

\end{theorem}

\begin{example}\label{exp-thm1}
Denote by $\II_G$ a $G\times G$ identity matrix.
Suppose $p=3G$ and the matrix $\LL$ takes the following form,
\begin{align}\label{eq-3l}
    \LL = 
( \II_G\ \II_G\ \II_G )^\top.    
\end{align}
Also suppose for each $j\in\{1,\ldots,3G\}$, the $\bo\Lambda_j$ of size $d_j\times K$ has full column rank $K$. 
Then the conditions in Theorem \ref{thm-attr1} hold, so $\bo\Lambda$, $\LL$ and $\bo\Phi$ are identifiable.
Theorem \ref{thm-attr1} implies that if $\LL$ contains any additional row vectors other than those in \eqref{eq-3l} the model is still identifiable.
\end{example}

Theorem \ref{thm-attr1} requires that each of the $G$ variable groups contains at least three variables, and that for each of these $3G$ variables, the corresponding conditional probability table $\bo\Lambda_{j}$ has linearly independent columns.
%
Theorem  \ref{thm-attr1} guarantees not only the continuous parameters are identifiable, but also the discrete variable grouping structure summarized by $\LL$ is identifiable. This is important practically as typically the appropriate variable grouping structure is unknown, and hence needs to be inferred from the data.

The conditions in Theorem \ref{thm-attr1} essentially requires at least $3G$ conditional probability tables, each being a matrix of size $d_j \times K$, to have full column rank. This implicitly requires $d_j\geq K$.
\cite{tan2017partitioned} proposed a moment-based estimation approach for  traditional mixed membership models and briefly discussed the identifiability issue, also assuming $d_j\geq K$  with some full-rank requirements.
However, the cases where the number of categories $d_j$'s are small but the number of extreme latent profiles $K$ is much larger can arise in applications; for example, the disability survey data analyzed in \cite{erosheva2007aoas} and \cite{manrique2014aoas} have binary responses with $d_1=\cdots=d_p=2$ while the considered $K$ ranges from 2 to 10.
Our next theoretical result establishes weaker conditions for identifiability that accommodates $d_j< K$, by taking advantage of the dimension-grouping property of our proposed model class.

Before stating the theorem, we first introduce two useful notions of matrix products.
Denote by $\bigotimes$ the Kronecker product of matrices and by $\bigodot$ the Khatri-Rao product. 
Consider two matrices $\mathbf A=(a_{i,j})\in\mathbb R^{m\times r}$, $\mathbf B=(b_{i,j})\in\mathbb R^{s\times t}$; 
and another two matrices $\mathbf C=(c_{i,j})=(\bo c_{\bcolon,1}\mid\cdots\mid\bo c_{\bcolon,k})\in\mathbb R^{n\times k}$,
$\mathbf D=(d_{i,j})=(\bo d_{\bcolon,1}\mid\cdots\mid\bo d_{\bcolon,k})\in\mathbb R^{\ell\times k}$, then there are $\mathbf A\bigotimes \mathbf B \in\mathbb R^{ms\times rt}$ and $\mathbf C\bigodot \mathbf D \in\mathbb R^{n \ell\times k}$ with

\begin{align*}
	\mathbf A\bigotimes \mathbf B
	=
	\begin{pmatrix}
		a_{1,1}\mathbf B & \cdots & a_{1,r}\mathbf B\\
		\vdots & \vdots & \vdots \\
		a_{m,1}\mathbf B & \cdots & a_{m,r}\mathbf B
	\end{pmatrix},
	\qquad
	\mathbf C\bigodot \mathbf D
	=
	\begin{pmatrix}
		\bo c_{\bcolon,1}\bigotimes\bo d_{\bcolon,1}
		\mid \cdots \mid
		\bo c_{\bcolon,k}\bigotimes\bo d_{\bcolon,k}
	\end{pmatrix}.
\end{align*}
The above definitions show the Khatri-Rao product is the column-wise Kronecker product. The Khatri-Rao product of matrices plays an important role in the technical definition of the proposed dimension-grouped MMM.
The following Theorem \ref{thm-attr-finer} exploits the grouping structure in $\LL$ to relax the identifiability conditions in the previous Theorem \ref{thm-attr1}.

\begin{theorem}\label{thm-attr-finer}
    Denote by $\mca_g = \{j\in[p]:\, \ell_{j,g}=1\}$ the set of variables that belong to group $g$.
    Suppose each $\mca_g$ can be partitioned into three sets $\mca_g = \cup_{m=1}^3 \mca_{g,m}$, and for each $g\in[G]$ and $m\in\{1,2,3\}$ the matrix $\tilde{\bo\Lambda}_{g,m}$ defined below has full column rank $K$.
    \begin{align}\label{eq-lama}
    \tilde{\bo\Lambda}_{g,m} :=
        \bigodot_{j\in \mca_{g,m}} \bo\Lambda_j.
    \end{align}
    Also suppose for each $j\in[p]$, not all the column vectors of $\bo\Lambda_j$ are identical.
    Then the model parameters $\LL$, $\bo\Lambda$, and $\bo\Phi$ are strictly identifiable.
\end{theorem}

Compared to Theorem \ref{thm-attr1},
Theorem \ref{thm-attr-finer} relaxes the identifiability conditions by lifting the full-rank requirement on the individual matrices $\bo\Lambda_j$'s. Rather, as long as the Khatri-Rao product of several different $\bo\Lambda_j$'s have full column rank as specified in \eqref{eq-lama}, identifiability can be guaranteed.
Recall that the Khatri-Rao product of two matrices $\bo\Lambda_{j_1}$ of size $d_{j_1}\times K$ and $\bo\Lambda_{j_2}$ of size $d_{j_2}\times K$ has size  $(d_{j_1} d_{j_2}) \times K$. So intuitively, requiring $\bo\Lambda_{j_1} \bigodot \bo\Lambda_{j_2}$ to have full column rank $K$ is weaker than requiring each of $\bo\Lambda_{j_1}$ and $\bo\Lambda_{j_2}$ to have full column rank $K$.
The following Example \ref{exp-ab} formalizes this intuition.

\begin{example}\label{exp-ab}
Consider $d_1=d_2=2$, $K=3$ with the following conditional probability tables 
\begin{align*}
    \bo\Lambda_1=\begin{pmatrix}
        a_1 & a_2 & a_3\\
        1-a_1 & 1-a_2 & 1-a_3
    \end{pmatrix};
    \quad
    \bo\Lambda_2=\begin{pmatrix}
        b_1 & b_2 & b_3\\
        1-b_1 & 1-b_2 & 1-b_3
    \end{pmatrix}.
\end{align*}
Suppose variables $j=1,2$ belong to the same group, e.g., $\ell_{1,:} = \ell_{2,:}$.
Then since $K>d_1=d_2$, both  $\bo\Lambda_1$ and $\bo\Lambda_2$ can not have full column rank $K$.
However, if we consider their Khatri-Rao product, it has size $4\times 3$ in the following form

\begin{align*}
    \bo\Lambda_1 \bigodot \bo\Lambda_2
    =
    \begin{pmatrix}
        a_1 b_1 & a_2 b_2 & a_3 b_3 \\
        a_1(1-b_1) & a_2(1-b_2) & a_3 (1-b_3) \\
        (1-a_1) b_1 & (1-a_2) b_2 & (1-a_3) b_3 \\
        (1-a_1) (1-b_1) & (1-a_2) (1-b_2) & (1-a_3) (1-b_3) \\
    \end{pmatrix}.
\end{align*}
Indeed, $\bo\Lambda_1 \bigodot \bo\Lambda_2$ has full column rank for ``generic'' parameters $\bo\theta:=(a_1,a_2,a_3,b_1,b_2,b_3)$; precisely speaking, for $\bo\theta$ varying almost everywhere in the parameter space $[0,1]^6$ (the 6-dimensional hypercube), the subset of $\bo\theta$ that renders $\bo\Lambda_1 \bigodot \bo\Lambda_2$ rank-deficient has Lebesgue measure zero in $\mathbb R^6$. To see this, let $\bo x=(x_1, x_2, x_3)^\top\in\mathbb R^3$ such that $(\bo\Lambda_1 \bigodot \bo\Lambda_2) \bo x = 0$, then 
\begin{align*}
&\begin{cases}
        a_1 b_1 x_1 + a_2 b_2 x_2 + a_3 b_3 x_3 = 0;\\
        a_1(1-b_1) x_1 + a_2(1-b_2) x_2 + a_3(1-b_3) x_3 = 0;\\
        (1-a_1)b_1 x_1 + (1-a_2)b_2 x_2 + (1-a_3)b_3 x_3 = 0;\\
        (1-a_1) (1-b_1) x_1 + (1-a_2) (1-b_2) x_2 + (1-a_3) (1-b_3) x_3 = 0;\\
\end{cases}
\\[1mm]
\stackrel{\text{\normalfont{invertible transform}}}{\iff}\quad
&\begin{cases}
        a_1 b_1 x_1 + a_2 b_2 x_2 + a_3 b_3 x_3 = 0;\\
        a_1 x_1 + a_2 x_2 + a_3 x_3 = 0;\\
        b_1 x_1 + b_2 x_2 + b_3 x_3 = 0;\\
        x_1 + x_2 + x_3 = 0.\\
\end{cases}
\end{align*}
Based on the last four equations above, one can use basic algebra to obtain the following set of equations about $(x_1, x_2, x_3)$,
\begin{align*}
    \begin{pmatrix}
        b_1 - b_3 & b_3 - b_2 \\
        a_1 - a_3 & a_3 - a_2 \\
    \end{pmatrix} 
    \begin{pmatrix}
        x_1 \\
        x_2
    \end{pmatrix}
    =
    \begin{pmatrix}
        b_2 - b_1 & b_1 - b_3 \\
        a_2 - a_1 & a_1 - a_3 \\
    \end{pmatrix} 
    \begin{pmatrix}
        x_2 \\
        x_3
    \end{pmatrix}
    =
\begin{pmatrix}
    0\\
    0
\end{pmatrix}.
\end{align*}
This implies as long as the following inequalities hold, there must be $x_1=x_2=x_3=0$,
\begin{align}\label{eq-ab0}
\begin{cases}
    (b_1-b_3)(a_3-a_2) - (a_1-a_3)(b_3-b_2)\neq 0;
    \\
    (b_2-b_1)(a_1-a_3) - (a_2-a_1)(b_1-b_3)\neq 0
    \end{cases}
\end{align}
Now note that the subset of the parameter space $\{(a_1,a_2,a_3,b_1,b_2,b_3)\in[0,1]^6:\, \text{Eq.~}\eqref{eq-ab0}\text{ holds}\}$ is a Lebesgue measure zero subset of $[0,1]^6$. This means for such ``generic'' parameters varying almost everywhere in the parameter space $[0,1]^6$, the $(\bo\Lambda_1 \bigodot \bo\Lambda_2)\bo x=\zero$ implies $\bo x = \zero$ which means $\bo\Lambda_1 \bigodot \bo\Lambda_2$ has full column rank $K=3$. 
\end{example}

Example \ref{exp-ab} shows that the Khatri-Rao product of two matrices seems to have full rank under fairly mild conditions.
This indicates that the conditions in Theorem \ref{thm-attr-finer} are much weaker than those in Theorem \ref{thm-attr1} by imposing the full-rankness requirement only on a certain Khatri-Rao product of the $\bo\Lambda_j$-matrices, instead of on individual $\bo\Lambda_j$s.
To be more concrete, the next Example \ref{exp-thm2} illustrates Theorem \ref{thm-attr-finer}, as a counterpart of Example \ref{exp-thm1}.

\begin{example}\label{exp-thm2}
Consider the following grouping matrix $\LL$ with $G=3$ and $p=6G=18$,

\vspace{-3mm}\singlespacing
\begin{equation}
    \LL = \begin{pmatrix}
        \LL_1 \\
        \LL_1 \\
        \LL_1
    \end{pmatrix},\quad \text{where}~~
    \LL_1 = \begin{pmatrix}
        1 & 0 & 0 \\
        1 & 0 & 0 \\
        0 & 1 & 0 \\
        0 & 1 & 0 \\
        0 & 0 & 1 \\
        0 & 0 & 1 \\
    \end{pmatrix}.
\end{equation}
Then $\LL$ contains six copies of the identity matrix $\II_G$ after a row permutation. Thanks to greater variable grouping compared to the previous Example \ref{exp-thm1}, we can use Theorem \ref{thm-attr-finer} (instead of Theorem \ref{thm-attr1}) to establish identifiability. Specifically, consider binary responses with $d_1=\cdots=d_{18}=:d=2$ and $K=3$ extreme latent profiles. For $g=1$, define sets $\mca_{g,1} = \{1,2\}$, $\mca_{g,2} = \{7,8\}$, $\mca_{g,3} = \{13,14\}$; 
for $g=2$, define sets $\mca_{g,1} = \{3,4\}$, $\mca_{g,2} = \{5,6\}$, $\mca_{g,3} = \{7,8\}$; 
and for $g=3$, define sets $\mca_{g,1} = \{5,6\}$, $\mca_{g,2} = \{11,12\}$, $\mca_{g,3} = \{17,18\}$. 
Then for each $(g,m)\in\{1,\ldots,G\}\times\{1,2,3\}$, the $\tilde{\bo\Lambda}_{g,m} = \bigodot_{j\in\mca_{g,m}}\bo\Lambda_j$ defined in Theorem \ref{thm-attr-finer} has size $d^2\times K$ which is $4\times 3$, similar to the structure in Example \ref{exp-ab}.
Now from the derivation and discussion in Example \ref{exp-ab}, we know such a $\tilde{\bo\Lambda}_{g,m}$ has full rank for almost all the valid parameters in the parameter space. So the conditions in Theorem \ref{thm-attr-finer} are easily satisfied, and for almost all the valid parameters of such a Gro-M$^3$, the identifiability conclusion follows.
\end{example}

\subsection{Generic Identifiability Conditions}\label{sec-genid}

Example \ref{exp-ab} shows that the Khatri-Rao product of conditional probability tables easily has full column rank in a toy case, and Example \ref{exp-thm2} leverages this observation to establish identifiability for almost all parameters in the parameter space using Theorem \ref{thm-attr-finer}.
We next generalize this observation to derive more practical identifiability conditions, under the \textit{generic identifiability} notion introduced by \cite{allman2009}.
Generic identifiability generally means that the unidentifiable parameters belong to a set of Lebesgue measure zero with respect to the parameter space. 
Its definition adapted to the current Gro-M$^3$s is given as follows.

\begin{definition}[Generic Identifiability of Gro-M$^3$s]\label{def-genid}
    Under a Gro-M$\,^3$, a parameter space $\mathcal T$ for $(\bo\Lambda, \bo\Phi)$ is said to be generically identifiable, if there exists a subset $\mathcal N \subseteq \mathcal T$ that has Lebesgue measure zero with respect to $\mathcal T$ such that for any $(\bo\Lambda, \bo\Phi)\in \mathcal T\setminus \mathcal N$ and an associated $\LL$ matrix, the following holds if and only if $(\LL, \bo\Lambda, \bo\Phi)$ and the alternative $(\overline\LL, \overline{\bo\Lambda}, \overline{\bo\Phi})$ are identical up to permutations of the $K$ extreme latent profiles and that of the $G$ variable groups,
    \begin{align*}
		\mathbb P(\bo y = \bo c\mid \LL, \bo\Lambda,\bo\Phi)
		=
		\mathbb P(\bo y = \bo c\mid \overline\LL, \overline{\bo\Lambda}, \overline{\bo\Phi}),
		\quad
		\forall \bo c\in \times_{j=1}^p [d_j].
	\end{align*}
\end{definition}

Compared to the strict identifiability notion in Definition \ref{def-strid}, the generic identifiability notion in Definition \ref{def-genid} is less stringent in allowing the existence of a measure zero set of parameters where identifiability does not hold; see the previous Example \ref{exp-ab} for an instance of a measure-zero set. Such an identifiability notion usually suffices for real data analyses \citep{allman2009}.
In the following Theorem \ref{thm-attrgen}, we propose simple conditions to ensure generic identifiability of  Gro-M$^3$s.

\begin{theorem}\label{thm-attrgen}
For the notation $\mca_g=\{j\in[p]: \ell_{j,g}=1\}$ defined in Theorem \ref{thm-attr-finer}, suppose each $\mca_g$ can be partitioned into three non-overlapping sets $\mca_g = \cup_{m=1}^3 \mca_{g,m}$, such that for each $g$ and $m$ the following holds,
    \begin{align}\label{eq-djk}
        \prod_{j\in\mca_{g,m}} d_j \geq K.
    \end{align}
    Then the matrix $\bigodot_{j\in\mca_{g,m}} \bo\Lambda_j$ has full column rank $K$ for generic parameters.
    Further, the $\bo\Lambda$, $\LL$, and $\bo\Phi$ are generically identifiable.
\end{theorem}

Compared to Theorem \ref{thm-attr-finer},  Theorem \ref{thm-attrgen} lifts the explicit full-rank requirements on \textit{any} matrix. Rather, Theorem \ref{thm-attrgen} only requires that certain products of $d_j$'s should not be smaller than the number of extreme latent profiles, which in turn guarantees that the Khatri-Rao products of matrices have full column rank for generic parameters.
Intuitively, the more variables belonging to a group and the more categories each variable has, the easier the identifiability conditions are to satisfy. This illustrates the benefit of dimension-grouping to model identifiability.

\section{Dirichlet Gro-M$^3$ and Bayesian Inference}\label{sec-bayes}

\subsection{Dirichlet model and identifiability}
The previous section studies identifiability of general Gro-M$^3$s, not restricting the distribution $D_{\aaa}(\cdot)$ of the mixed membership scores to be a specific form. Next we focus on an interesting special case where $D_{\aaa}(\cdot)$ is a Dirichlet distribution with unknown parameters $\aaa$.
Among all the possible distributions for the individual mixed-membership proportions, the Dirichlet distribution is the most popular.
It is widely used in applications including social science survey data \citep{erosheva2007aoas,Wang2015AVE}, topic modeling \citep{blei2003lda, griffiths2004pnas}, and data privacy \citep{manrique2012jasa}.
We term the Gro-M$^3$ with $\bo\pi_i$ following a Dirichlet distribution the Dirichlet Gro-M$^3$, and propose a Bayesian inference procedure to estimate both the discrete variable groupings and the continuous parameters. Such a Dirichlet Gro-M$^3$ has the graphical model representation in Figure \ref{fig-graph}(d).

For an unknown vector $\aaa=(\alpha_1, \ldots, \alpha_K)$ with $\alpha_k>0$ for all $k\in[K]$,
suppose 
\begin{flalign}\label{eq-dir}
\text{Dirichlet Gro-M$^3$:}
&&
&\bo \pi_i = (\pi_{i,1},\ldots,\pi_{i,K}) \stackrel{\text{i.i.d.}}{\sim} \text{Dirichlet}(\alpha_1,\ldots,\alpha_K).&
\end{flalign}
The vector $\aaa$ characterizes the distribution of membership scores. As $\alpha_k \to 0$, the model simplifies to a latent class model in which each individual belongs to a single latent class.  For larger $\alpha_k$'s, there will tend to be multiple elements of $\bo \pi_i$ that are not close to $0$ or $1$.

Recall that the previous identifiability conclusions in Theorems \ref{thm-attr1}--\ref{thm-attrgen} generally apply to $\LL$, $\bo\Lambda$, and $\bo\Phi$, where $\bo\Phi$ is the \textit{core tensor} with $K^G$ entries in our hybrid tensor decomposition. 
Now with the core tensor $\bo\Phi$ parameterized by the Dirichlet distribution in particular, we can further investigate the identifiability of the Dirichlet parameters $\aaa$.
The following proposition establishes the identifiability of $\aaa$ for Dirichlet Gro-M$^3$s.

\begin{proposition}\label{prop-dir}
    Consider a Dirichlet Gro-M$\,^3$. If $G\geq 2$, then following conclusions hold.
    \begin{itemize}
        \item[(a)] If the conditions in Theorem \ref{thm-attr1} or Theorem \ref{thm-attr-finer} are satisfied, then the Dirichlet parameters $\aaa = (\alpha_1, \ldots, \alpha_K)$ are strictly identifiable.
        
        \item[(b)] If the conditions in Theorem \ref{thm-attrgen} are satisfied, then the Dirichlet parameters $\aaa = (\alpha_1, \ldots, \alpha_K)$ are generically identifiable.
    \end{itemize}
\end{proposition}

\begin{remark}\label{rmk-ctucker}
Our identifiability results 
have implications for the collapsed Tucker (c-Tucker) decomposition for multivariate categorical data \citep{johndrow2017}. 
Our assumption that the latent memberships underlying several variables are in one state is similar to that in c-Tucker. However, c-Tucker does not model mixed memberships, and the c-Tucker tensor core, $\bo\Phi$ in our notation, is assumed to arise from a CP decomposition \citep{goodman1974} with $\phi_{k_1,\ldots,k_G} = \sum_{v=1}^r w_{v} \prod_{g=1}^G \psi_{g,k_g,v}$. 
We can invoke the uniqueness of the CP decomposition \citep[e.g.,][]{kruskal1977three, allman2009} to obtain identifiability of parameters $\bo w=(w_v;\, v\in[r])$ and $\bo\psi=(\psi_{g,k,v};\, g\in[G],\, k\in[K], v\in[r])$.
Hence, under our assumptions on the variable grouping structure in Section \ref{sec-id}, imposing existing mild conditions on $\bo w$ and $\bo \psi$ will yield identifiability of all the  c-Tucker parameters.
\end{remark}

\subsection{Bayesian inference}
Considering the complexity of our 
latent structure model, we adopt a Bayesian approach.
We next describe the prior specification for $\LL$, $\bo\Lambda$, and $\aaa$ in Dirichlet Gro-M$^3$s.
The number of variable groups $G$ and number of extreme latent profiles $K$ are assumed known; we relax this assumption in Section \ref{sec-simu}.
Recall the indicators $s_1,\ldots,s_p\in[G]$ are defined as $s_j=g$ if and only if $\ell_{j,g}=1$, so there is a one-to-one correspondence between the matrix $\LL$ and the vector $\bo s=(s_1,\ldots,s_p)$.
We adopt the following prior for the $s_j$'s, 
\begin{align*}
    s_1,\ldots, s_p 
    &\stackrel{\text{i.i.d.}}{\sim} \text{Categorical}([G],\; \xi_1,\ldots,\xi_G),
\end{align*}
where $\text{Categorical}([G],\; \xi_1,\ldots,\xi_G)$ is a categorical distribution over $G$ categories with proportions $\xi_g\geq 0$ and $\sum_{g=1}^G \xi_g = 1$. We choose uniform priors over the probability simplex for $(\xi_1,\ldots,\xi_G)$ and each column  of $\bo\Lambda_j$. We remark that if certain prior knowledge about the variable groups is available for the data, then it is also possible to employ informative priors such as those in \cite{paganin2021centered} for the $s_j$'s.
For the Dirichlet parameters $\aaa$, defining $\alpha_0 = \sum_{k=1}^K \alpha_k$ and $\bo\eta = (\alpha_1/\alpha_0, \ldots, \alpha_K/\alpha_0)$, we choose 
the hyperpriors $\alpha_0\sim \text{Gamma}(a_\alpha, b_\alpha)$ and $\bo\eta$ is uniform over the $(K-1)$-probability simplex.

Given a sample of size $n$, denote the observed data by $\bm Y = \{\bo y_i;\, i=1,\ldots,n\}$.
We propose a Metropolis-Hastings-within-Gibbs sampler {and also a Gibbs sampler} for posterior inference of $\LL$, $\bo\Lambda$, and $\aaa$ based on the data $\bm Y$. 

\paragraph{{Metropolis-Hastings-within-Gibbs Sampler.}} This sampler cycles through the following steps.

\begin{description}
    \item[\textbf{Step 1--3.}] Sample each column of the conditional probability tables $\bo\Lambda_j$'s, the individual mixed-membership proportions $\bo\pi_i$'s, and the individual latent assignments $z_{i,g}$'s from their full conditional posterior distributions. Define indicator variables $y_{i,j,c} = \mathbb I(y_{i,j}=c)$ and $z_{i,g,k} = \mathbb I(z_{i,g}=k)$. These posteriors are
\begin{eqnarray*}
   \{\bo\lambda_{j,\bcolon, k}\mid -\}_{s_j=g}
    &\sim &
   \text{Dirichlet} \left(
    1 + \sum_{i=1}^n z_{i,g,k} y_{i,j,1},
     ~\ldots,~
    1 + \sum_{i=1}^n z_{i,g,k} y_{i,j,d_j}
   \right);
    \\[2mm]
   \bo\pi_{i}\mid - &\sim & \text{Dirichlet}\left( \alpha_1+\sum_{g=1}^G z_{i,g,1}, ~\ldots,~ \alpha_K + \sum_{g=1}^G z_{i,g,K} \right);
    \\[3mm]
   \mathbb P(z_{i,g}=k \mid -)
   &= &
   \frac{\pi_{i,k} \prod_{j: \, s_j=g} \prod_{c=1}^{d_j}  \lambda_{j,c, k}^{y_{i,j,c}}}{\sum_{k'=1}^K \pi_{i,k'} \prod_{j: \, s_j=g} \prod_{c=1}^{d_j}  \lambda_{j,c, k'}^{y_{i,j,c}}},
   \quad k\in[K].
\end{eqnarray*}

\item[\textbf{Step 4.}] Sample the variable grouping structure $(s_1,\ldots,s_p)$.
The posterior of each $s_j$ is
\begin{align*}
    \mathbb P(s_j = g\mid -)
    =
    \frac{\xi_{g} \prod_{i=1}^n \lambda_{j,y_{i,j},z_{i,g}}}{\sum_{g'=1}^G \xi_{g'} \prod_{i=1}^n \lambda_{j,y_{i,j},z_{i,g'}}}, \quad g\in[G].
\end{align*}
The posterior of $(\xi_1,\ldots,\xi_G)$ is 
\begin{align*}
    (\xi_1, \ldots, \xi_G)\mid -
    \sim
    \text{Dirichlet}\left(1 + \sum_{j=1}^p \mathbb I(s_j=1),\ldots, 1 + \sum_{j=1}^p \mathbb I(s_j=G)\right).
\end{align*}

\item[\textbf{Step 5.}]
Sample the Dirichlet parameters $\aaa=(\alpha_1, \ldots, \alpha_K)$ via Metropolis-Hastings sampling. 
The conditional posterior distribution of $\aaa$ (or equivalently, $\alpha_0$ and $\bo\eta$) is
\begin{align*}
    p(\aaa\mid -)
    &\propto
    \text{Gamma}(\alpha_0\mid a,b)\times 
    \text{Dirichlet}(\bo\eta\mid \one_K)
    \times \prod_{i=1}^n \text{Dirichlet}(\bo\pi_i\mid\aaa) \\
    &\propto \alpha_0^{a_\alpha-1}\exp(-b_\alpha \alpha_0)\times 
    \left[\frac{\Gamma(\alpha_0)}{\prod_{k=1}^K \Gamma(\alpha_k)}\right]^n 
    \times \prod_{k=1}^K \left[\prod_{i=1}^n \pi_{i,k}\right]^{\alpha_k},
\end{align*}
which is not an easy-to-sample-from distribution. 
We use a Metropolis-Hastings sampling strategy in \cite{manrique2012jasa}. The steps are detailed as follows.
\begin{itemize}
    \item Sample each entry of $\aaa^\star=(\alpha_1^\star,\ldots,\alpha_K^\star)$ from independent lognormal distributions (proposal distribution $g(\aaa^\star\mid\aaa)$) as
    \begin{align}\label{eq-sigalpha}
        \alpha_k^\star \stackrel{\text{ind.}}{\sim} \text{lognormal}(\log \alpha_k, \sigma_\alpha^2), 
    \end{align}
    where $\sigma_\alpha$ is a tuning parameter that affects the acceptance ratio of the draw. 
    Based on our preliminary simulations, $\sigma$ should be relatively small to avoid the acceptance ratio to be always too close to zero.
    
    \item Let $\alpha_0^\star = \sum_{k=1}^K \alpha_k^\star$. Define
    \begin{align*}
        r^\star 
        =&~ \frac{p(\aaa^\star\mid-) g(\aaa\mid\aaa^\star)}{p(\aaa\mid-) g(\aaa^\star\mid\aaa)} \\
        =&~ \left(\frac{\alpha_0^\star}{\alpha_0}\right)^{a_\alpha-1} \exp\left( -b_\alpha(\alpha_0^\star - \alpha_0) \right) \times \left(\frac{\Gamma(\alpha_0^*)}{\Gamma(\alpha_0)} \cdot \frac{\prod_{k=1}^K \Gamma(\alpha_k)}{\prod_{k=1}^K \Gamma(\alpha_k^\star)}\right)^n
        \\
        &~ \times \prod_{k=1}^K\left(\prod_{i=1}^n \pi_{i,k}\right)^{\alpha_k^\star - \alpha_k} \times \prod_{k=1}^K  \frac{\alpha_k^\star}{\alpha_k}
    \end{align*}
    The Metropolis-Hastings acceptance ratio of the proposed $\aaa^\star$ is $r = \min\left\{1, \; r^\star\right\}$.
\end{itemize}
\end{description}
{We track the acceptance ratio in the Metropolis-Hastings step along the MCMC iterations in a simulation study. Figure \ref{fig-accratio} shows the boxplots of the average acceptance ratios for various sample sizes in the same simulation as the later Table \ref{tab-acc-K4}. This figure shows that the Metropolis-Hastings acceptance ratio is generally high and mostly exceeds 80\%.}

\begin{figure}[h!]
    \centering
    \includegraphics[width=0.6\textwidth]{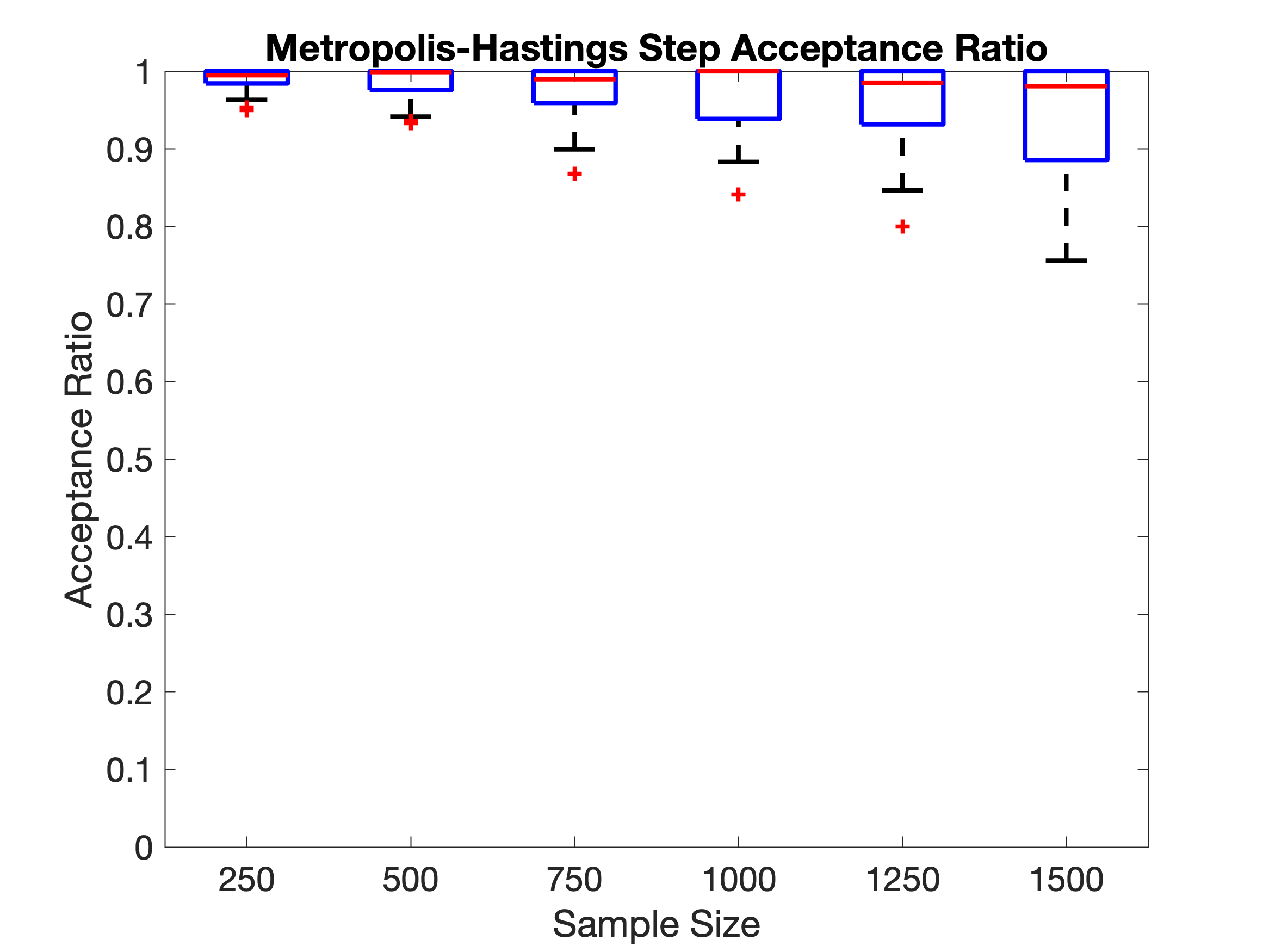}
    \caption{{Metropolis-Hastings average acceptance ratio in the simulation setting $(p,G,K) = (30, 6, 4)$, corresponding to the first setting in Table \ref{tab-acc-K4} in the manuscript.}}
    \label{fig-accratio}
\end{figure}

\medskip
\paragraph{Gibbs Sampler.}
We also develop a fully Gibbs sampling algorithm for our Gro-M$^3$, leveraging the auxiliary variable method in \cite{zhou2018nbfa} to sample the Dirichlet parameters $\aaa$.
Especially, since we have proved in Proposition \ref{prop-dir} that the entire Dirichlet parameter vector $\aaa = (\alpha_1,\ldots,\alpha_K)$ is identifiable from the observed data distribution, we choose to freely and separately sample all the entries $\alpha_1, \ldots, \alpha_K$  instead of constraining these $K$ entries to be equal as in \cite{zhou2018nbfa}. 
Recall that for each subject $i$,  $z_{i,g} \in[K]$ for $g\in[G]$ denotes the latent profile realization for the $g$th group of items.
Introduce new notation $Z^{\text{mult}}_{ik} = \sum_{g=1}^G \mathbbm{1}(z_{i,g}=k)$ for $i\in[N]$ and  $k\in[K]$. Then $(Z^{\text{mult}}_{i1}, \ldots, Z^{\text{mult}}_{i1})$ follows the Dirichlet-Multinomial distribution with parameters $G$ and $(\alpha_1,\ldots,\alpha_K)$.
We introduce auxiliary Beta variables $q_i$ for $i\in[N]$ and auxiliary Chinese Restaurant Table (CRT) variables $t_{ik}$ for $i\in[N]$ and  $k\in[K]$. 
Endowing the Dirichlet parameter $\alpha_k$ with the prior $\alpha_k \sim \text{Gamma}(a_0, ~b_0)$, we have the following Gibbs updates for sampling $\alpha_k$.
\begin{description}
\item[Step 5$^\star$] Sample the auxiliary variables $q_i$, $t_{ik}$ and the Dirichlet parameters $\alpha_k$ from the following full conditional posteriors:
\begin{align*}
    q_i &\sim \text{Beta}\left( \sum_{k=1}^K Z^{\text{mult}}_{ik},~  \sum_{k=1}^K \alpha_k\right),\qquad i\in[n];\\[2mm]
    t_{ik} &\sim \text{CRT}(Z^{\text{mult}}_{ik},~ \alpha_k),\qquad i\in[n],~ k\in[K];\\[2mm]
    \alpha_k &\sim \text{Gamma} \left( a_0 + \sum_{i=1}^n t_{ik},~ b_0 - \sum_{i=1}^n \log(1-q_i) \right), \qquad k\in[K].
\end{align*}
\end{description}
Replacing the previous Step 5 in the Metropolis-within-Gibbs sampler with the above Step 5$^\star$ gives a fully Gibbs sampling algorithm for Gro-M$^3$.

Our simulations reveal the following empirical comparisons between the Gibbs sampler and the Metropolis-Hastings-within-Gibbs (MH-within-Gibbs) sampler.
In terms of Markov chain mixing, the Gibbs sampler mixes faster than the MH-within-Gibbs sampler as expected, and requires fewer MCMC iterations to generate quality posterior samples \emph{if initialized well}.
However, in terms of estimation accuracy, we observe that the MH-within-Gibbs sampler tends to have better accuracy in estimating the identifiable model parameters. This is likely because that the MH-within-Gibbs sampler performs better on exploring the entire posterior space through the proposal distributions; whereas the Gibbs sampler tends to be more heavily influenced by the initial value of the parameters and can converge to suboptimal distributions if not initialized well. We next provide the experimental evidence behind the above observations.

Figure \ref{fig-trace} provides typical traceplots for the MH-within-Gibbs sampler (left) and the Gibbs sampler (middle and right) in one simulation trial in the same setting as the later Table \ref{tab-acc-K4}. The four horizontal lines in each panel denote the true parameter values $\aaa=(\alpha_1,\alpha_2,\alpha_3,\alpha_4) = (0.4,0.5,0.6,0.7)$. The left and middle panels of Figure \ref{fig-trace} are traceplots of $\alpha_k$ in MCMC chains initialized randomly with the same initial value, whereas the right panel corresponds to a chain initialized with the true parameter value $\aaa$. Figure \ref{fig-trace} shows that when initialized randomly with the same value, the MH-within-Gibbs chain converges to distributions much closer to the truth than the Gibbs sampler; in contrast, the Gibbs chain only manages to converge to the desirable posteriors when initialized with the true $\aaa$.
Furthermore, Figure \ref{fig-2rmse} plots the root mean squared error quantitles (25\%, 50\%, 75\%) of $\aaa$ estimated using the two samplers from the 50 simulation replicates in each setting. The parameter initialization in each replicate for the two samplers is random and identical. Figure \ref{fig-2rmse} clearly shows that the MH-within-Gibbs sampler has lower estimation error for $\aaa$.
In summary, when \emph{initialized randomly using the same mechanism}, the MH-within-Gibbs sampler has higher parameter estimation accuracy despite that the Gibbs sampler mixes faster.
Therefore, we choose to present the estimation results of the MH-within-Gibbs sampler in the later Section \ref{sec-simu}.

\begin{figure}[h!]
    \centering
    \resizebox{\textwidth}{!}{%
    \includegraphics[height=3cm]{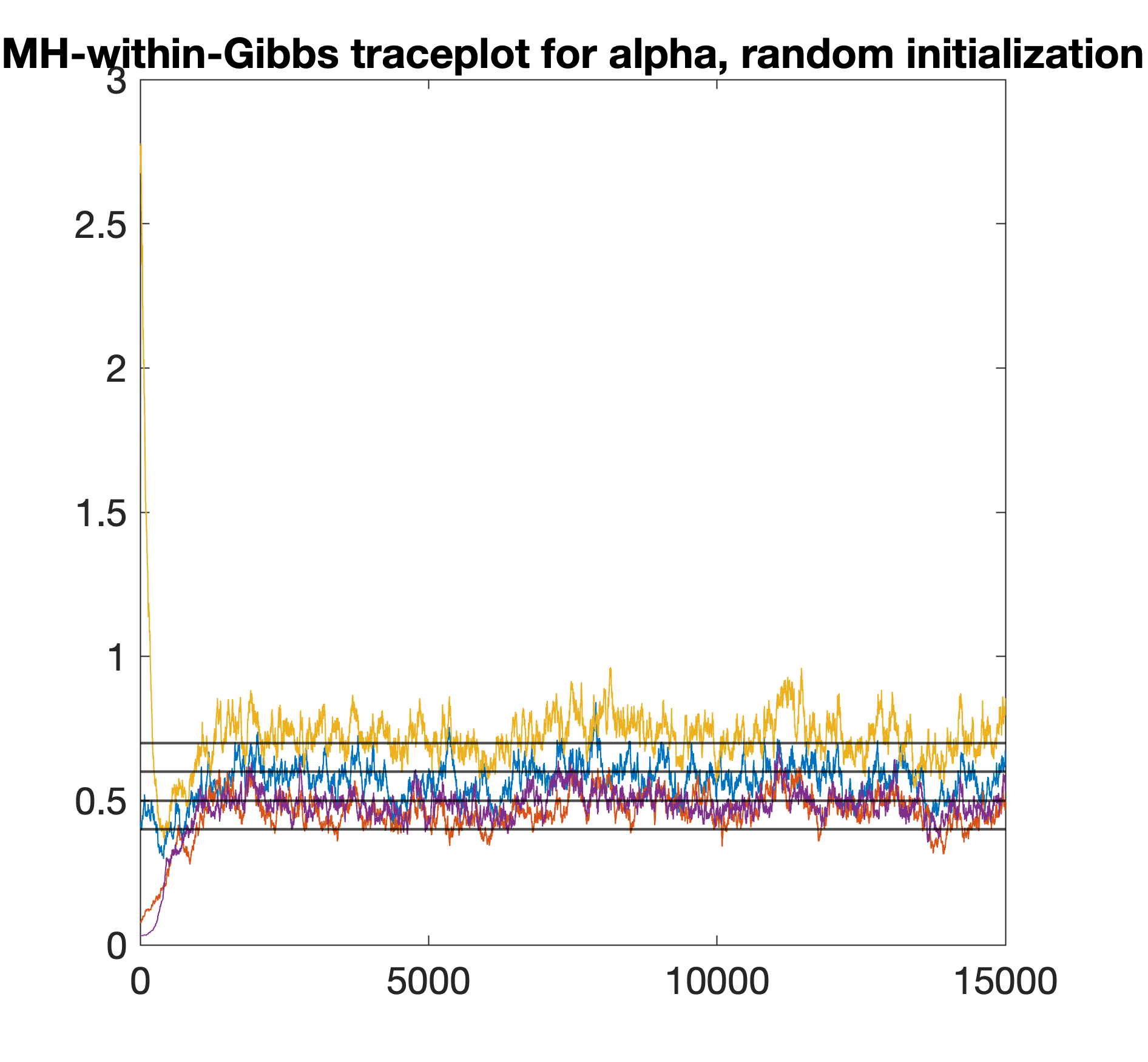}
    \includegraphics[height=3cm]{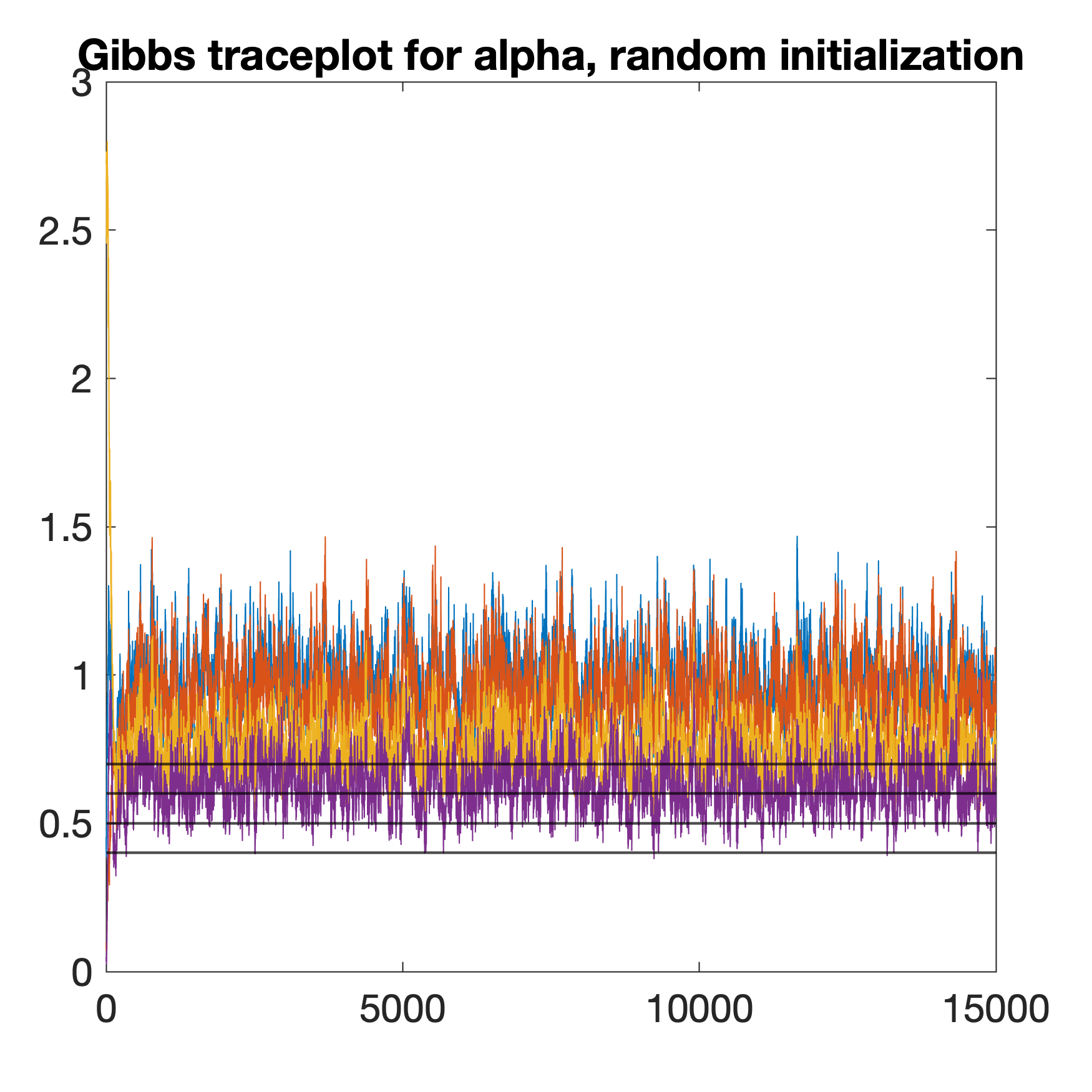}
    \includegraphics[height=3cm]{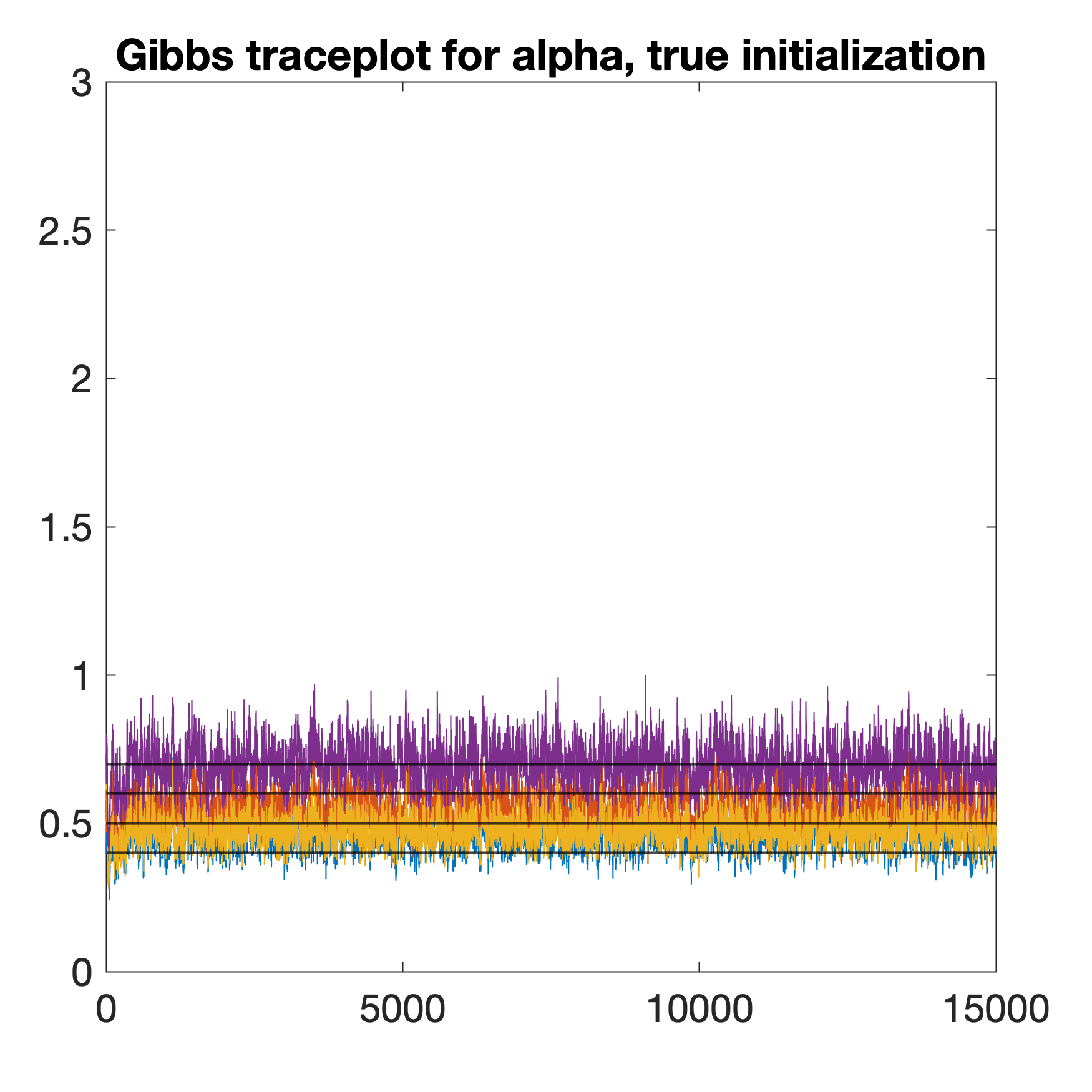}
    }
    \caption{{Traceplots of the MH-within-Gibbs sampler (left) and the Gibbs sampler (middle and right) applied to one simulated dataset with $(n,p,G,K)=(500,30,6,4)$.
    The horizontal lines in each panel denote the true $\aaa=(\alpha_1,\alpha_2,\alpha_3,\alpha_4) = (0.4,0.5,0.6,0.7)$. The left and middle panels correspond to chains initialized randomly with the same initial value, whereas the right panel corresponds to a chain initialized with the true parameter value $\aaa$.}}
    \label{fig-trace}
\end{figure}
\begin{figure}[h!]
    \centering
    \includegraphics[width=0.56\textwidth]{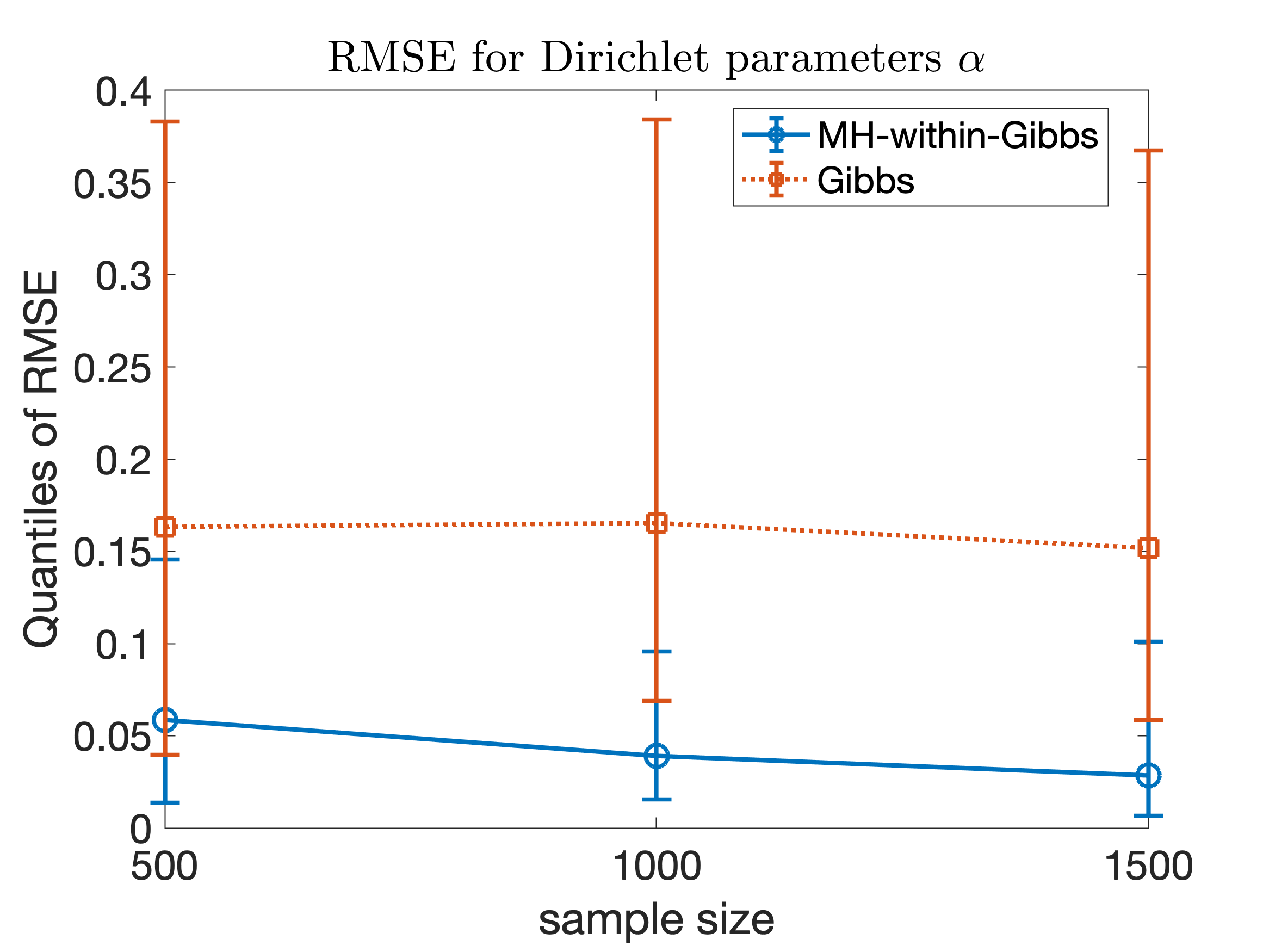}
    \caption{{Root mean squared errors (RMSE) quantiles (25\%, 50\%, 75\%) for the MH-within-Gibbs sampler and the Gibbs sampler obtained from 50  simulation replicates for each sample size. In each simulation replicate, the initializations of the Gibbs chain and the MH-within-Gibbs chain are identical.}}
    \label{fig-2rmse}
\end{figure}
\color{black}

After collecting posterior samples from the output of the MCMC algorithm, for those continuous parameters in the model we can calculate their posterior means  as point estimates. 
As for the discrete variable grouping structure, we can obtain the posterior modes of each $s_j$. That is, given the $T$ posterior samples of $\bo s^{(t)}=(s^{(t)}_1,\ldots,s^{(t)}_p)$ for $t=1,\ldots,T$, we define point estimates $\overline{\bo s}$ and $\overline\LL$ with entries
\begin{align}\label{eq-lmode}
    \overline s_j &= \argmax_{g\in [G]} \sum_{t=1}^T \mathbb I(s^{(t)}_j=g);
    \qquad
    \overline \ell_{j,g} =
    \begin{cases}
    1, & \text{if } \overline s_j = g;\\
    0, & \text{otherwise}.
    \end{cases}
\end{align}

\section{Simulation Studies}\label{sec-simu}
In this section, we carry out simulation studies to assess the performance of the proposed Bayesian estimation approach, while verifying that identifiable parameters are indeed estimated more accurately as sample size grows. 
In Section \ref{sec-simu1}, we perform a simulation study to assess the estimation accuracy of the model parameters, assuming the number of extreme latent profiles $K$ and the number of variable groups $G$ are known. This is the same assumption as in many existing estimation methods of traditional MMMs \citep[e.g.,][]{manrique2012jasa}. 
In Section \ref{sec-simu2}, to facilitate the use of our estimation method in applications, we propose data-driven criteria to select $K$ and $G$ and perform a corresponding simulation study.

\subsection{Estimation of Grouping Structure and Model Parameters}\label{sec-simu1}

In this simulation study, we assess the proposed algorithm's performance in estimating the $(\LL, \bo\Lambda, \aaa)$ in Dirichlet Gro-M$^3$s. 
We consider various simulation settings, with $K = 2,3$, or $4$, and $(p,G) = (30,6)$, $(60,12)$, or $(90,15)$. The number of categories of each $y_j$ is specified to be three, i.e., $d_1=\cdots=d_p=3$. The true $\LLambda$-parameters are specified as follows: in the most challenging case with $K=4$ and $(p,G)=(90,15)$, for $u=0,1,\ldots,p/6-1$ we specify

\vspace{-8mm}\singlespacing
    $$
    {\small\LLambda_{6u+1} = \begin{pmatrix}
     0.1 & 0.7 & 0.3 & 0.1 \\
     0.8 & 0.2 & 0.4 & 0.1 \\
     0.1 & 0.1 & 0.3 & 0.8
    \end{pmatrix};
    ~
    \LLambda_{6u+2} = \begin{pmatrix}
     0.1 & 0.8 & 0.1 & 0.2 \\
     0.2 & 0.1 & 0.6 & 0.5 \\
     0.7 & 0.1 & 0.3 & 0.3
    \end{pmatrix};
    ~
    \LLambda_{6u+3} = \begin{pmatrix}
     0.1 & 0.8 & 0.2 & 0.9 \\
     0.2 & 0.1 & 0.5 & 0.05 \\
     0.7 & 0.1 & 0.3 & 0.05
    \end{pmatrix};}
  $$
  $${\small
  \LLambda_{6u+4} = \begin{pmatrix}
     0.1 & 0.1 & 0.8 & 0.3 \\
     0.8 & 0.2 & 0.1 & 0.6 \\
     0.1 & 0.7 & 0.1 & 0.1
    \end{pmatrix};
    ~
    \LLambda_{6u+5} = \begin{pmatrix}
     0.2 & 0.7 & 0.3 & 0.1 \\
     0.6 & 0.2 & 0.4 & 0.1 \\
     0.2 & 0.1 & 0.3 & 0.8
    \end{pmatrix};
    ~
    \LLambda_{6u+6} = \begin{pmatrix}
     0.1 & 0.8 & 0.1 & 0.2 \\
     0.2 & 0.1 & 0.1 & 0.6 \\
     0.7 & 0.1 & 0.8 & 0.2
    \end{pmatrix}.
    }
    $$
As for other simulation settings with smaller $K$ and $(p,G)$, we specify the $\LLambda_j$'s by taking a subset of the above matrices and retaining a subset of columns from each of these matrices. The true Dirichlet parameters $\aaa$ are set to $(0.4,0.5)$ for $K=2$, $(0.4,0.5,0.6)$ for $K=3$, and $(0.4,0.5,0.6,0.7)$ for $K=4$. The true grouping matrix $\LL$ of size $p\times G$ is specified to containing $p/G$ copies of identity submatrices $\II_G$ up to a row permutation. Under these specifications, our identifiability conditions in Theorem \ref{thm-attr-finer} are satisfied.
We consider sample sizes $n=250, 500, 1000, 1500$.
In each scenario, 50 independent datasets are generated and fitted with the proposed MCMC algorithm described in Section \ref{sec-bayes}. In our MCMC algorithm under all simulation settings, we take hyperparameters to be $(a_\alpha, b_\alpha) = (2,1)$ and $\sigma_\alpha = 0.02$. The MCMC sampler is run for 15000 iterations, with the first 10000 iterations as burn-in and every fifth sample is collected after burn-in to thin the chain. 

We observed good mixing and convergence behaviors of the model parameters from examining the trace plots. In particular, simulations show that the estimation of the discrete variable grouping structure in matrix $\LL$ (equivalently, vector $\bo s$) is quite accurate in general, and the posterior means of the continuous $\LLambda$ and $\aaa$ are also close to their truth. 
Next we first present details of two typical simulation trials as an illustration, before presenting summaries across the independent simulation replicates. 

Two random simulation trials were taken from the settings $(n,p,G,K) = (500,30,6,2)$ and $(n,p,G,K) = (500,90,15,2)$. All the parameters were randomly initialized from their prior distributions. 
In Figure \ref{fig-simu-lmat}, the left three plots in each of the first two rows show the sampled $\LL_{\text{iter.}}$ in the MCMC algorithm, after the 1st, 201st, and 401st iterations, respectively; the fourth plot show the posterior mode $\overline\LL$ defined in \eqref{eq-lmode}, and the last plot shows the simulation truth $\LL$. If an $\tilde\LL$ equals the true $\LL$ after a column permutation then it indicates $\tilde\LL$ and $\LL$ induce identical variable groupings.
The bottom two plots in Figure \ref{fig-simu-lmat} show the Adjusted Rand Index \citep[ARI,][]{rand1971objective} of the variable groupings of $\LL_{\text{iter.}}$ ($\bo s_{\text{iter.}}$) with respect to the true $\LL$ (true $\bo s$) along the first 1000 MCMC iterations.
The ARI measures the similarity between two clusterings, and it is appropriate to compare a true $\bo s$ and an estimated $\overline{\bo s}$ because they each summarizes a clustering of the $p$ variables into $G$ groups.
The ARI is at most $1$, with $\text{ARI}=1$ indicating perfect agreement between two clusterings.
The bottom row of Figure \ref{fig-simu-lmat} shows that in each simulation trial, the ARI measure starts with values around 0 due to the random MCMC initialization, and within a few hundred iterations the ARI increases to a distribution over much larger values. 
For the simulation with $(n,p,G,K) = (500,90,15,2)$, the posterior mode of $\LL$ exactly equals the truth, and the corresponding plot on the bottom right of Figure \ref{fig-simu-lmat} shows the ARI is distributed very close to 1 after just about 500 MCMC iterations.
In general, our MCMC algorithm has excellent performance in inferring the $\LL$ from randomly initialized simulations; also see the later Tables \ref{tab-acc-K2}--\ref{tab-acc-K4} for more details.

\begin{figure}[h!]
    \centering
    \includegraphics[width=0.8\textwidth]{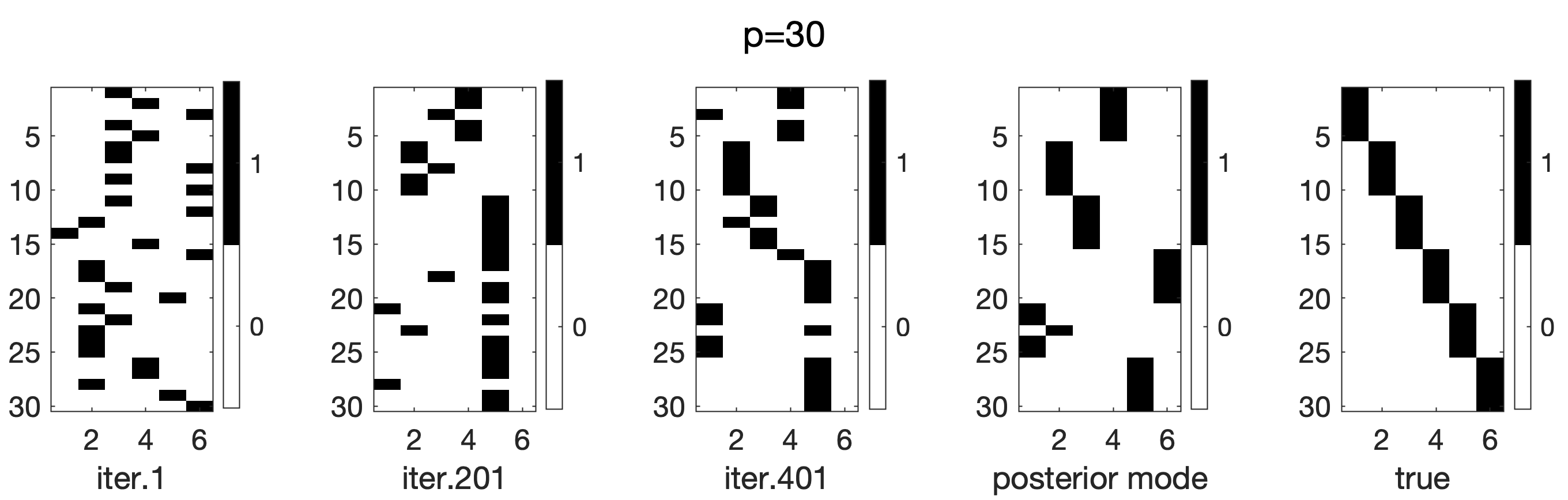}
    
    \medskip
    \includegraphics[width=0.8\textwidth]{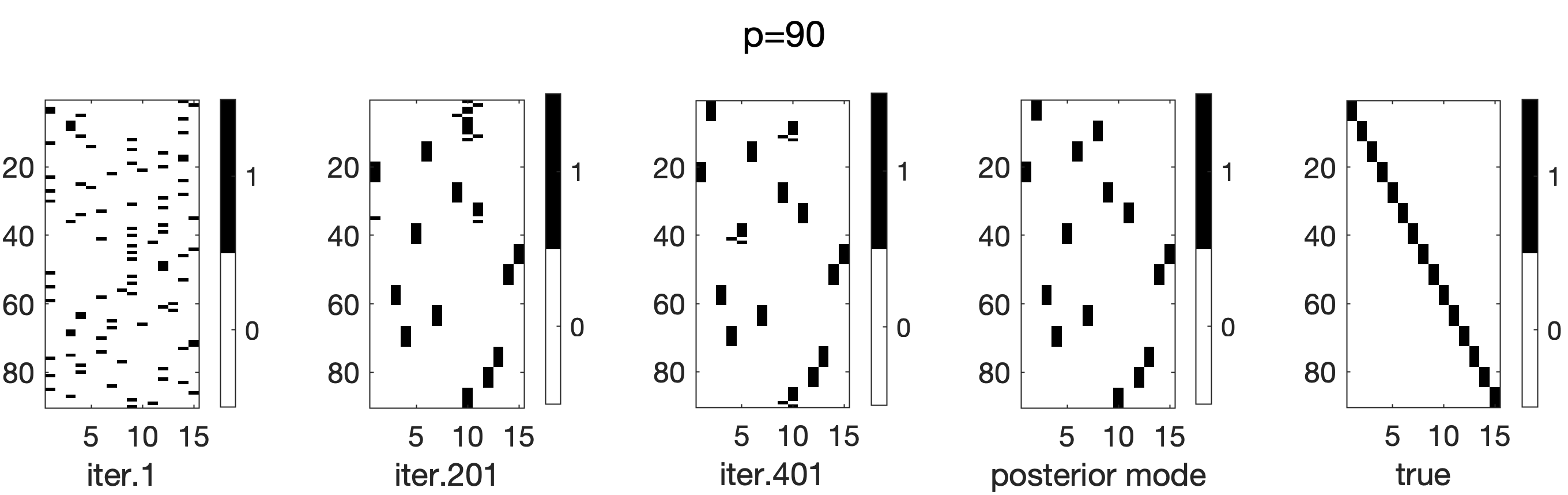}
    
    \medskip
    \includegraphics[width=0.48\textwidth]{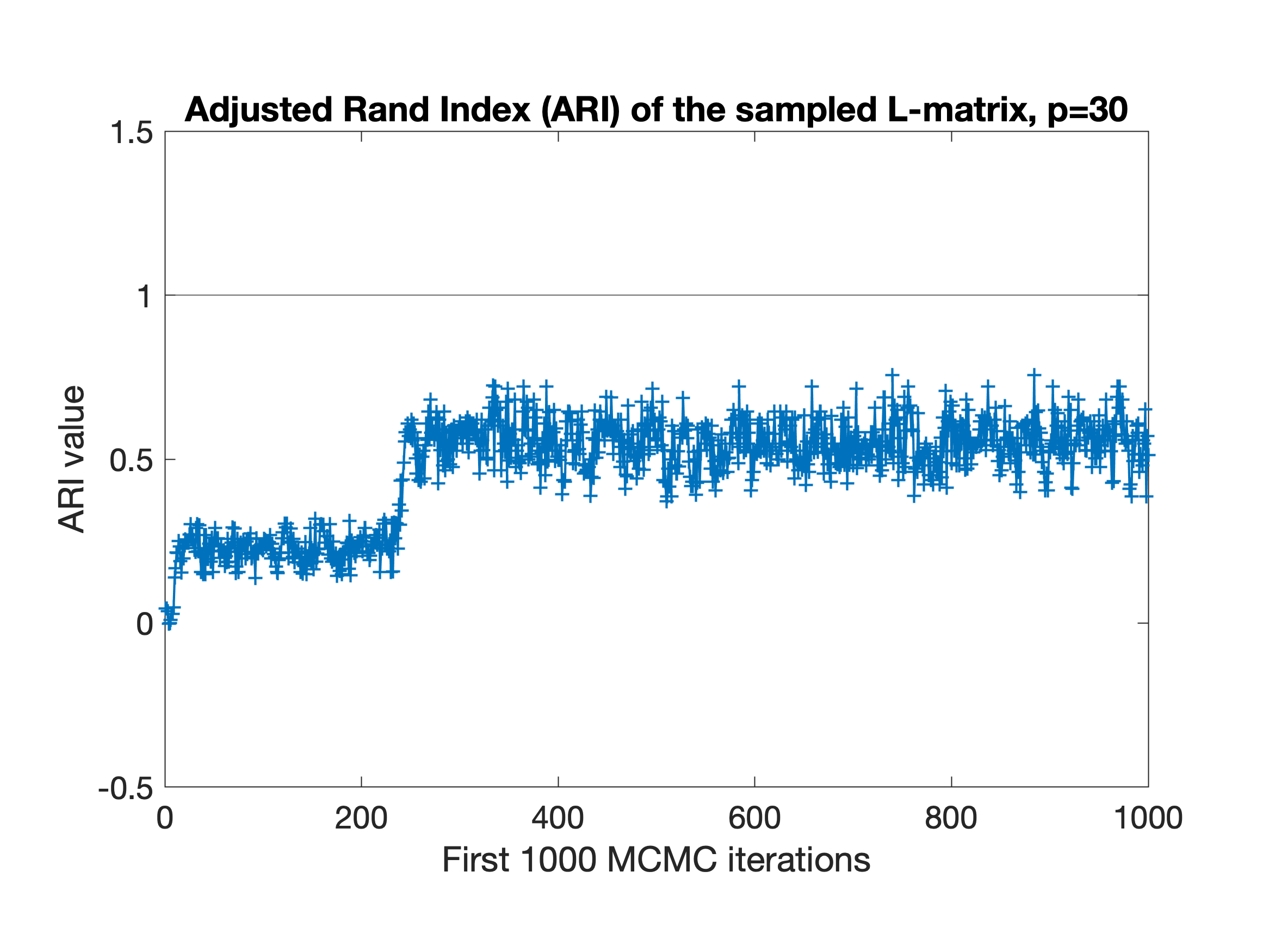}
    \hfill
    \includegraphics[width=0.48\textwidth]{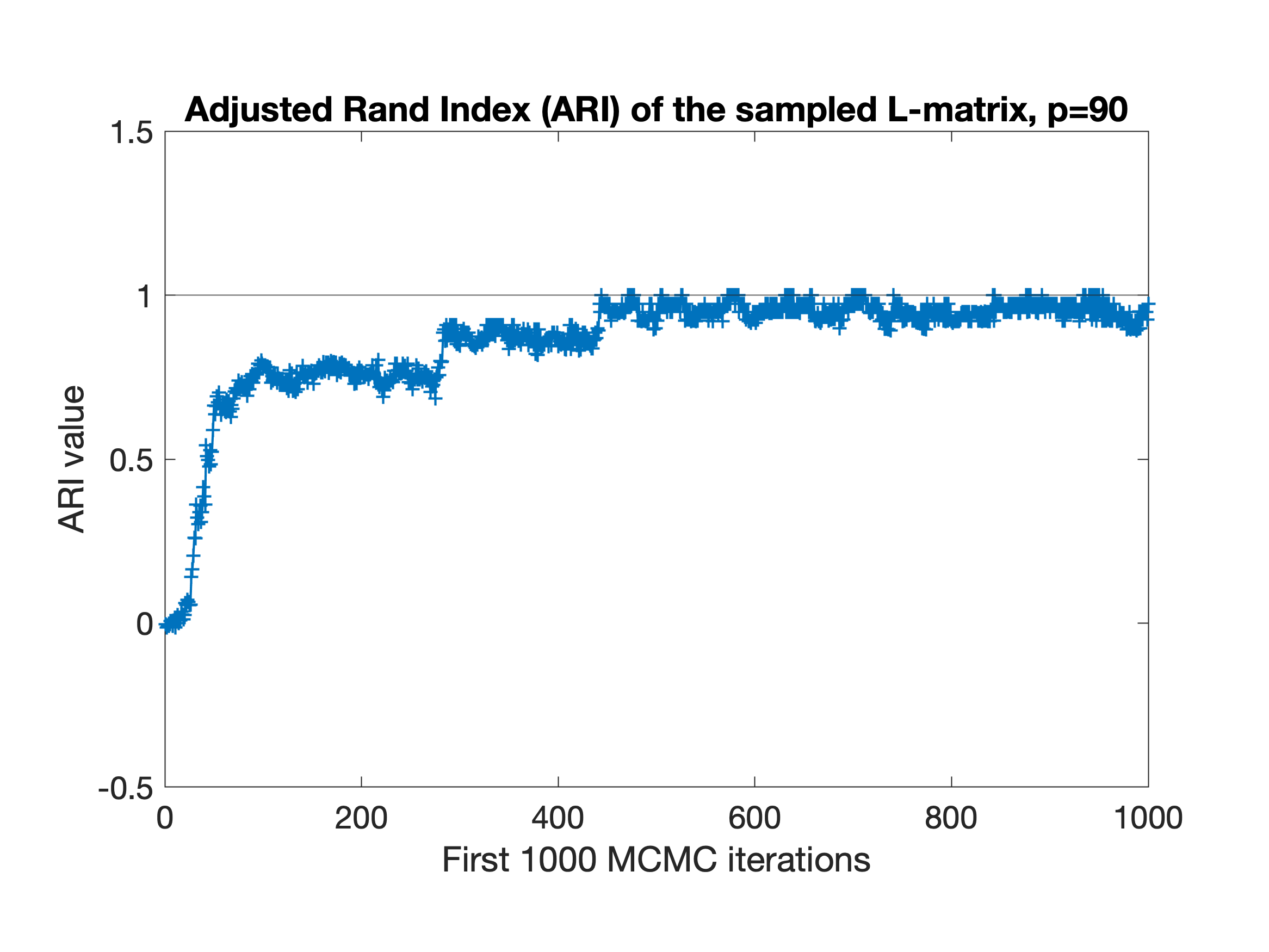}
    
    \caption{Estimation of $\LL$ (from $\bo s$) in two random simulation trials, one under $(n,p,G,K) = (500,30,6,2)$ and the other under $(n,p,G,K) = (500,90,15,2)$. In each of the first two rows, the left three plots record the sampled $\LL_{\text{iter.}}$ after the 1st, 201st, and 401st MCMC iteration, respectively. The fourth plot shows the posterior mode $\overline\LL$ and the last shows the true $\LL$. The two plots in the bottom row record the ARI of the clustering of $p$ variables given by $\LL_\text{iter.}$ along the first 1000 MCMC iterations, for each of the two simulation scenarios.}
    \label{fig-simu-lmat}
\end{figure}

We next present estimation accuracy results of both $\LL$ and $(\LLambda,\aaa)$ summarized across 50 simulation replicates in each setting.
For continuous parameters $(\bo\Lambda, \aaa)$, we calculate their Root Mean Squared Errors (RMSEs) to evaluate the estimation accuracy.
To obtain the estimation error of $(\bo\Lambda, \aaa)$ after collecting posterior samples, we need to find an appropriate permutation of the $K$ extreme latent profiles in order to compare the $(\overline{\bo\Lambda}, \overline{\aaa})$ and the true $(\bo\Lambda, \aaa)$. To this end, we first reshape each of $\overline{\bo\Lambda}$ and $\bo\Lambda$ to a $(\sum_{j=1}^p d_j) \times K$ matrix $\overline{\bo\Lambda}_{\text{mat}}$ and $\bo\Lambda_{\text{mat}}$, calculate the inner product matrix $(\bo\Lambda_{\text{mat}})^\top \overline{\bo\Lambda}_{\text{mat}}$, and then find the index $i_k$ of the largest entry in each $k$th row of the inner product matrix. Such a vector of indices $(i_1,\ldots,i_K)$ gives a permutation of the $K$ profiles, and we will compare $\overline\LLambda_{j,:,(i_1,\ldots,i_K)}$ to $\LLambda_j$ and compare $\overline\aaa_{(i_1,\ldots,i_K)}$ to $\aaa$. 
In Tables \ref{tab-acc-K2}--\ref{tab-acc-K4}, we present the RMSEs of $(\LLambda,\aaa)$ and the ARIs of $\LL$ under the aforementioned 36 different simulation settings.
The median and interquartile range of the ARIs or RMSEs across the simulation replicates are shown in these tables.

Tables \ref{tab-acc-K2}--\ref{tab-acc-K4} show that under each setting of true parameters with a fixed $(p,G,K)$, the ARIs of the variable grouping $\LL$ generally increase as sample size $n$ increases, and the RMSEs of $\LLambda$ and $\aaa$ decreases as $n$ increases. This shows the increased estimation accuracy with an increased sample size. In particular, the estimation accuracy of the variable grouping structure is quite high across the considered settings. The estimation errors are slightly larger for larger values of $K$ in Table \ref{tab-acc-K4} compared to smaller values of $K$ in Tables \ref{tab-acc-K2} and \ref{tab-acc-K3}. Overall, the simulation results empirically confirm the identifiability and estimability of the model parameters in our Dirichlet Gro-M$^3$.

\begin{table}[h!]
  \centering
  \resizebox{\textwidth}{!}{%
  \begin{tabular}{llr@{\hskip 24pt}cccccc}
    \toprule
    \multirow{2}{*}{} & \multirow{2}{*}{\centering $\{p,\, G\}$} & \multirow{2}{*}{\centering $n$} &
    \multicolumn{2}{c}{ARI of $\LL$} &
    \multicolumn{2}{c}{RMSE of ${\bo\Lambda}$} &  \multicolumn{2}{c}{RMSE of $\aaa$}\\
    \cmidrule(lr){4-5}
    \cmidrule(lr){6-7}
    \cmidrule(lr){8-9}
    & & & Median & (IQR) & Median & (IQR) & Median & (IQR) \\
    \midrule
    &   \multirow{4}{*}{\centering  $(30,\, 6)$} 
    &  250
        & 0.74 & (0.18)
        & 0.042 & (0.005)
        & 0.064 & (0.056)
        \\
    &  &   500
        & 0.88 & (0.17)
        & 0.030 & (0.004)
        & 0.031 & (0.043)
        \\
    & &  1000
        & 0.91 & (0.29)
        & 0.023 & (0.014)
        & 0.027 & (0.028)
         \\
    &  &  1500
        & 0.91 & (0.31)
        & 0.018 & (0.022)
        & 0.026 & (0.045)
         \\
\cmidrule(lr){2-9}
    \multirow{4}{*}{\centering $K=2$} 
    &  \multirow{4}{*}{\centering  $(60,\, 12)$} 
    &  250
        & 0.73 & (0.13)
        & 0.042 & (0.004)
        & 0.039 & (0.041)
        \\
    &   & 500
        & 0.79 & (0.14)
        & 0.032 & (0.003)
        & 0.031 & (0.021)
        \\
    &   & 1000
        & 0.85 & (0.20)
        & 0.027 & (0.010)
        & 0.018 & (0.029)
        \\
    &   & 1500
        & 0.81 & (0.21)
        & 0.028 & (0.016)
        & 0.024 & (0.025)
        \\
\cmidrule(lr){2-9}
    &  \multirow{4}{*}{\centering $(90,\, 15)$} 
    &  250
        & 0.95 & (0.05)
        & 0.042 & (0.003)
        & 0.045 & (0.045)
        \\
    &   &  500
        & 1.00 & (0.00)
        & 0.026 & (0.002)
        & 0.032 & (0.023)
        \\
    &   &  1000
        & 1.00 & (0.00)
        & 0.018 & (0.001)
        & 0.019 & (0.021)
        \\
    &   &  1500
        & 1.00 & (0.08)
        & 0.015 & (0.010)
        & 0.017 & (0.017)
        \\
    \bottomrule
  \end{tabular}
 }
\caption{Simulation results of the Dirichlet Gro-M$^3$ for $K=2$. 
``ARI'' of $\LL$ is the Adjusted Rand Index of the estimated variable groupings with respect to the truth. ``RMSE'' of ${\bo\Lambda}$ and ${\aaa}$ are Root Mean Squared Errors. ``Median'' and ``IQR'' are based on 50 replicates in each simulation setting.}
\label{tab-acc-K2}
\end{table}

\begin{table}[h!]
  \centering
  \resizebox{\textwidth}{!}{%
  \begin{tabular}{llr@{\hskip 24pt}cccccc}
    \toprule
    \multirow{2}{*}{} & \multirow{2}{*}{\centering $(p,\, G)$} & \multirow{2}{*}{\centering $n$} &
    \multicolumn{2}{c}{ARI of $\LL$} &
    \multicolumn{2}{c}{RMSE of ${\bo\Lambda}$} &  \multicolumn{2}{c}{RMSE of $\aaa$}\\
    \cmidrule(lr){4-5}
    \cmidrule(lr){6-7}
    \cmidrule(lr){8-9}
    & & & Median & (IQR) & Median & (IQR) & Median & (IQR) \\
\midrule
\multirow{4}{*}{\centering } 
    &  \multirow{4}{*}{\centering  $(30,\, 6)$} &  250
        &  1.00   &  (0.00)
        &  0.045  &  (0.004)
        &  0.046  &  (0.048)
        \\
    &  &  500
        & 1.00 & (0.00)
        &  0.033 &  (0.003)
        &  0.046 &  (0.059)
        \\
    &  &  1000
        & 1.00 & (0.00)
        & 0.023  &  (0.022)
        & 0.039  &  (0.037)
         \\
    &   &  1500
        &  1.00 & (0.00)
        &  0.019  &  (0.023)
        &  0.029  &  (0.032)
         \\
\cmidrule(lr){2-9}
    \multirow{4}{*}{\centering $K=3$} 
    &  \multirow{4}{*}{\centering  $(60,\, 12)$} &  250
        & 1.00 & (0.00) 
        &   0.045  &  (0.004)
        &   0.044  &  (0.030)
        \\
    &   &  500
        & 1.00 & (0.00)  
        &  0.032  &  (0.002)
        &  0.030  &  (0.018)
        \\
    &   &  1000
        & 1.00 & (0.00)  
        &  0.023  &  (0.002)
        &  0.021  &  (0.017)
         \\
    &   &  1500
        & 1.00 & (0.00)  
        &  0.018  &  (0.002)
        &  0.020  &  (0.017)
         \\
\cmidrule(lr){2-9}
    \multirow{4}{*}{\centering } 
    &  \multirow{4}{*}{\centering  $(90,\, 15)$} &  250
        &  1.00   &  (0.00) 
        &  0.045  &  (0.002)
        &  0.047  &  (0.036)
        \\
    &  &  500
        & 1.00  & (0.00)
        & 0.031 & (0.002)
        & 0.026 & (0.022)
        \\
    &  &  1000
        &  1.00  & (0.00)
        &  0.022 & (0.001)
        &  0.021 & (0.013)
        \\
    &   &  1500
        &  1.00  & (0.21)
        &  0.019 & (0.024)
        &  0.024 & (0.023)
        \\
    \bottomrule
  \end{tabular}
}
\caption{Simulation results of the Dirichlet Gro-M$^3$  for $K=3$. 
See the caption of Table \ref{tab-acc-K2} for the meanings of columns.}
\label{tab-acc-K3}
\end{table}

\begin{table}[h!]
  \centering
  \resizebox{\textwidth}{!}{%
  \begin{tabular}{llr@{\hskip 24pt}cccccc}
    \toprule
    \multirow{2}{*}{} & \multirow{2}{*}{\centering $(p,\, G)$} & \multirow{2}{*}{\centering $n$} &
    \multicolumn{2}{c}{ARI of $\LL$} &
    \multicolumn{2}{c}{RMSE of ${\bo\Lambda}$} &  \multicolumn{2}{c}{RMSE of $\aaa$}\\
    \cmidrule(lr){4-5}
    \cmidrule(lr){6-7}
    \cmidrule(lr){8-9}
    & & & Median & (IQR) & Median & (IQR) & Median & (IQR) \\
    \midrule
    &  \multirow{4}{*}{\centering  $(30,\, 6)$} &  250
        &  1.00 & (0.00)
        &  0.064  &  (0.007)
        &  0.078  &  (0.056)
        \\
    &  &  500
        &  1.00 & (0.00)
        &  0.046  &  (0.006)
        &  0.062  &  (0.072)
        \\
    &  &  1000
        & 1.00 & (0.00)
        & 0.032  &  (0.004)
        & 0.043  &  (0.046)
         \\
    &   &  1500
        &  1.00 & (0.00)
        &  0.026  &  (0.004)
        &  0.032  &  (0.036)
         \\
\cmidrule(lr){2-9}
    \multirow{4}{*}{\centering $K=4$} 
    &  \multirow{4}{*}{\centering  $(60,\, 12)$} &  250
        &  1.00   &  (0.00)
        &  0.064  &  (0.005)
        &  0.060  &  (0.031)
        \\
    &  &  500
        &  1.00   &  (0.00)
        &  0.043  &  (0.003)
        &  0.047  &  (0.027)
        \\
    &  &  1000
        & 1.00   &  (0.00)
        & 0.031  &  (0.002)
        & 0.032  &  (0.014)
         \\
    &   &  1500
        & 1.00   &  (0.00)
        & 0.025  &  (0.001)
        & 0.023  &  (0.017)
         \\
\cmidrule(lr){2-9}
    &  \multirow{4}{*}{\centering  $(90,\, 15)$} &  250
        & 1.00 & (0.00)  
        &  0.046  &  (0.004)
        &  0.053  &  (0.036)
        \\
    &   &  500
        & 1.00 & (0.00)  
        &  0.041  &  (0.003)
        &  0.037  &  (0.022)
        \\
    &   &  1000
        & 1.00 & (0.00)  
        & 0.029 &  (0.001)
        & 0.026 &  (0.027)
        \\
    &   &  1500
        & 1.00 & (0.00)  
        &  0.024 &   (0.001)
        &  0.026 &   (0.020)
        \\
    \bottomrule
  \end{tabular}
 }
\caption{Simulation results of the Dirichlet Gro-M$^3$ for $K=4$. 
See the caption of Table \ref{tab-acc-K2} for the meanings of columns.}
\label{tab-acc-K4}
\end{table}

{Our MCMC algorithm can be viewed as a novel algorithm for Bayesian factorization of probability tensors. 
To see this, note that the observed response vector ranges in the $p$-way contingency table $\bo y_{i} \in [d_1] \times [d_2]  \cdots \times [d_p]$, and the marginal probabilities of a random vector $\bo y_i$ falling each of the $\prod_{j=1}^p d_j$ cells therefore form a probability tensor with $p$ modes. 
Our Gro-M$^3$ model provides a general and interpretable hybrid tensor factorization; it reduces to the nonnegative CP decomposition when the grouping matrix equals the $p\times 1$ one-vector and reduces to the nonnegative Tucker decomposition when the grouping matrix equals the $p\times p$ identity matrix.
Specifically, our estimated Dirichlet parameters $\bo\alpha$ help define the tensor core and our estimated conditional probability parameters $\bo\lambda_{j,k}$ constitute the tensor arms.
In this regard, we view our proposed MCMC algorithm as contributing a new tensor factorization method with nice uniqueness guarantee (i.e., identifiability guarantee) and good empirical performance.}

{We conduct a simulation study to empirically verify the theoretical identifiability results. Specifically, in the simulation setting $(p,G,K) = (30, 6, 4)$, corresponding to the first setting in Table \ref{tab-acc-K4}, we now consider more sample sizes $n\in\{250,~ 500,~ 750,~ 1000,~ 1250, ~1500\}$. For each sample size, we conducted 50 independent simulation replications and calculated the average root mean squared errors (RMSEs) of the model parameters $\bo\Lambda$ and $\aaa$. Figure \ref{fig-rmse} plots the RMSEs versus the sample size $n$ and shows that as $n$ increases, the RMSEs decrease gradually. This trend provides an empirical verification of identifiability, and corroborates the conclusion that under an identifiable model, the model parameters can be estimated increasingly accurately as one collects more and more samples.}

\begin{figure}[h!]
    \centering
    \includegraphics[width=0.45\textwidth]{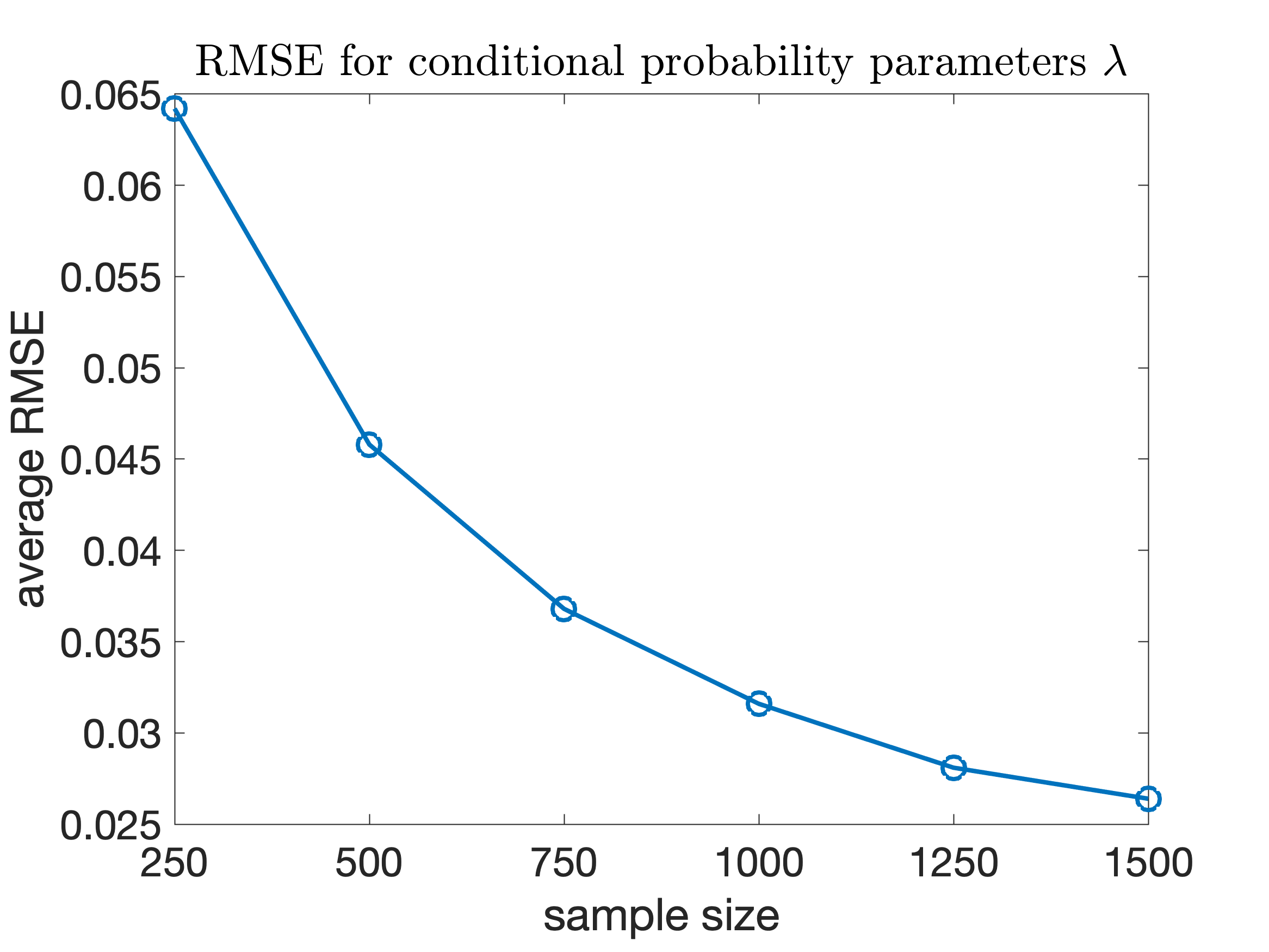}
    \quad
    \includegraphics[width=0.45\textwidth]{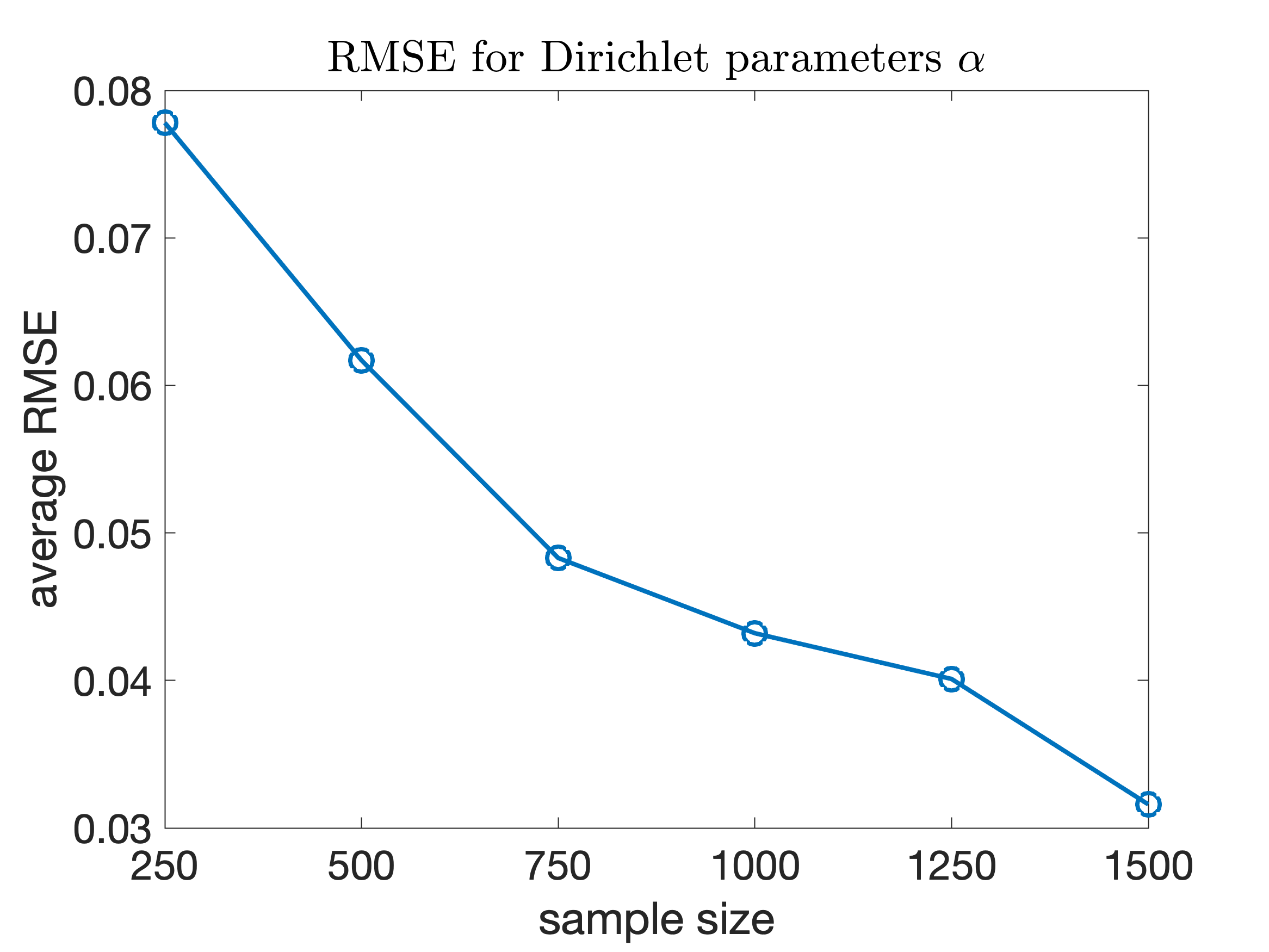}
    \caption{{Empirical verification of identifiability. Root mean square errors (RMSEs) of model parameters averaged across simulation replicates decrease as sample size increases. The simulation setting is $(p,G,K) = (30, 6, 4)$, which is the first setting in Table \ref{tab-acc-K4}.}}
    \label{fig-rmse}
\end{figure}

\subsection{Selecting $G$ and $K$ from Data}
\label{sec-simu2}

In Section \ref{sec-id}, model identifiability is established under the assumption that $G$ and $K$ are known, like many other latent structure models; for example, generic identifiability of latent class models in \cite{allman2009} is established assuming the number of latent classes is known.
But in order to provide a practical estimation pipeline applicable to real-world applications, we next briefly discuss how to select $G$ and $K$ in a data-driven way.

Our basic rationale is to use a practically useful criterion that favors a model with good out-of-sample predictive performance while remaining parsimonious. 
\cite{gelman2014ic} contains a comprehensive review of various predictive information criteria for evaluating Bayesian models.
We first considered using the Deviance Information Criterion \citep[DIC,][]{spiegelhalter2002dic}, a traditional model selection criteria for Bayesian models.
However, our preliminary simulations imply that DIC does not work well for selecting the latent dimensions in Gro-M$^3$s. In particular, we observed that DIC sometimes severely overselects the latent dimensions in our model, while that the WAIC \citep[Widely Applicable Information Criterion,][]{watanabe2010waic} has better performance in our simulation studies (see the next paragraph for details). Our observation about DIC agrees with previous studies on the inconsistency of DIC in several different settings \citep{gelman2014ic, hooten2015guide,  piironen2017comparison}.

\cite{watanabe2010waic} proved that WAIC is asymptotically equal to Bayesian leave-one-out cross validation and provided a solid theoretical justification for using WAIC to choose models with relatively good predictive ability.
WAIC is particularly useful for models with hierarchical and mixture structures, making it well suited to selecting the latent profile dimension $K$ and variable group dimension $G$ in our proposed model.
Denote the posterior samples by $\bo\theta^{(t)}$, $t=1,\ldots,T$.
For each $i\in[n]$ and $t\in[T]$, denote
\begin{align*}
p(\bo y_i \mid \bo{\theta}^{(t)})
	=&~ 
	\prod_{m=1}^G \left[ \sum_{k=1}^K \pi_{ik}^{(t)} \prod_{\ell_{j,m}^{(t)}=1} \prod_{c=1}^{d_j} \left(\lambda_{j,c,k}^{(t)}\right)^{y_{i,j,c}} \right].
\end{align*}
In particular, \cite{gelman2014ic} recommended using the following version of the WAIC, where ``lppd'' refers to \textit{log pointwise predictive density} and $p_{\text{WAIC}_2}$ measures the model complexity through the variance,
\begin{align}\label{eq-waic}
    \text{WAIC}
    &=
    -2\left(\text{lppd} - p_{\text{WAIC}_2}\right)\\ \notag
    &=
    -2\sum_{i=1}^n \log\left(\frac{1}{T} \sum_{t=1}^T p(\bo y_i\mid \bo{\theta}^{(t)}) \right)
    +
    2\sum_{i=1}^n \text{var}_{t=1}^T \left(\log p\left(\bo y_i\mid \bo{\theta}^{(t)}\right) \right),
\end{align}
where
$\text{var}_{t=1}^T$ refers to the variance based on $T$ posterior samples, with definition $\text{var}_{t=1}^T (a_t) = 1/(T-1) \sum_{t=1}^T \big(a_t - \sum_{t'=1}^T a_{t'}/T \big)^2$.
Based on the above definition, the WAIC can be easily calculated based on posterior samples. The model with a smaller WAIC is favored. 

We carried out a simulation study to evaluate how WAIC performs on selecting $G$ and $K$, focusing on the previous setting where 50 independent datasets are generated from $(n,p,G,K)=(1000,30,6,3)$. When fixing the candidate $K$ to the truth $K=3$ and varying the candidate $G_{\text{candi}} \in\{4,5,6,7,8\}$, the percentages of the datasets that each of $G=4,5,6,7,8$ is selected are 0\%, 0\%, 74\% (true $G$), 20\%, 6\%, respectively. 
When fixing the candidate $G$ to the truth $G=6$ and varying $K_{\text{candi}} \in \{2,3,4,5,6\}$, the percentages of the datasets that each of $K=2,3,4,5,6$ is selected are 0\%, 80\% (true $K$), 6\%, 4\%, 10\%, respectively. 
Further, when varying $(K,G)$ in the grid of 25 possible pairs $\{2,3,4,5,6\}\times \{4,5,6,7,8\}$, the percentage of the datasets for which the true pair $(K,G)=(3,6)$ is selected by WAIC is 58\% and neither $K$ nor $G$ ever gets underselected.
In general, our simulations show that the WAIC does not tend to underselect the latent dimensions $K$ and $G$, and that it generally has a reasonably good accuracy of selecting the truth.
We remark that here our goal was to pick a practical selection criterion that can be readily applied in real-world applications. 
To develop a selection strategy for deciding on the number of latent dimensions with rigorous theoretical guarantees under the proposed models would need future investigations.

\section{Real Data Applications}\label{sec-real}

\subsection{NLTCS Disability Survey Data}
In this section we apply Gro-M$^3$ methodology to a functional disability dataset extracted from the National Long Term Care Survey (NLTCS), created by the former Center for Demographic Studies at Duke University. 
This dataset has been widely analyzed, both with mixed membership models  \citep{erosheva2007aoas, manrique2014aoas}, and with other models for multivariate categorical data \citep{dobra2011copula, johndrow2017}.
Here we reanalyze this dataset as an illustration of our dimension-grouped mixed membership approach.

The \verb|NLTCS| dataset was downloaded from at \verb|http://lib.stat.cmu.edu/datasets/|. It is an extract  containing responses from $n=21574$ community-dwelling elderly Americans aged 65 and above, pooled over 1982, 1984, 1989, and 1994 survey waves.
The disability survey contains $p=16$ items, with respondents being either coded as healthy (level 0) or as disabled (level 1) for each item.
{Each respondent provides a 16-dimensional response vector $\bo y_i = (y_{i,1}, \ldots, y_{i,16}) \in \{0,1\} \times \cdots \times \{0,1\}$, where each variable $y_{i,j}$ follows a special categorical distribution with two categories, i.e., a Bernoulli distribution, with parameters specific to item $j$.}
Among the $p=16$ NLTCS disability items, functional disability researchers distinguish six activities of daily living (ADLs) and ten instrumental activities of daily living (IADLs).
Specifically, the first six ADL items are more basic and relate to hygiene and personal care: eating, getting in/out of bed, getting around inside, dressing, bathing, and getting to the bathroom or using a toilet. 
The remaining ten IADL items are related to activities needed to live without dedicated professional care: doing heavy house work, doing light house work, doing laundry, cooking, grocery shopping, getting about outside, travelling, managing money, taking medicine, and telephoning.

Here, we apply the MCMC algorithm developed for the Dirichlet Gro-M$^3$ to the data; the Dirichlet distribution was also used to model the mixed membership scores in \cite{erosheva2007aoas}.
Our preliminary analysis of the NLTCS data indicates the Dirichlet parameters $\aaa$ are relatively small, so we adopt a  small $\sigma_\alpha = 0.002$ in the lognormal proposal distribution in Eq.~\eqref{eq-sigalpha} in the Metropolis-Hastings sampling step.
For each setting of $(G,K)$, we run the MCMC for 40000 iterations and consider the first 20000 as burn-in to be conservative. We retain every 10th sample after the burn-in.
The candidate values for the $(G,K)$ are all the combinations of $G\in\{2,3,\ldots,15,16\}$ and $K\in\{6,7,\ldots,11,12\}$.

For selecting the values of latent dimensions $(G,K)$ in practice, we recommend picking the $(G^\star, K^\star)$ that provide the lowest WAIC value and also do not contain any empty groups of variables. 
In particular, for certain pairs of $(G,K)$ (in our case, for all $G>10$) under the NLTCS data, we observe that the posterior mode of the grouping matrix, $\overline\LL$, has some all-zero columns. 
If $\tilde G$ denotes the number of not-all-zero columns in $\overline\LL$, 
this means after model fitting, the number of groups occupied by the $p$ variables is $\tilde G < G$.
Models with $\tilde G< G$ are difficult to interpret because empty groups that do not contain any variables cannot be assigned meaning.
Therefore, we focus only on models where $\overline\LL$ does not contain any all-zero columns and pick the one with the smallest WAIC among these models.
Using this criterion, for the NLTCS data, the model with $G^\star=10$ and $K^\star=9$ is selected. 
We have observed reasonably good convergence and mixing of our MCMC algorithm for the NLTCS data.
The proposed new dimension-grouping model provides a better fit in terms of WAIC and a parsimonious alternative  to traditional MMMs.

\begin{figure}[h!]
  \centering
  \includegraphics[width=0.95\textwidth]{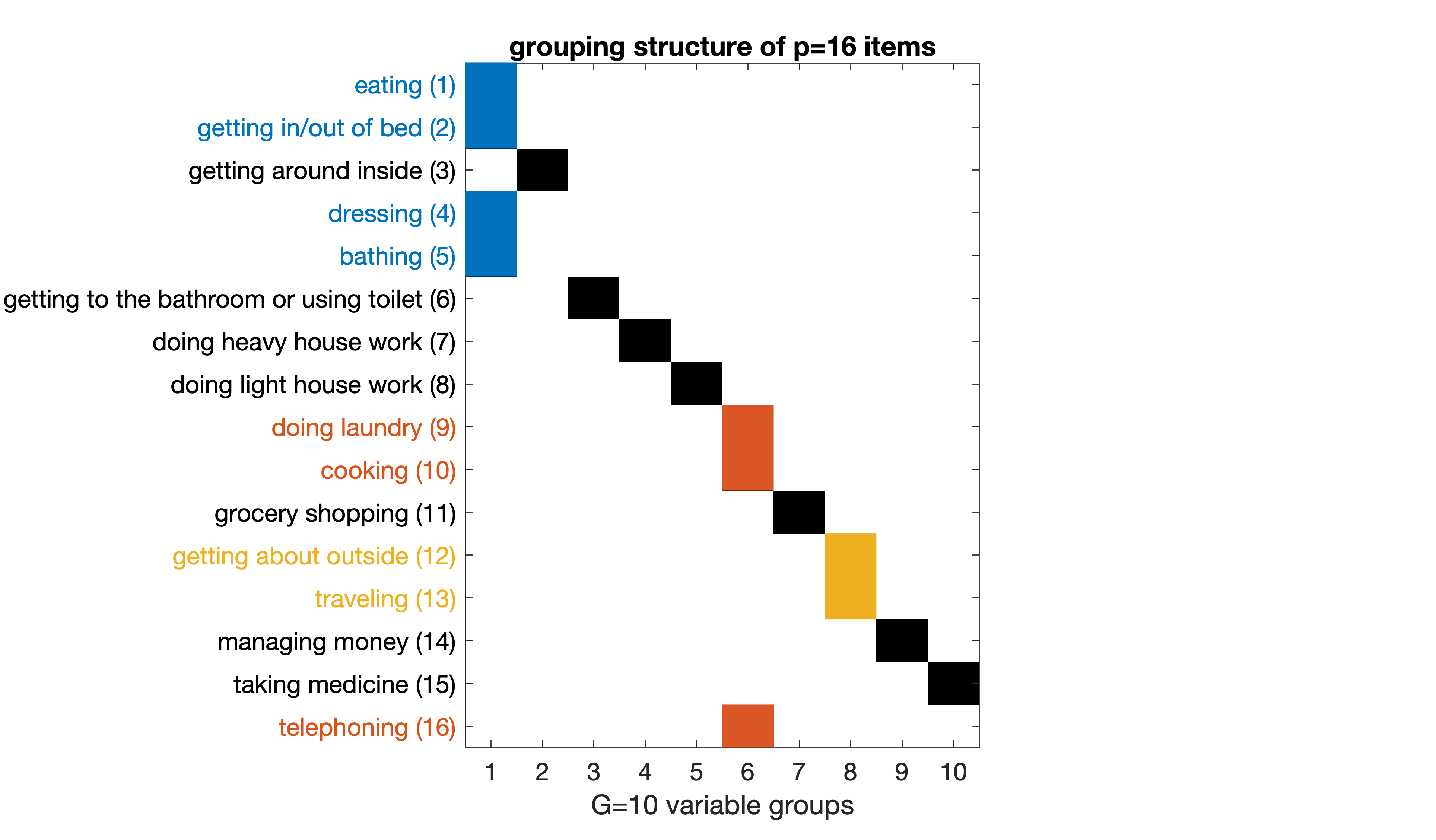}
\caption{Estimated variable grouping structure $\bo s$ (i.e., $\LL$) for the NLTCS data with $(G^\star, K^\star) = (10, 9)$. The first six items are ADL  ``activities of daily living'' and the remaining ten items are IADL ``instrumental activities of daily living''. Out of the $G^\star=10$ variable groups, the three groups containing multiple items are colored coded in blue ($j=1,2,4,5$), red ($j=9,10,16$), and yellow ($j=12,13$) for better visualization.}
\label{fig-nltcs-groupmat}
\end{figure}

We provide the estimated $\overline\LL$ under the selected model with $G^\star=10$ and $K^\star=9$ in Figure \ref{fig-nltcs-groupmat}.
{The estimated variable groupings are given in Figure \ref{fig-nltcs-groupmat}. Out of the $G^\star=10$ groups, there are three groups that contain multiple items. In Figure \ref{fig-nltcs-groupmat}, the item labels of these three groups are colored in blue ($j=1,2,4,5$), red ($j=9,10,16$), and yellow ($j=12,13$) for better visualization. These groupings obtained by our model lead to several observations. First, four out of six ADL variables ($j=1,2,4,5$) are categorized into one group. This group of items are basic self-care activities that require limited mobility. 
Second, the three IADL variables ($j=9,10,16$) in one group may be related to traditional gender roles -- these items correspond to activities performed more frequently by women than by men.
Finally, the two items $j=12$ ``getting about outside'' and  $j=13$ ``traveling'' that require high level of mobility form another group.
}
Note that
such a model-based grouping of the items is different than the established groups (ADL and IADL), and could not have been obtained by applying previous mixed membership models~\citep{erosheva2007aoas}.

In addition to the variable grouping structures, we plot posterior means of the positive response probabilities $\overline\LLambda_{:,1,:}$ in Figure \ref{fig-nltcs-prob} for the selected model. For each survey item $j\in[p]$ and each extreme latent profile $k\in[K]$, the $\LLambda_{j,1,k}$ records the conditional probability of giving a positive response of being disabled on this item conditional on possessing the $k$th latent profile. 
The $K^\star=9$ profiles are quite well separated and can be interpreted as usual in mixed membership analysis. 
For example, in Figure \ref{fig-nltcs-prob}, the leftmost column for $k=1$ represents a relatively healthy latent profile, the rightmost column for $k=9$ represents a relatively severely disabled latent profile.
As for the Dirichlet parameters $\aaa$, their posterior means are 
$\overline{\aaa} = 
(0.0245,\;    0.0289,\;    0.0074,\;   0.0176,\;    0.0231,\;    0.0193,\;    0.0001,\;    0.0001,\;    0.0242)$.
Such small values of the Dirichlet parameters imply that membership score vectors tend to be dominated by one component for a majority of individuals. This observation is
consistent with \cite{erosheva2007aoas}.
Meanwhile, here we obtain a simpler model than that in \cite{erosheva2007aoas} as each subject can partially belong to up to $G$ latent profiles according to the grouping of variables, rather than $p=16$ ones as in the traditional MMMs.

\begin{figure}[h!]
    \centering
    \includegraphics[width=0.5\textwidth]{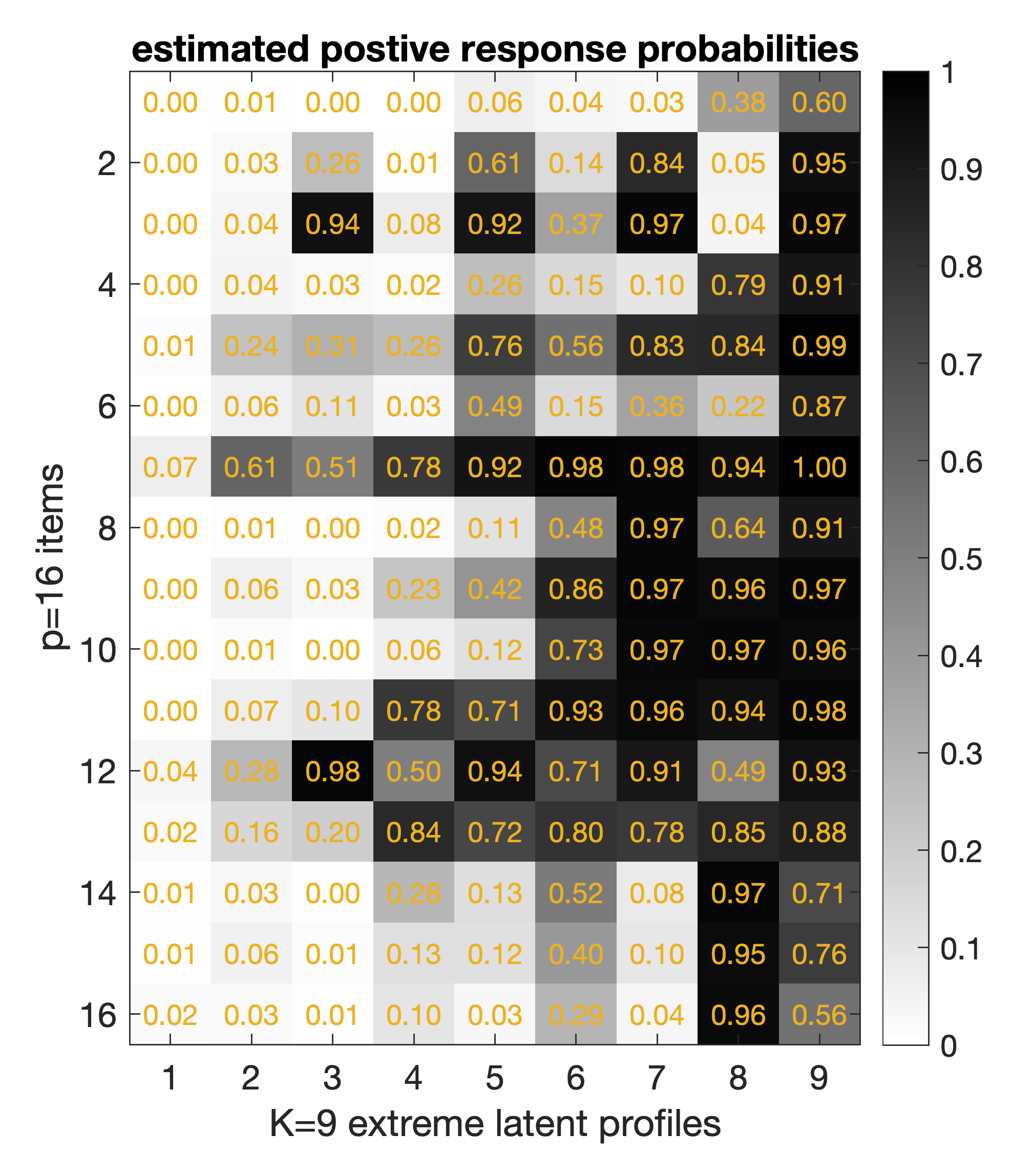}
    \caption{Estimated positive response probabilities $\LLambda_{:,1,:}$ for the NLTCS data with $(G^\star, K^\star) = (10, 9)$. Each column represents one extreme latent profile. Entries are conditional probabilities of giving a positive response ($1=\text{disabled}$) to each item given that latent profile.}
    \label{fig-nltcs-prob}
\end{figure}

{We emphasize again that the bag-of-words topic models make the exchangeability assumption, which is fundamentally different from, and actually more restrictive than, our Gro-M$^3$ when modeling non-exchangeable item response data. 
Specifically, the exchangeability assumption would force all the item parameters $\{\bo\lambda_{j,k}\in \mathbb R^{d}:~ k\in[K]\}$ across all the items $j\in[p]$ to be identical, which is unrealistic for the survey response data (or the personality test data to be analyed in Section \ref{sec-ipip}) in which different items clearly have different characteristics. 
For example, if one were to use a topic model such as LDA to analyze the NLTCS disability survey data, then a plot of the $16\times K$ conditional probability table like Figure \ref{fig-nltcs-prob} would not have been possible, because all the $p=16$ items would share the same $K$-dimensional vector of conditional Bernoulli probabilities.}

\subsection{International Personality Item Pool (IPIP) Personality Test Data}\label{sec-ipip}

We also apply the proposed method to analyze a personality test dataset containing multivariate \emph{polytomous} responses: the International Personality Item Pool (IPIP) personality test data. This dataset is publicly available from the  Open-Source Psychometrics Project website \verb|https://openpsychometrics.org/_rawdata/|.
The dataset contains $n_{\text{all}}=1005$ subjects' responses to $p=40$  Likert rated  personality test items in the International Personality Item Pool.
After dropping those subjects who have missing entries in their responses, there are $n=901$ complete response vectors left.
Each subject's observed response vector is $40$-dimensional, where each dimension ranges in $\{1,2,3,4,5\}$ with $d_1=d_2=\cdots=d_p=5$ categories. Each of these 40 items was designed to measure one of the four personality factors: Assertiveness (short as ``AS''),  Social confidence (short as ``SC''), Adventurousness (short as ``AD''), and  Dominance (short as ``DO''). Specifically, items 1-10 measure AS, items 11-20 measure SC, items 21-30 measure AD, and items 31-40 measure DO.
The responses of certain reversely-termed items (i.e., items 7--10, 16--20, 25--30) are preprocessed to be $\tilde y_{ij} = 6-y_{ij}$. We apply our new model to analyze this dataset for various numbers of variable groups $G \in \{3,4,5,6,7\}$ and $K=4$ extreme latent profiles, and the WAIC selects the model with $G=5$ groups. We plot the posterior mode of the estimated grouping matrix in Figure \ref{fig-ipip}, together with the content of each item.

Figure \ref{fig-ipip} shows that our new modeling component of variable grouping is able to reveal the item blocks that measure different personality factors in a totally unsupervised manner. 
Moreover, the estimated variable grouping cuts across the four established personality factors to uncover a more nuanced structure. For example, the group of items $\{$AS1, SC4, SC10$\}$ concerns the verbal expression aspect of a person; the group of items $\{$AS3--AS10, SC5, SC7$\}$ concerns a person's intention to lead and influence other people.
In summary, for this new personality test dataset, the proposed Gro-M$^3$ not only provides better model fit than the usual GoM model (since $G=5 \ll p=40$ is selected by WAIC), but also enjoys interpretability and uncovers meaningful subgroups of the observed variables.

\begin{figure}[h!]
\centering
    \includegraphics[width=0.8\textwidth]{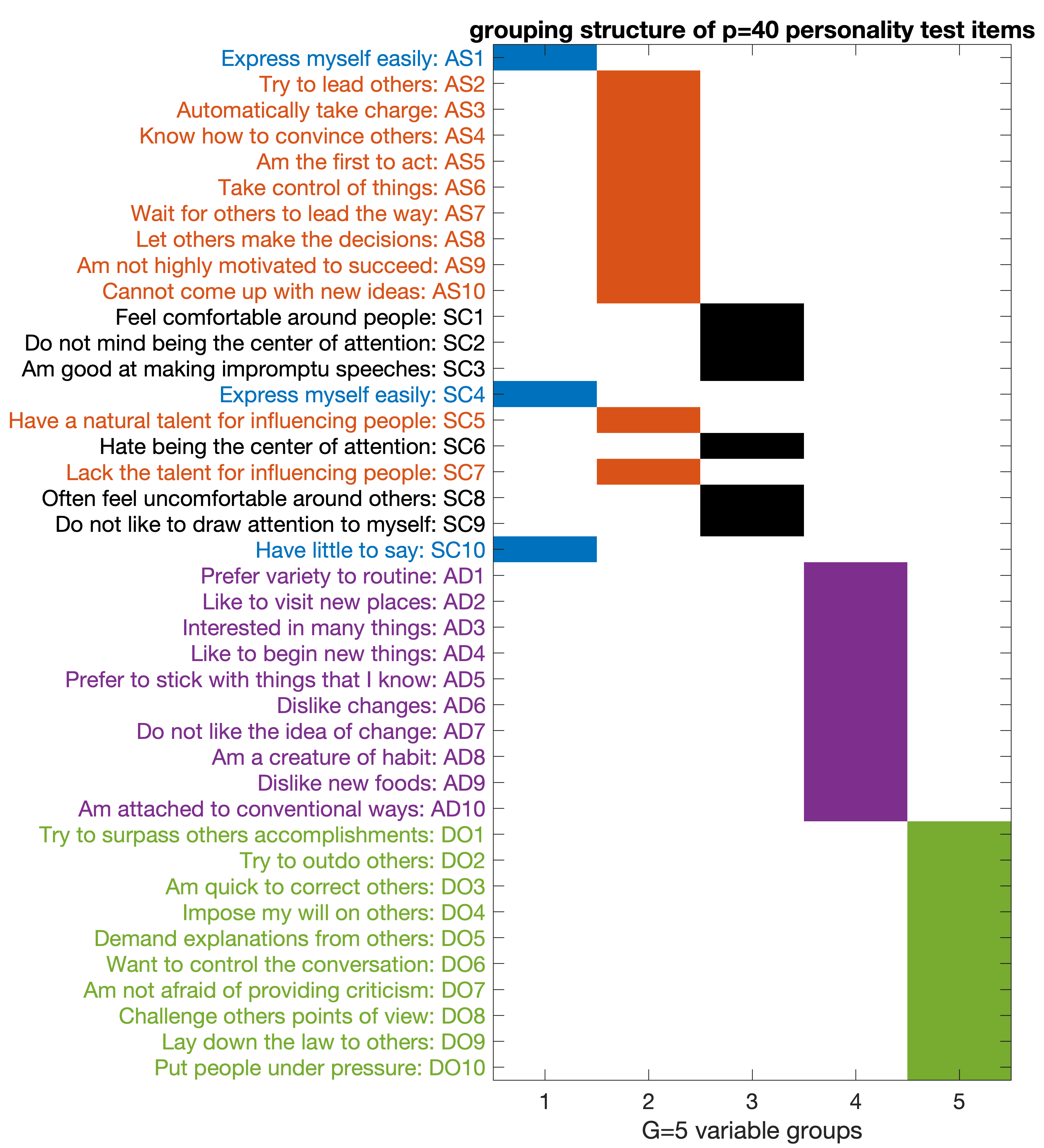}
    \caption{{IPIP personality test items grouping structure estimated from our Gro-M$^3$. Item type abbreviations are: ``AS'' represents ``Assertiveness'', ``SC'' represents ``Social confidence'', ``AD'' represents ``Adventurousness'', and ``DO'' represents ``Dominance''.}}
\label{fig-ipip}
\end{figure}

We also conduct experiments to compare our probabilistic hybrid decomposition Gro-M$^3$ with the probabilistic CP decomposition (the latent class model in \cite{dunson2009lcm}) and the probabilistic Tucker decomposition (the GoM model in \cite{erosheva2007aoas}) on the IPIP personality test data. 
After fitting each tensor decomposition method to the data, we calculate the model-based Cramer's V measure between each pair of items and see how different methods perform on recovering meaningful item dependence structure.
Figure \ref{fig-ipip-pair} presents the model-based pairwise Cramer's V calculated using the three tensor decompositions, along with the model-free Cramer's V calculated directly from data.
Figure \ref{fig-ipip-pair} shows that our Gro-M$^3$ decomposition clearly outperforms probabilistic CP and Tucker decomposition in recovering the meaningful block structure of the personality test items. 
\color{black}

\begin{figure}[h!]
    \centering
    \resizebox{0.75\textwidth}{!}{
    \includegraphics[height=3cm]{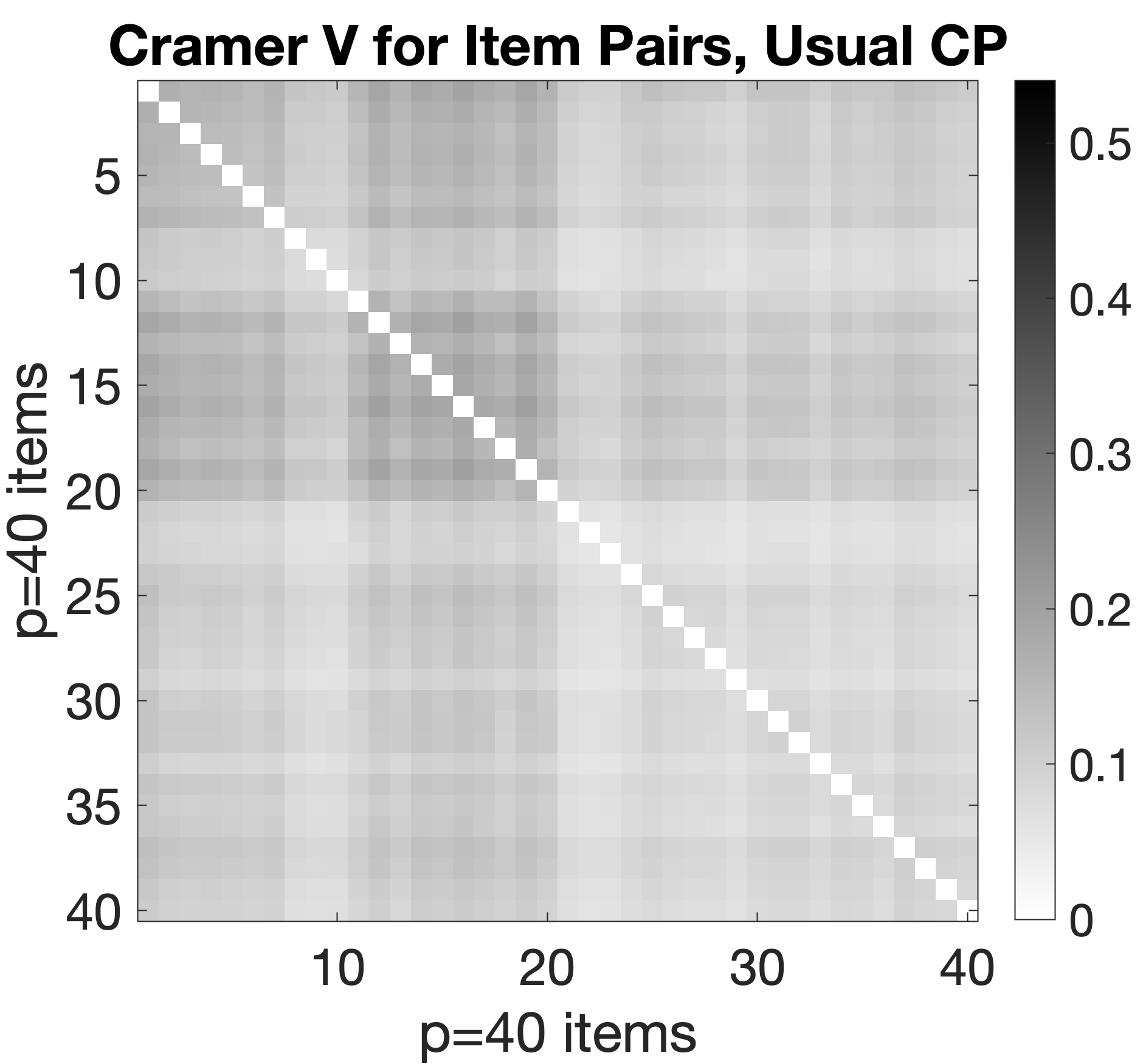}
    ~
    \includegraphics[height=3cm]{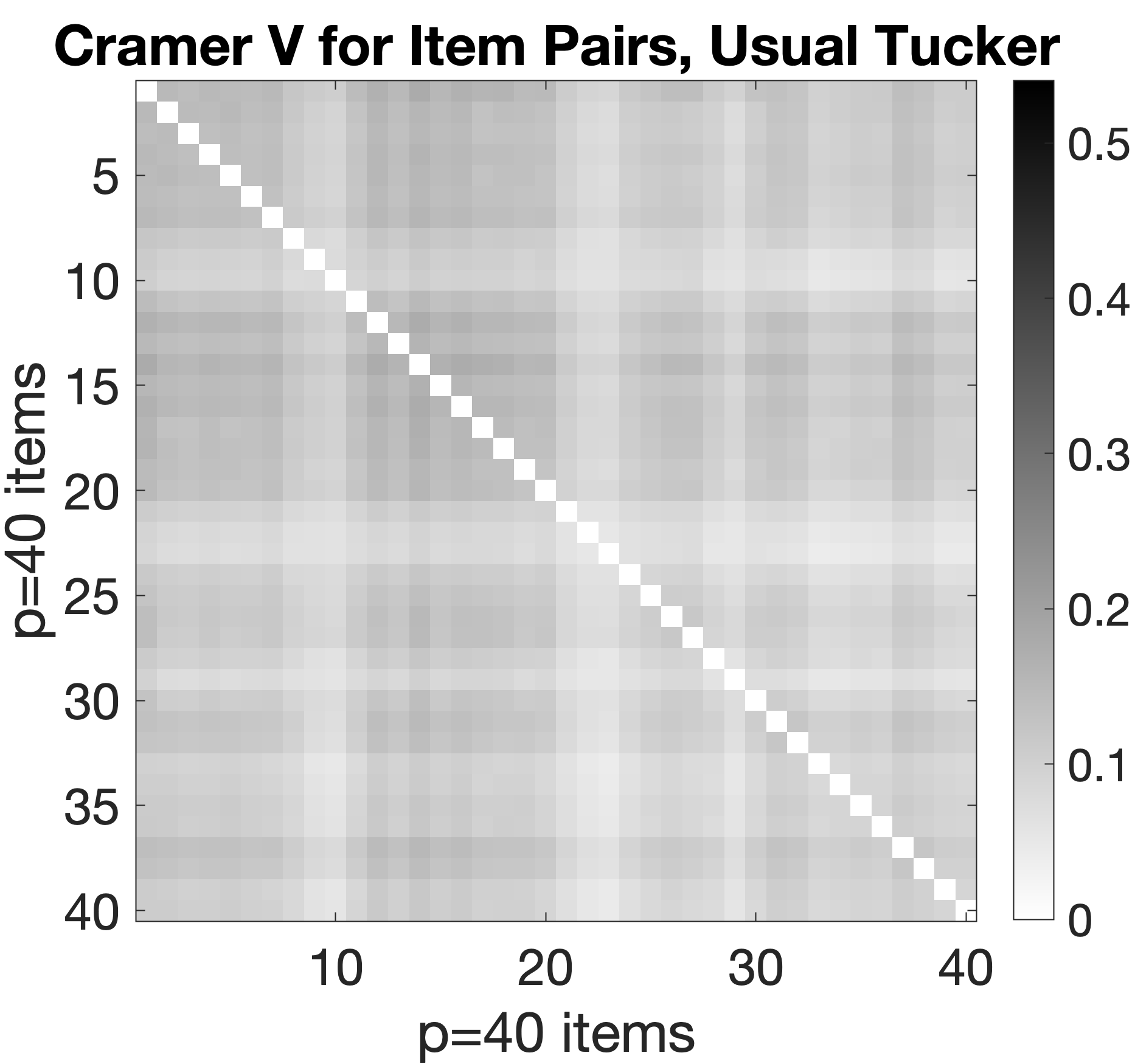}
    }
    
    \bigskip
    \resizebox{0.75\textwidth}{!}{
    \includegraphics[height=3cm]{figures/ipip_crv_grom3.png}
    ~
    \includegraphics[height=3cm]{figures/ipip_crv_sample.png}
    }
    \caption{{Upper two panels: Cramer's V posterior means for item pairs obtained using the usual CP decomposition (latent class model) and the usual Tucker decomposition (grade of membership model). Bottom left: Cramer's V posterior means for item pairs obtained using the Gro-M$^3$. Bottom right: Sample Cramer's V for item pairs calculated directly from data.}}
    \label{fig-ipip-pair}
\end{figure}

\color{black}

\section{Discussion}\label{sec-discuss}

We have proposed a new class of mixed membership models for multivariate categorical data, dimension-grouped mixed membership models (Gro-M$^3$s), studied its model identifiability, and developed a Bayesian inference procedure for Dirichlet Gro-M$^3$s.
On the methodology side, the new model strikes a nice balance between model flexibility and model parsimony.
Considering popular existing latent structure models for multivariate categorical data, the Gro-M$^3$ bridges the parsimonious yet insufficiently flexible Latent Class Model (corresponding to CP decomposition of probability tensors) and the very flexible yet not parsimonious Grade of Membership Model (corresponding to Tucker decomposition of probability tensors).
On the theory side, we establish the identifiability of population parameters that govern the distribution of Gro-M$^3$s. 
The quantities shown to be identifiable include not only the continuous model parameters, but also the key discrete structure -- how the variables' latent assignments are partitioned into groups.
The obtained identifiability conclusions lay a solid foundation for reliable statistical analysis and real-world applications. 
We have performed Bayesian estimation for the new model using a Metropolis-Hastings-within-Gibbs sampler. Numerical studies show that the method can accurately estimate the quantities of interest, empirically validating the identifiability results.

{For the special case of binary responses with $d_1=\cdots=d_p=2$,
as pointed out by a reviewer, models with Bernoulli-to-latent-Poisson link in \cite{zhou2016augmentable} and the Bernoulli-to-latent-Gaussian link in multivariate item response theory models in \cite{embretson2013item} are useful tools that can capture certain lower-dimensional latent constructs. 
Our model differs from these models in terms of statistical and practical interpretation.
In our Gro-M$^3$, each subject's latent variables are a mixed membership vector $\bo\pi_i$ ranging in the probability simplex $\Delta^{K-1}$, and can be interpreted as that each subject is a partial member of each of the $K$ extreme latent profiles.
For $k\in[K]$, the $k$th extreme latent profile also can be directly interpreted by inspecting the estimated item parameters $\{\bo\lambda_{j,k}:~j\in[p]\}$. Geometrically, the entry $\pi_{ik}$ captures the relative proximity of each subject to the $k$th extreme latent behavioral profile.
Such an interpretation of individual-level mixtures are highly desirable in applications such as social science surveys \citep{erosheva2003bayesian} and medical diagnosis \citep{woodbury1978gom}, where each extreme latent profile represents a prototypical response pattern. 
Therefore, in these applications, the mixed membership modeling is more interpretable and preferable to using a nonlinear transformation of certain underlying Gaussian or Poisson latent variables to model binary matrix data (such as the Bernoulli-to-latent-Poisson or Bernoulli-to-latent-Gaussian link).}

{
We remark that our proposed Gro-M$^3$ covers the usual GoM model as a special case. In fact, the GoM model can be readily recovered by setting our grouping matrix $\mathbf{L}$ to be the $p\times p$ identity matrix (i.e., $\mathbf{L} = \mathbf I_p$).
In terms of practical estimation, we can simply fix $\mathbf{L}=\mathbf I_p$ throughout our MCMC iterations and estimate other quantities in the same way as in our current algorithm.
Using this approach, we have compared the performance of our flexible Gro-M$^3$ and the classical GoM model in the real data analyses.
Specifically, for both the NLTCS disability survey data and the IPIP personality test data, fixing $\mathbf{L}=\mathbf I_p$ with $G=p$ variable groups gives larger WAIC values than the selected more parsimonious model with $G \ll p$. This indicates that the traditional GoM model is not favored by the information criterion and gives a poorer model fit to the data.
We also point out that our MCMC algorithm can be viewed as a novel Bayesian factorization algorithm for probability tensors, in a similar spirit to the existing Bayesian tensor factorization methods such as \cite{dunson2009lcm} and \cite{zhoudunson2015}. 
Our Bayesian Gro-M$^3$ factorization outperforms usual probabilistic tensor factorizations in recovering the item dependence structure in the IPIP personality data analysis.
Therefore, we view our proposed MCMC algorithm as contributing a new type of tensor factorization approach with nice uniqueness guarantee (i.e., identifiability guarantee) and a Bayesian factorization procedure with good empirical performance.
}

Our modeling assumption of the variable grouping structure can be useful to other related models.
For example, \cite{manrique2014aoas} proposed a longitudinal MMM to capture heterogeneous pathways of disability and cognitive trajectories of elderly population for disability survey data.
The proposed dimension-grouping assumption can provide an interesting new interpretation to such longitudinal settings. Specifically, when survey items are answered in multiple time points, it may be plausible to assume that a subject's latent profile locally persists for a block of items, before potentially switching to a different profile for the next block of items. This can be readily accommodated by the dimension-grouping modeling assumption, with the slight modification that items belonging to the same group should be forced to be close in time.
Our identifiability results can be applied to this setup. Similar computational procedures can also be developed.
Furthermore, although this work focuses on modeling multivariate categorical data,  the applicability of the new dimension-grouping assumption is not limited to such data. A similar assumption may be made in other mixed membership models; examples include the generalized latent Dirichlet models for mixed data types studied in \cite{zhao2018gmm}.

In terms of identifiability, the current work has focused on the population quantities, including the variable grouping matrix $\LL$, the conditional probability tables $\bo\Lambda$, and the Dirichlet parameters $\aaa$. 
In addition to these \textit{population parameters}, an interesting future question is the identification of individual mixed membership proportions $\{\bo\pi_i;\, i=1,\ldots,n\}$ for subjects \textit{in the sample}. Studying the identification and accurate estimation of $\bo\pi_i$'s presumably requires quite different conditions from ours. A recent work \citep{mao2020jasa} considered a similar problem for mixed membership stochastic block models for network data. 
Finally, in terms of estimation procedures, in this work we have employed a Bayesian approach to Dirichlet Gro-M$^3$s, and the developed MCMC sampler shows excellent computational performance. 
In the future, it would also be interesting to consider method-of-moments estimation for the proposed models related to  \cite{zhao2018gmm} and \cite{tan2017partitioned}.

{This work has focused on proposing a new interpretable and identifiable mixed membership model for multivariate categorical data, and our MCMC algorithm has satisfactory performance in real data applications. 
In the future, it would  be interesting to develop scalable and online variational inference methods, which would make the model more applicable to massive-scale real-world datasets.
We expect that it is possible to develop variational inference algorithms similar in spirit to \cite{blei2003lda} for topic models and \cite{airoldi2008msbm} for mixed membership stochastic block models to scale up computation.
In addition, just as the hierarchical Dirichlet process \citep{teh2006hdp} is a natural nonparametric generalization of the parametric latent Dirichlet allocation \citep{blei2003lda} model,
it would also be interesting to generalize our Gro-M$^3$ to the nonparametric Bayesian setting to automatically infer $K$ and $G$. 
Developing a method to automatically infer $K$ and $G$  will be of great practical value, because in many situations there might not be enough prior knowledge for these quantatities. We leave these directions for future work. 
}

\singlespacing

\section*{Acknowledgements}
This work was partially supported by NIH grants R01ES027498 and R01ES028804, NSF grants  DMS-2210796, SES-2150601 and SES-1846747,  and IES grant  R305D200015. This project also received funding from the European Research Council under
the European Union’s Horizon 2020 research and innovation program (grant agreement No.
856506).
The authors thank the Action Editor Dr. Mingyuan Zhou and three anonymous reviewers for constructive and helpful comments that helped improve this manuscript.

\vskip 0.2in
\bibliographystyle{apalike}
\bibliography{ref}


\clearpage
\begin{center}
    {\LARGE \bf Supplementary Material}
\end{center}
\medskip

This Supplementary Material contains two sections. The first section Supplement A contains the proofs of all the theoretical results in the paper. The second section Supplement B presents a note on the pairwise mutual information measures between categorical variables.

\section*{Supplement A: Proofs of Theoretical Results}

\subsection{Proof of Theorem \ref{thm-attr1}}
For notational simplicity, from now on we will omit the subscript $i$ of subject-specific random variables without loss of generality; all such variables including $\bo y$ and $\bo z$ should be understood as associated with a random subject.
Denote by $\bo z = (z_1,\ldots,z_G) \in[K]^G$ a configuration of the latent profiles realized for the $G$ groups of variables.
Recall that given a fixed grouping matrix $\LL$, the associated group assignment vector $\bo s=(s_1,\ldots,s_p)$ is defined as $s_j=g$ if and only if $\ell_{j,g}=1$. 
We next introduce a new notation.
For each variable $j\in[p]$, each category $c_j\in[d_j]$, and each possible latent profile configuration $\bo z\in \{1,\ldots,K\}^G$, define a new parameter $\gamma_{j, c_j, \bo z}$ to be
\begin{align}\label{eq-newlambda}
    \gamma_{j, c_j, \bo z}
    =&~ \lambda_{j, c_j, z_{s_j}}.
\end{align}
Collect all the $\gamma$-parameters in $\bo\Gamma=(\gamma_{j, c_j, \bo z})$, then $\bo\Gamma$ is a three-way array (which is a tensor of size $p\times d\times K^G$ if $d_1=\cdots=d_p=d$) since $j\in[p]$, $c_j\in[d_j]$, and $\bo z\in [K]^G$.
For each $j\in[p]$, we will denote the $d_j\times K^G$ matrix  $\bo\Gamma_{j,:,:}$ by $\bo\Gamma_j$ for simplicity.
The representation in \eqref{eq-newlambda} implies that many entries in $\bo\Gamma$ are equal. 
Specifically, for two arbitrary latent assignment vectors $\bo z=(z_{1}, \ldots, z_{G})$ and $\bo z'=(z'_{1}, \ldots, z'_{G})$ with $\bo z \neq \bo z'$, as long as $z_{s_j} = z'_{s_j}$ there would be $\gamma_{j, c_j, \bo z} = \gamma_{j, c_j, \bo z'}$. 
We choose to use the over-parameterization in \eqref{eq-newlambda} since this notation facilitates the study of identifiability through the underlying tensor decomposition structure, as will be revealed soon. In particular, the $\bo\Gamma_j$'s have the following property.

\begin{lemma}\label{lem-gamma}
Under the definition in \eqref{eq-newlambda}, 
for any set of indices $S\subseteq [p]$ such that $\{s_{j}:\, j\in S\} \supseteq  [G]$, there is
\begin{align}\label{eq-gamma}
    \bigotimes_{g=1}^G \left(\bigodot_{j\in S:\, s_j=g} \bo\Lambda_j \right) =  \bigodot_{j\in S} \bo\Gamma_{j}.
\end{align}
\end{lemma}
%

%
Now we can equivalently rewrite the previous model specification \eqref{eq-ctucker} as follows,
\begin{align}\notag
    &~\mathbb P^{\text{C-M}^3}(y_1=c_1, \ldots, y_p = c_p \mid \LL, \bo\Lambda, \bo\Phi) 
    = \pi_{c_1,\ldots,c_p}
    \\ \notag
    =&~ \sum_{z_1=1}^K \cdots\sum_{z_G=1}^K \phi_{z_1,\ldots,z_G}\prod_{j=1}^p \lambda_{j, c_j, z_{s_j}}
    = \sum_{z_1=1}^K \cdots\sum_{z_G=1}^K \phi_{z_1,\ldots,z_G}\prod_{j=1}^p \gamma_{j, c_j, \bo z}
    \\ \label{eq-c3m-repar}
    =&~ \sum_{\bo z\in[K]^G} \phi_{\bo z} \prod_{j=1}^p \gamma_{j, c_j, \bo z}
    =
    \mathbb P(y_1=c_1, \ldots, y_p = c_p \mid \bo\Gamma, \bo\Phi),
\end{align}
where $\bo c=(c_1,\ldots, c_p)\in\times_{j=1}^p [d_j]$.
Denote by $\bo\Phi=(\phi_{z_1,\ldots,z_G})$ a $G$-th order tensor of size $K\times \cdots \times K$. 
Denote by $\vect(\bo\Pi)$ the vectorized version of $\bo\Pi$, so $\vect(\bo\Pi)$ is a vector of length $\prod_{j=1}^p d_j$; in particular, this vector has entries defined as follows,
\begin{align}\label{eq-defvec}
    \vect(\bo\Pi)_{c_1 + (c_2-1)d_1 + \cdots + (c_p-1)d_1\cdots d_{p-1}} = 
    \pi_{c_1,c_2,\cdots,c_p}
\end{align}
for any $\bo c=(c_1,\ldots, c_p) \in \times_{j=1}^p [d_j]$.
Suppose alternative parameters $(\overline\LL, \overline{\bo\Lambda}, \overline{\bo\Phi})$ lead to the same distribution of the observed variables; that is
    $\mathbb P(\bo y = \bo c\mid \LL, \bo\Lambda, \bo\Phi)
    =
    \mathbb P(\bo y = \bo c\mid \overline\LL, \overline{\bo\Lambda}, \overline{\bo\Phi})$
holds for each possible response pattern $\bo c\in \times_{j=1}^p [d_j]$.
Then by the equivalence in \eqref{eq-c3m-repar}, we also have 
$\mathbb P(\bo y = \bo c\mid \bo\Gamma, \bo\Phi)
    =
    \mathbb P(\bo y = \bo c\mid  \overline{\bo\Gamma}, \overline{\bo\Phi})$ for all $\bo c\in \times_{j=1}^p [d_j]$.

The following two lemmas will be useful.

\begin{lemma}\label{lem-vect}
Without loss of generality, suppose the first $\sum_{j=1}^p \ell_{j,1}$ variables belong to the first group, the second $\sum_{j=1}^p \ell_{j,2}$ variables belong to the second group, etc. That is, the matrix $\LL$ takes a block-diagonal form.
The Gro-M$\,^3$ in \eqref{eq-c3m-repar} implies the following identity
\begin{align}\label{eq-vecpi}
    \vect(\bo\Pi) = 
     \left\{\bigotimes_{g=1}^G \bigodot_{j:\, \ell_{j,g}=1} \bo\Lambda_j\right\}
     \cdot 
     \vect(\bo\Phi).
\end{align}
\end{lemma}

\begin{lemma}\label{lem-krkr}
    Suppose there are two disjoint sets of $G$ observed variables $S^{(1)} = \{j_1^{(1)}, \ldots, j_G^{(1)}\}$ and $S^{(2)} = \{j_1^{(2)}, \ldots, j_G^{(2)}\}$ satisfying $s_{j^{(1)}_g} = s_{j^{(2)}_g} = g$ for each $g=1,\ldots,G$. Then 
    \begin{align}\label{eq-krkr}
      \bigotimes_{g=1}^G 
      \left\{\bo\Lambda_{j^{(1)}_g} \bigodot \bo\Lambda_{j^{(2)}_g} \right\}
      =
      \left\{\bigotimes_{g=1}^G \bo\Lambda_{j^{(1)}_g}\right\} \bigodot 
      \left\{\bigotimes_{g=1}^G \bo\Lambda_{j^{(2)}_g}\right\} ~~\text{up to a permutation of rows}.
    \end{align}
    If there further is $s_{j^{(3)}_G} = G$ for some $j^{(3)}_G \in [p]$, then up to a permutation of rows there is
    \begin{align}\label{eq-krkr2}
      &~\bigotimes_{g=1}^{G-1}
      \left\{\bo\Lambda_{j^{(1)}_g} \bigodot \bo\Lambda_{j^{(2)}_g} \right\}
      \bigotimes 
      \left\{\bo\Lambda_{j^{(1)}_G} \bigodot \bo\Lambda_{j^{(2)}_G}
      \bigodot \bo\Lambda_{j^{(3)}_G}
      \right\}
      \\ \notag
      =&~
      \left\{\bigotimes_{g=1}^{G} \bo\Lambda_{j^{(1)}_g}\right\} \bigodot 
      \left\{\bigotimes_{g=1}^{G-1} \bo\Lambda_{j^{(2)}_g} \bigotimes \left(\bo\Lambda_{j^{(2)}_G}\bigodot \bo\Lambda_{j^{(3)}_G} \right) \right\}.
    \end{align}
\end{lemma}

We continue with the proof of Theorem \ref{thm-attr1}. 
Under the conditions of the theorem, without loss of generality we can assume that for each $g\in[G]$, the first three variables (among the $p$ ones) belonging to the $g$th group have their corresponding $\bo\Lambda_j$ full-column-rank; denote the indices of these three variables by $j^{(1)}_g, j^{(2)}_g, j^{(3)}_g$. For example, if $\LL$ takes the block diagonal form, then for $g=1$ such three variables are indexed by $\left\{j^{(1)}_1,\, j^{(2)}_1,\, j^{(3)}_1\right\} = \left\{1,2,3\right\}$; for $g=2$ they are indexed by $\left\{j^{(1)}_2,\, j^{(2)}_2,\, j^{(3)}_2\right\} = \left\{\sum_{j=1}^p \ell_{j,1}+1, \sum_{j=1}^p \ell_{j,1}+2, \sum_{j=1}^p \ell_{j,1}+3\right\}$, etc.
Define the following sets of variable indices,
\begin{align*}
    S^{(m)} = \left\{j^{(m)}_1, \ldots, j^{(m)}_G\right\} \text{ for }m=1,2,3;\quad
    S^{(0)} = \{1,\ldots,p\}\setminus\bigcup_{m=1}^3  S^{(m)}.
\end{align*}
Lemmas \ref{lem-vect} and  \ref{lem-krkr} imply that we can write $\vect(\bo\Pi)$ under the true parameters as follows
\begin{align}\notag
&~\vect(\bo\Pi)=
    \left\{\bigotimes_{g=1}^G \bigodot_{j:\, \ell_{j,g}=1} \bo\Lambda_j\right\}
     \cdot 
    \vect(\bo\Phi)\\ \notag
    = &~
    \left[\left\{\bigotimes_{g=1}^G \bo\Lambda_{j_g^{(1)}}\right\}
    \bigodot
    \left\{\bigotimes_{g=1}^G \bo\Lambda_{j_g^{(2)}}\right\}
    \bigodot
    \left\{ \bigotimes_{g=1}^G
    \left(\bo\Lambda_{j_g^{(3)}} \bigodot\left(\bigodot_{j\in S^{(0)}:\, \ell_{j,g}=1} \bo\Lambda_j\right) \right)
    \right\}
    \right]\cdot 
    \vect(\bo\Phi)\\[2mm] \label{eq-3way}
    \stackrel{(\star)}{=} &~
    \left[\left\{\bigodot_{g=1}^G \bo\Gamma_{j_g^{(1)}}\right\}
    \bigodot
    \left\{\bigodot_{g=1}^G \bo\Gamma_{j_g^{(2)}}\right\}
    \bigodot
    \left\{ \bigodot_{j\in S^{(3)} \cup S^{(0)}}
    \bo\Gamma_{j} 
    \right\}
    \right]\cdot 
    \vect(\bo\Phi).
\end{align}
The last equality $(\star)$ above follows from Lemma \ref{lem-gamma}, by noting that for each $m=1,2,3$, the index set $\{j_1^{(m)},\ldots,j_G^{(m)}\} = [K]$. The last equality in the above display results from the property of the Khatri-Rao product.
Define
\begin{align*}
    f^{(1)}(\bo\Gamma) := \bigodot_{g=1}^G \bo\Gamma_{j_g^{(1)}} = \bigodot_{j\in S^{(1)}} \bo\Gamma_j,
    \quad
    f^{(2)}(\bo\Gamma) := \bigodot_{g=1}^G \bo\Gamma_{j_g^{(2)}} = \bigodot_{j\in S^{(2)}} \bo\Gamma_j,
    \quad
    f^{(3)}(\bo\Gamma) := \bigodot_{j\in S^{(3)} \cup S^{(0)}}
    \bo\Gamma_{j}.
\end{align*}
It can be seen that the definitions of the above three functions $f^{(1)}(\cdot),~f^{(2)}(\cdot),~f^{(3)}(\cdot)$ of $\bo\Gamma$ only depend on the two sets of variable indices $S^{(1)}$ and $S^{(2)}$, which in turn are determined by the true grouping matrix $\LL$. 
Now \eqref{eq-3way} can be further written as
\begin{align*}
    \vect(\bo\Pi) = \left( f^{(1)}(\bo\Gamma) \bigodot f^{(2)}(\bo\Gamma) \bigodot f^{(3)}(\bo\Gamma) \right)\cdot \vect(\bo\Phi)
    = \bigodot_{j=1}^p \bo\Gamma_{j} \cdot \vect(\bo\Phi).
\end{align*}
So for true parameters $(\bo\Gamma, \bo\Phi)$ and alternative parameters $(\overline{\bo\Gamma}, \overline{\bo\Phi})$ that lead to the same distribution of the observed $\bo y$, we have
\begin{align}\label{eq-gamdec}
    \vect(\bo\Pi) = &\left( f^{(1)}(\bo\Gamma) \bigodot f^{(2)}(\bo\Gamma) \bigodot f^{(3)}(\bo\Gamma) \right)\cdot \vect(\bo\Phi)
    \\ \notag
    = &\left( f^{(1)}(\overline{\bo\Gamma}) \bigodot f^{(2)}(\overline{\bo\Gamma}) \bigodot f^{(3)}(\overline{\bo\Gamma}) \right)\cdot \vect(\overline{\bo\Phi}).
\end{align}
Recall that under the assumptions in the theorem and the current notation, for each $j\in S^{(1)}\cup S^{(2)}\cup S^{(3)} $ the matrix $\bo\Lambda_j$ has full column rank $K$. According to the property of the Kronecker product, the matrices $\bigotimes_{g=1}^G \bo\Lambda_{j_g^{(1)}}$ and $\bigotimes_{g=1}^G \bo\Lambda_{j_g^{(2)}}$ each has full column rank $K^G$.
Further, since $\bo\Lambda_{j_g^{(3)}}$ has full column rank $K$, the Khatri-Rao product $\bo\Lambda_{j_g^{(3)}} \bigodot\left(\bigodot_{j\in S^{(0)}:\, \ell_{j,g}=1} \bo\Lambda_j\right)$ must have full column rank $K$. Therefore the matrix $\left\{\bigotimes_{g=1}^G \left(\bo\Lambda_{j_g^{(3)}} \bigodot\left(\bigodot_{j\in S^{(0)}:\, \ell_{j,g}=1} \bo\Lambda_j\right) \right)\right\}$ also has full column rank $K^G$. 
By definition of $f^{(m)}(\bo\Gamma)$'s, the above full-rank assertions indeed mean that $f^{(1)}(\bo\Gamma)$, $f^{(2)}(\bo\Gamma)$, and $f^{(3)}(\bo\Gamma)$ all have full column rank $K^G$.

We next invoke a useful lemma on the uniqueness of three-way tensor decompositions, the Kruskal's theorem established in \cite{kruskal1977three}, and then proceed similarly as the proof procedures in \cite{allman2009}. For a matrix $\bm M$, its Kruskal rank is defined to be the largest number $r$ such that any $r$ columns of $\bm M$ are linearly independent. Denote the Kruskal rank of matrix $\bm M$ by $\text{rank}_K(\bm M)$.

\begin{lemma}[Kruskal's Theorem]\label{lem-kruskal}
	Suppose $\bm M_1, \bm M_2, \bm M_3$ are three matrices of dimension $a_m\times K$ for $m=1,2,3$, $\bm N_1, \bm N_2, \bm N_3$ are three matrices each with $K$ columns, and $\bigodot_{m=1}^3 \bm M_m = \bigodot_{m=1}^3 \bm N_m$. If $\rank_{K}( \bm M_1) + \rank_{K}(\bm  M_2) + \rank_{K}(\bm  M_3) \geq 2K + 2$, then there exists a permutation matrix $\bm P$ and three invertible diagonal matrices $\bm D_m$ with $\bm D_1 \bm D_2 \bm D_3 =\mathbf I_K$ and $\bm N_m =\bm  M_m \bm D_m \bm P$ for each $m=1,2,3$.
\end{lemma}

If a matrix has full column rank $K$, then it must also have Kruskal rank $K$ by definition. As a corrolary of Lemma \ref{lem-kruskal}, if the three matrices $\bm M_1,\,\bm M_2,\,\bm M_3$ all have full column rank $K$, then the condition $\rank_{K}(\bm M_1) + \rank_{K}(\bm M_2) + \rank_{K}(\bm M_3)=3K \geq 2K + 2$ is satisfied and the uniqueness conclusion follows.
We now take $\bm M_m = f^{(m)}(\bo\Gamma)$, $\bm N_m=f^{(m)}(\overline{\bo\Gamma})$ for $m=1,2$, and define
\begin{align*}
    \bm M_3 = f^{(3)}(\bo\Gamma) \cdot \diag(\vect(\bo\Phi)),\quad 
    \bm N_3 = f^{(3)}(\overline{\bo\Gamma}) \cdot \diag(\vect(\overline{\bo\Phi})),
\end{align*}
then there is $\vect(\bo\Pi) = \bigodot_{m=1}^3 \bm M_m = \bigodot_{m=1}^3 \bm N_m$.
According to our argument right after \eqref{eq-gamdec}, $\rank_K(\bm M_m)= \rank_K(f^{(m)}(\bo\Gamma)) = K$ for $m=1,2$. As for $\bm M_3$, since $f^{(3)}(\bo\Gamma)$ has full column rank $K$ and the entries of $\bo\Phi$ are positive, the $\bm M_3$ also has full column rank $K$. 
Therefore, we can invoke Lemma \ref{lem-kruskal} to establish that there exists a permutation matrix $\bm P$ and three invertible diagonal matrices $\bm D_m$ with $\bm D_1 \bm D_2 \bm D_3 =\mathbf I_K$ such that
\begin{align*}
    f^{(m)}(\overline{\bo\Gamma}) = \bm N_m = \bm  M_m \bm D_m \bm P
    =f^{(m)}({\bo\Gamma}) \bm D_m \bm P
\end{align*}
for $m=1,2,3$.

The next step is to show that the diagonal matrices $\bm D_i$ are all identity matrices.
Note that each column of the $\prod_{j\in S^{(1)}} d_j\times K$ matrix $f^{(1)}({\bo\Gamma}) = \bigodot_{g=1}^G \bo\Gamma_{j_g^{(1)}} = \bigotimes_{g=1}^G \bo\Lambda_{j_g^{(1)}}$ 
characterizes the conditional joint distribution of $\{y_j:\, j\in S^{(1)}\}$ given the latent assignment vector $\bo z\in[K]^G$ under the true $\bo\Lambda$-parameters. 
And similarly, each column of $f^{(1)}(\overline{\bo\Gamma}) = \bigodot_{g=1}^G \overline{\bo\Gamma}_{j_g^{(1)}} = \bigotimes_{g=1}^G \overline{\bo\Lambda}_{j_g^{(1)}}$ 
characterizes the conditional joint distribution of $\{y_j:\, j\in S^{(1)}\}$ given $\bo z\in[K]^G$ under the alternative $\overline{\bo\Lambda}$.
Therefore the sum of each column of $f^{(1)}({\bo\Gamma})$ or that of $f^{(1)}(\overline{\bo\Gamma})$ equals one, which implies the diagonal matrix $\bm D_m$ is an identity matrix for $m=1$ or 2. 
Since Lemma \ref{lem-kruskal} ensures $\bm D_1\bm D_2\bm D_3=\II_K$, we also obtain $\bm D_3=\II_K$.  
By far we have obtained $f^{(m)}(\overline{\bo\Gamma}) = f^{(m)}({\bo\Gamma}) \bm P$ for $m=1,2$ and
\begin{align}\label{eq-f3gam}
    f^{(3)}(\overline{\bo\Gamma})\bcdot\text{diag}(\vect(\overline{\bo\Phi})) = f^{(3)}({\bo\Gamma})\bcdot \text{diag}(\vect({\bo\Phi})) \bm P.
\end{align} 
Note that the permutation matrix $\bm P$ has rows and columns both indexed by latent assignment vectors $\zz\in[K]^G$.
For $m=3$,
consider an arbitrary $\zz_1\in[K]^G$ and assume without loss of generality that $\zz_2\in[K]^G$ satisfies that the $(\zz_2,\zz_1)$th entry of matrix $\bm P$ is $\bm P_{\zz_2,\zz_1}=1$. 
Then the $\zz_1$th column of the matrix equality $f^{(3)}(\overline{\bo\Gamma})\bcdot\text{diag}(\vect(\overline{\bo\Phi})) = f^{(3)}({\bo\Gamma})\bcdot \text{diag}(\vect({\bo\Phi})) \bm P$ takes the form 
$$
f^{(3)}(\overline{\bo\Gamma})_{\cc,\zz_1}\cdot\vect(\overline{\bo\Phi})_{\zz_1}
=f^{(3)}({\bo\Gamma})_{\cc,\zz_2}\cdot\vect({\bo\Phi})_{\zz_2}
$$
for each $\cc\in\times_{j=1}^p [d_j]$; summing the above equality over the index $\cc\in\times_{j\in\mathcal A_1}[d_j]$ gives $\vect(\overline{\bo\Phi})_{\zz_1}=\vect(\bo\Phi)_{\zz_2}$. Note we have generally established $\vect(\overline{\bo\Phi})_{\zz_1}=\vect({\bo\Phi})_{\zz_2}$ whenever $\bm P_{\zz_2,\zz_1}=1$, which essentially implies $\vect(\overline{\bo\Phi})^\top=\vect({\bo\Phi})^\top\cdot \bm P$. Note that this shows the identifiability of the tensor core in our hybrid tensor decomposition formulation of the Gro-M$^3$.
Further, $\vect(\overline{\bo\Phi})^\top=\vect({\bo\Phi})^\top\cdot \bm P$ implies that $$\diag(\vect(\overline{\bo\Phi}))
=\diag(\vect({\bo\Phi})\cdot \bm P)
=\diag(\vect({\bo\Phi}))\cdot \bm P.$$ 
Combining the above display to the previous \eqref{eq-f3gam}, since $\diag(\vect(\overline{\bo\Phi}))$ is a diagonal matrix with positive diagonal entries, we can right multiply the inverse of this matrix with the LHS of \eqref{eq-f3gam} and meanwhile right multiply the inverse of $\diag(\vect({\bo\Phi}))\cdot \bm P$ with the RHS of \eqref{eq-f3gam}; this gives $f^{(3)}(\overline{\bo\Gamma}) = f^{(3)}({\bo\Gamma})\bm P$.

Our final step of proving the theorem is to show that the established $f^{(m)}(\overline{\bo\Gamma}) = f^{(m)}({\bo\Gamma})\bm P$ for $m=1,2,3$ implies the identifiability of $\bo\Lambda$ and $\LL$. 
First, since $f^{(m)}({\bo\Gamma})$ is defined as certain Khatri-Rao products of the individual $\bo\Gamma_j$'s, we claim that the $f^{(m)}(\overline{\bo\Gamma}) = f^{(m)}({\bo\Gamma})\bm P$ indeed implies that ${\overline{\bo\Gamma}}_j = {{\bo\Gamma}}_j \bm P$ for each $j\in[p]$. 
To see this, note that each column of $f^{(m)}(\overline{\bo\Gamma})$ and  $f^{(m)}({\bo\Gamma})\bm P$ characterizes the conditional joint distribution of variables $\{y_j:\, j\in S^{(m)}\}$ given the $\zz$. So the conditional marginal distribution $\bo\Gamma_j$ can be obtained by summing up appropriate row vectors of the matrices $f^{(m)}(\overline{\bo\Gamma})$ and  $f^{(m)}({\bo\Gamma})\bm P$, corresponding to marginalizing out other variables except the $j$th one.
Now without loss of generality we can assume that $\bm P=\II_{K^G}$, then ${\overline{\bo\Gamma}}_j = {\bo\Gamma}_j \bm P$ gives $\bar\gamma_{j,c_j,\zz} = \gamma_{j,c_j,\zz}$ for all $\zz\in[K]^G$.
Now it only remains to show that $\bo\Gamma$ uniquely determines $\bo\Lambda$ and $\LL$. By definition \eqref{eq-newlambda} there is $\gamma_{j,c_j,\zz} = \lambda_{j,c_j,z_{s_j}}$. For arbitrary $s_j$ and $\overline s_j$, we first consider an arbitrary latent assignment $\zz\in[K]^G$ such that $z_{s_j} = z_{\overline s_j}$, then 
$$
\lambda_{j,c_j,k} = \lambda_{j,c_j,z_{s_j}} = \gamma_{j,c_j,\zz} = 
\overline \gamma_{j,c_j,\zz}
= \overline\lambda_{j,c_j,z_{\overline s_j}} 
= \overline\lambda_{j,c_j,k}.
$$
The above reasoning proves the identifiability of $\bo\Lambda$. Thus far we have proved part (a) of Theorem \ref{thm-attr1}.

We next prove part (b) of the theorem.
We use proof by contradiction to show the identifiability of the grouping matrix $\LL$ (or equivalently, the identifiability of the vector $\bo s$). If there exists some $j\in[p]$ such that the $j$th rows of $\LL$ and $\overline\LL$ are different, then $s_j \neq \overline s_j$; denote $s_j =: g$ and $\overline s_j=:g'$. 
Next for arbitrary two different indices $k, k'\in[K]$ and $k\neq k'$, we consider a latent assignment $\zz\in[K]^G$ such that $z_g=k$ and $z_g=k'$. Then there are 
$$
\lambda_{j,c_j,k} = \lambda_{j,c_j,z_{s_j}} = \gamma_{j,c_j,\zz} = 
\overline \gamma_{j,c_j,\zz}
= \overline\lambda_{j,c_j,z_{\overline s_j}} 
= \overline\lambda_{j,c_j,k'}~\text{ for all }~ c_j\in[d_j].
$$
Since $k\neq k'$, the above equality means the $k$th and $k'$th columns of $\bo\Lambda_j$ are identical. Since $k$ and $k'$ are two arbitrary indices, this means all the column vectors in the matrix $\bo\Lambda_j$ are identical.
This contradicts the assumption in part (b) of the theorem.
Therefore we have shown that $s_j=\overline s_j$ must hold for an arbitrary $j\in[p]$. This proves the identifiability of $\LL$ from $\bo\Gamma$. 
This completes the proof of Theorem \ref{thm-attr1}.


\subsection{Proof of Theorem \ref{thm-attr-finer}}
The proof of this theorem is similar in spirit to the previous Theorem \ref{thm-attr1} by exploiting the inherent tensor decomposition structure, but differing in taking advantage of the more dimension-grouping structure under the assumptions here. 
Recall $\mca_g=\{g\in[p]:\, \ell_{j,g}=1\} = \cup_{m=1}^3 \mca_{g,m}$.
We write $\vect(\bo\Pi)$ under the true parameters as
\begin{align*}
    &~\vect(\bo\Pi) 
    = \left\{\bigotimes_{g=1}^G \bigodot_{j\in\mca_g} \bo\Lambda_j\right\}
     \cdot \vect(\bo\Phi)
    \\ \notag
    = &~ \left\{\bigotimes_{g=1}^G \left(\bigodot_{m=1}^3 \bigodot_{j\in\mca_{g,m}} \bo\Lambda_j\right)\right\}
     \cdot \vect(\bo\Phi)
    \\ \notag
    = &~
    \left[
    \left\{\bigotimes_{g=1}^G \left(\bigodot_{j\in\mca_{g,1}}\bo\Lambda_{j}\right)\right\}
    \bigodot
    \left\{\bigotimes_{g=1}^G \left(\bigodot_{j\in\mca_{g,2}}\bo\Lambda_{j}\right)\right\}
    \bigodot
    \left\{\bigotimes_{g=1}^G \left(\bigodot_{j\in\mca_{g,3}}\bo\Lambda_{j}\right)\right\}
    \right]\cdot 
    \vect(\bo\Phi).
\end{align*}
Since each $\mca_{g,m}$ is nonempty under the assumption in the theorem and $\cup_{g=1}^G \mca_{g,m} \supseteq [G]$, we can use Lemma \ref{lem-gamma} to obtain that 
$$
\bigodot_{j\in \cup_{g=1}^G \mca_{g,m}} \bo\Gamma_j
=
\bigotimes_{g=1}^G \left(\bigodot_{j\in\mca_{g,m}}\bo\Lambda_{j}\right)
=
\bigotimes_{g=1}^G \tilde{\bo\Lambda}_{g,m},
$$
where the second equality above follows from the definition of $\tilde{\bo\Lambda}_{g,m}$ in the theorem.
We now define
$$
 f^{(m)}(\bo\Gamma) := \bigodot_{j\in \cup_{g=1}^G \mca_{g,m}} \bo\Gamma_j, \quad m=1,2,3.
$$
Since the theorem has the assumption that each $\tilde{\bo\Lambda}_{g,m}$ has full column rank $K$, we have that $f^{(m)}(\bo\Gamma)$ has full rank $K$. 
Note that each $f^{(m)}(\cdot)$ is the Khatri-Rao product of certain $\bo\Gamma_j$'s and $f^{(m)}$ depends on the true grouping matrix $\LL$. Also note that there is
\begin{align*}
    \vect(\bo\Pi) = &\left( f^{(1)}(\bo\Gamma) \bigodot f^{(2)}(\bo\Gamma) \bigodot f^{(3)}(\bo\Gamma) \right)\cdot \vect(\bo\Phi)
    \\ \notag
    = &\left( f^{(1)}(\overline{\bo\Gamma}) \bigodot f^{(2)}(\overline{\bo\Gamma}) \bigodot f^{(3)}(\overline{\bo\Gamma}) \right)\cdot \vect(\overline{\bo\Phi}).
\end{align*}
Now the problem is in exactly the same formulation as that in the proof of Theorem \ref{thm-attr1}, so we can proceed in the same way to establish the identifiability of $\bo\Phi$ and individual $\bo\Gamma_j$'s. The identifiability of $\bo\Gamma$ further gives the identifiability of $\bo\Lambda$ and $\LL$. This finishes the proof of Theorem \ref{thm-attr-finer}.
\qed

\subsection{Proof of Theorem \ref{thm-attrgen}}
We first prove the following claim: 

\noindent\textbf{Claim 1.} Under condition \eqref{eq-djk} that $\prod_{j\in\mca_{g,m}} d_j\geq K$ in the theorem, the following matrix $\tilde{\bo\Lambda}_{g,m}$ has full column rank $K$ for generic parameters,
\begin{align*}
    \tilde{\bo\Lambda}_{g,m} = \bigodot_{j\in\mca_{g,m}} \bo\Lambda_j.
\end{align*}
The $\tilde{\bo\Lambda}_{g,m}$ above has the same definition as that in Theorem \ref{thm-attr-finer}.
The proof of this claim is similar in spirit to that of Lemma 13 in \cite{allman2009}. 
Note that the statement that the $\prod_{j\in  \mca_{g,m}} d_j \times K$ matrix $\tilde{\bo\Lambda}_{g,m}$ does not have full column rank is equivalent to the statement that the maps sending $\tilde{\bo\Lambda}_{g,m}$ to its $K\times K$ minors are all zero maps. There are 
$$\binom{\prod_{j\in \mca_{g,m}} d_j}{K}$$
such maps, and each of this map is a polynomial with indeterminants   $\lambda_{j,c_j,k}$'s.
To show that $\tilde{\bo\Lambda}_{g,m}$ has full column rank $K$ for generic parameters, we just need to show that these maps are not all zero polynomials. 
According to the property of the polynomial maps, it indeed suffices to find one particular set of $\{\bo\Lambda_j;\, j\in \mca_{g,m}\}$ such that the resulting Khatri-Rao product $\tilde{\bo\Lambda}_{g,m}$ has full column rank.

Consider a set of distinct prime numbers denoted by $\{a_{j,c};\,j=1,\ldots,p,\, c=1,\ldots,d_j\}$. Define
 \begin{align}\label{eq-vdm}
 	\bo\Lambda_j^\star = 
 	\begin{pmatrix}
 		1 & a_{j,1}   &  a^2_{j,1}  & \cdots & a^{K-1}_{j,1} \\
 		1 & a_{j,2}   &  a^2_{j,2}  & \cdots & a^{K-1}_{j,2} \\
 		\vdots & \vdots   &  \vdots  & \vdots & \vdots \\
 		1 & a_{j,d_j} &  a^2_{j,d_j}  & \cdots & a^{K-1}_{j,d_j} \\
 	\end{pmatrix},
 \end{align}
 then $\bo\Lambda_j^\star$ is a $d_j\times K$ Vandermonde matrix. 
 Generally, for a $d$-dimensional vector $\bo b$, let $\text{VDM}(\bo b) = \text{VDM}(b_1,\ldots, b_d)$ denote the the $d\times d$ Vandermonde matrix with the $(i,c)$th entry being $b_i^{c-1}$, so the $\bo\Lambda_j^\star$ defined in \eqref{eq-vdm} can be written as $\bo\Lambda_j^\star=\text{VDM}(a_{j,1},\ldots,a_{j,d_j})$.
 Now consider a $\tilde{\bo\Lambda}^\star_{g,m}$ defined as 
 $$
 \tilde{\bo\Lambda}^\star_{g,m} = \bigodot_{j\in\mca_{g,m}} \bo\Lambda^\star_j.
 $$
 Under the assumption \eqref{eq-djk} in the theorem that $\prod_{j\in \mca_{g,m}} d_j\geq K$, the $K$ columns of $\tilde{\bo\Lambda}_{g,m}$ are indeed the first $K$ columns in the following Vandermonde matrix
 \begin{align*}
 	\bm V_{g,m}=\text{VDM}\left(\prod_{j\in \mca_{g,m}} a_{j,1},\ldots,\prod_{j\in \mca_{g,m}} a_{j,d_j}\right).
 \end{align*}
Since by construction the $a_{j,c}$'s are distinct prime numbers, for each $j\in S_m$ the $d_j$ products $\prod_{j\in \mca_{g,m}} a_{j,1},~\allowbreak\ldots,~ \allowbreak \prod_{j\in \mca_{g,m}} a_{j,d_j}$ are also distinct numbers. Therefore the $\bm V_{g,m}$ defined above has full rank $\prod_{j\in S_m} d_j$. Since $\prod_{j\in \mca_{g,m}} d_j\geq K$ and $\tilde{\bo\Lambda}^\star_{g,m}$ has columns from the first $K$ columns of $\bm V_{g,m}$, we have that $\tilde{\bo\Lambda}^\star_{g,m}$ has full column rank $K$ for this particular choice of parameters. 
Note that the $\bo\Lambda_j^\star$ defined in \eqref{eq-vdm} does not have each column summing up to one, as is required in the parameterization of the probability tensor. But performing a positive rescaling of the each column of $\bo\Lambda_j^\star$ to a conditional probability table $\bo\Lambda_j$ would not change the above reasoning and conclusion about matrix rank; so we have proved the earlier Claim 1 that each $\tilde{\bo\Lambda}_{g,m}$ has full column rank $K$ for generic parameters. Given this conclusion, for generic parameters in the parameter space the situation is reduced back to that under Theorem \ref{thm-attr-finer}. So the identifiability condition in Theorem \ref{thm-attr1} carries over, and we can obtain the conclusion that identifiability holds here for generic parameters. This completes the proof of Theorem \ref{thm-attrgen}.
\qed

\subsection{Proof of Proposition \ref{prop-dir}}
Recall that under the conditions of the previous theorem, we already have the conclusion that $\bo\Lambda$ and $\bo\Phi$ are identifiable. Next the question boils down to whether $(\alpha_1,\ldots,\alpha_K)$ are identifiable from $\bo\Phi=(\phi_{k_1,\ldots,k_G})$.
By the definition, we have
\begin{align*}
    \phi_{k_1,\ldots,k_G} = \mathbb E_{\bo\pi\sim\text{Dir}(\aaa)}\left[\pi_{k_1}\cdots \pi_{k_G}\right]
    = \int_{\Delta^{K-1}} \pi_{k_1}\cdots \pi_{k_G} d\text{Dir}_{\aaa}(\bo\pi).
\end{align*}
First we consider the case of $G=2$. Denote $\alpha_0=\sum_{k=1}^K\alpha_k$. Then according to the moment property of the Dirichlet distribution, 
there is
\begin{align*}
    \mathbb E_{\bo\pi\sim\text{Dir}(\aaa)}\left[\pi_{k}\pi_{\ell}\right]
    =
    \begin{cases}
    \dfrac{\alpha_k\alpha_{\ell}}{\alpha_0(\alpha_0+1)}, & \text{if }k\neq\ell;\\[3mm]
    \dfrac{\alpha_k(\alpha_k+1)}{\alpha_0(\alpha_0+1)}, & \text{if }k=\ell.
    \end{cases}
\end{align*}
Therefore for $k\neq \ell$, consider $x$ and $y$ defined as follows,
\begin{align*}
    x:=\dfrac{\mathbb E_{\bo\pi\sim\text{Dir}(\aaa)}\left[\pi_{k}^2\right]}{\mathbb E_{\bo\pi\sim\text{Dir}(\aaa)}\left[\pi_{k}\pi_{\ell}\right]} 
    &= \dfrac{\alpha_k+1}{\alpha_\ell}, \\[2mm]
    y:=\dfrac{\mathbb E_{\bo\pi\sim\text{Dir}(\aaa)}\left[\pi_{\ell}^2\right]}{\mathbb E_{\bo\pi\sim\text{Dir}(\aaa)}\left[\pi_{k}\pi_{\ell}\right]} 
    &= \dfrac{\alpha_\ell+1}{\alpha_k}.
\end{align*}
Since $x$ and $y$ are already identified, then we can solve for $\alpha_k$ and $\alpha_\ell$ as follows
\begin{align*}
    \alpha_k = \frac{x+1}{xy-1},\quad
    \alpha_\ell = \frac{y+1}{xy-1}.
\end{align*}
Since the above reasoning holds for arbitrary pairs of $(k,\ell)$ with $k\neq\ell$, we have obtained the identifiability of the entire vector $\aaa=(\alpha_1,\ldots,\alpha_K)$.

Next we consider the general case of $G>2$. 
For arbitrary $1\leq k\neq \ell\leq K$, consider two sequences $(k,k,k_3,\ldots,k_G)$,  $(k,\ell,k_3,\ldots,k_G)\in[K]^G$.
According to the property of the Dirichlet distribution, we have
\begin{align}\label{eq-dir-u}
    \frac{\mathbb E_{\bo\pi\sim\text{Dir}(\aaa)}\left[\pi_{k}\pi_k \pi_{k_3}\cdots\pi_{k_G}\right]}
    {\mathbb E_{\bo\pi\sim\text{Dir}(\aaa)}\left[\pi_{k}\pi_\ell \pi_{k_3}\cdots\pi_{k_G}\right]}
    =
    \frac{\alpha_k+1}{\alpha_\ell},\\[2mm]
    \label{eq-dir-v}
    \frac{\mathbb E_{\bo\pi\sim\text{Dir}(\aaa)}\left[\pi_{\ell}\pi_\ell \pi_{k_3}\cdots\pi_{k_G}\right]}
    {\mathbb E_{\bo\pi\sim\text{Dir}(\aaa)}\left[\pi_{k}\pi_\ell \pi_{k_3}\cdots\pi_{k_G}\right]}
    =
    \frac{\alpha_\ell+1}{\alpha_k}.
\end{align}
Now that the left hand sides of the above two equations are identified by the previous theorem, we denote them by $u:=\text{LHS of }\eqref{eq-dir-u}$ and $v:=\text{LHS of }\eqref{eq-dir-v}$. The $u$ and $v$ are identified constants. Solving for $\alpha_k$ and $\alpha_\ell$ gives
$$
\alpha_k = \frac{u+1}{uv-1},\quad
    \alpha_\ell = \frac{v+1}{uv-1}.
$$
Since $k,\ell$ are arbitrary, we have shown that the entire vector $\aaa=(\alpha_1,\ldots,\alpha_K)$ is identifiable. This completes the proof of Proposition \ref{prop-dir}.
\qed

\subsection{Proof of Supporting Lemmas}
\begin{proof}[Proof of Lemma \ref{lem-gamma}]
First note that the $\bigodot_{j\in S} \bo\Gamma_{j}$ on the right hand side of \eqref{eq-gamma} has size $\prod_{j\in S}d_j \times K^G$. Further, since $\{s_{j}:\, j\in S\} \supseteq  [G]$, the set $\{j\in S:\, s_j=g\}$ is nonempty. So the $\bigodot_{j\in S:\, s_j=g} \bo\Lambda_j$ has $K$ columns and hence the left hand side of \eqref{eq-gamma} also has size $\prod_{j\in S}d_j \times K^G$.
Without loss of generality, suppose $S=\{1,2,\ldots,|S|\}$, where $|S|$ denotes the cardinality of the set $S$.
The $(c_1 + (c_2-1)d_1 + \cdots + (c_{|S|}-1)d_1\cdots d_{|S|-1}, \; z_1 + (z_2-1)K + \cdots + (z_G-1)K^{G-1})$th entry of the RHS of \eqref{eq-gamma} is 
$\prod_{j\in S} \gamma_{j,c_j,\bo z}$, which by definition equals $\prod_{j\in S} \lambda_{j,c_j,z_{s_j}} = \prod_{g=1}^G \prod_{j\in S:\, s_j=g} \lambda_{j,c_j,z_{g}}$; this is exactly the $(c_1 + (c_2-1)d_1 + \cdots + (c_{|S|}-1)d_1\cdots d_{|S|-1}, \; z_1 + (z_2-1)K + \cdots + (z_G-1)K^{G-1})$th entry of the LHS of \eqref{eq-gamma}. This completes the proof of Lemma \ref{lem-gamma}.
\end{proof}

\begin{proof}[Proof of Lemma \ref{lem-vect}]
First note that both hand sides of \eqref{eq-vecpi} are vectors of size $\prod_{j=1}^p d_j \times 1$.
To see this for the right hand side of \eqref{eq-vecpi}, note the matrix $\bigodot_{j:\, \, \ell_{j,g}=1} \bo\Lambda_j$ has size $\prod_{j:\, \ell_{j,g}=1} d_j \times K^{\sum_{j=1}^p \ell_{j,g}}$, and hence the matrix $\bigotimes_{g=1}^G \bigodot_{j:\, \ell_{j,g}=1} \bo\Lambda_j$ has size 
$$\prod_{g=1}^G\prod_{j:\, \ell_{j,g}=1} d_j \times K^{\sum_{g=1}^G \sum_{j=1}^p \ell_{j,g}},$$
which is just $\prod_{j=1}^p d_j \times K^G.$
Further note that the vector $\vect(\bo\Phi)$ has size $K^G \times 1$, so the $\left\{\bigotimes_{g=1}^G \bigodot_{j:\, \ell_{j,g}=1}  \bo\Lambda_j\right\} \vect(\bo\Phi)$ on the right hand side of \eqref{eq-vecpi} has size $\prod_{j=1}^p d_j \times 1$, matching the size of the left hand side. Next consider the individual entries of both hand sides of \eqref{eq-vecpi}.
First, by definition of the vec() operator, the $[c_1 + (c_2-1)d_1 + \cdots + (c_p-1)d_1\cdots d_{p-1}]$-th entry of the left hand side of \eqref{eq-vecpi} is $\pi_{c_1,\ldots,c_p}$.
%
Next, according to \eqref{eq-c3m-repar}, the $\pi_{c_1,\ldots,c_p}$ can be written in the following way,
\begin{align*}
\pi_{c_1,\ldots,c_p}=   
    &~\sum_{z_1=1}^K \cdots\sum_{z_G=1}^K \phi_{z_1,\ldots,z_G} \prod_{j=1}^p \gamma_{j, c_j, \bo z}\\
    =
    &~
    \sum_{z_1=1}^K \cdots\sum_{z_G=1}^K
    \vect(\bo\Phi)_{z_1 + (z_2-1)K + \cdots + (z_G-1) K^{G-1}}
    \times
    \prod_{g=1}^G \prod_{j:\, \ell_{j,g}=1} \lambda_{j,c_j, z_g}
    \\
    =
    &~
    \sum_{z_1=1}^K \cdots\sum_{z_G=1}^K
    \vect(\bo\Phi)_{z_1 + (z_2-1)K + \cdots + (z_G-1)K^{G-1}}
    \\
    &\qquad \times \left\{ 
    \bigotimes_{g=1}^G \bigodot_{j:\, \ell_{j,g}=1} \bo\Lambda_j
    \right\}_{c_1 + (c_2-1)d_1 + \cdots + (c_p-1)d_1\cdots d_{p-1}, \; z_1 + (z_2-1)K + \cdots + (z_G-1)K^{G-1}}
    \\[2mm]
    =&~
    \vect(\bo\Phi)^\top \cdot \left\{ 
    \bigotimes_{g=1}^G \bigodot_{j:\, \ell_{j,g}=1} \bo\Lambda_j
    \right\}^\top_{c_1 + (c_2-1)d_1 + \cdots + (c_p-1)d_1\cdots d_{p-1},\, \bcolon}.
\end{align*}
The last row in the above display exactly equals the $[c_1 + (c_2-1)d_1 + \cdots + (c_p-1)d_1\cdots d_{p-1}]$-th entry of the RHS of \eqref{eq-vecpi}. This proves the equality in \eqref{eq-vecpi} and completes the proof of Lemma \ref{lem-vect}.
\end{proof}

\begin{proof}[Proof of Lemma \ref{lem-krkr}]
In the LHS of \eqref{eq-krkr}, the term $\bo\Lambda_{j^{(1)}_g} \bigodot \bo\Lambda_{j^{(2)}_g}$ has size $d_{j_g^{(1)}}d_{j_g^{(2)}} \times K$ and hence the Kronecker product $\bigotimes_{g=1}^G \left\{\bo\Lambda_{j^{(1)}_g} \bigodot \bo\Lambda_{j^{(2)}_g} \right\}$ has size $\prod_{j\in S^{(1)} \cup S^{(1)}} d_j \times K^G$. 
In the RHS of \eqref{eq-krkr}, the term $ \left\{\bigotimes_{g=1}^{G} \bo\Lambda_{j^{(1)}_g}\right\}$ has size $\prod_{g=1}^G d_{j_g^{(1)}} \times K^G$, and hence the Khatri-Rao product of two such terms
$$
\left\{\bigotimes_{g=1}^{G} \bo\Lambda_{j^{(1)}_g}\right\} \bigodot \left\{\bigotimes_{g=1}^{G}\bo\Lambda_{j^{(2)}_g} \right\}
$$
has size $\left(\prod_{g=1}^G d_{j_g^{(1)}} d_{j_g^{(2)}}\right) \times K^G$.
So both hand sides of \eqref{eq-krkr} has size $\prod_{j\in S^{(1)} \cup S^{(2)}} d_j \times K^G$. 
The equality \eqref{eq-krkr} can be similarly shown as in the proof of Lemma \ref{lem-vect} by writing out and checking the individual elements of the two matrices on the LHS and RHS of \eqref{eq-krkr}.

Similarly, the LHS and RHS of \eqref{eq-krkr2} both have size $\prod_{j\in S^{(1)} \cup S^{(1)}\cup\left\{j_G^{(3)}\right\}} d_j \times K^G$ and the equality can be similarly shown as in the proof of Lemma \ref{lem-vect}.
\end{proof}


\section{Supplement B: Pairwise Cramer's V between Categorical Variables}\label{sec-mi}

According to the definition of mutual information in information theory, for two discrete variables $y_{j}\in[d_{j}]$ and $y_{m}\in[d_{m}]$, their Cramer's V is 
\begin{align}\label{eq-mi-pop}
    \text{CRV}(y_{j}, y_{m}) = 
    \left\{\frac{1}{\min(d_j,~ d_m)}
    \sum_{c_1\in[d_{j}]} \sum_{c_2\in[d_{m}]}  
    \frac{\left(p_{(y_{j}, y_{m})}(c_1, c_2)  - p_{y_{j}}(c_1) p_{y_{m}}(c_2)\right)^2}{p_{y_{j}}(c_1) p_{y_{m}}(c_2)} \right\}^{1/2}
\end{align}
where $p_{y_{j}}(c_1) = \mathbb P (y_j = c_1)$ denotes the marginal distribution of $y_j$ and $p_{(y_{j}, y_{m})}(c_1, c_2) = \mathbb P (y_j=c_1, y_m=c_2)$ denotes the joint distribution of $y_j$ and $y_m$. The Cramer's V measures the the inherent dependence expressed in the joint distribution of two variables relative to their marginal distributions under the independence assumption. Therefore, Cramer's V measures the dependence between variables and it equals zero if and only of the two variables are independent; otherwise Cramer's V is positive.

The expression of Cramer's V in \eqref{eq-mi-pop} is the population version. Given a sample $\bo y_1, \ldots, \bo y_n$ with $\bo y_i = (y_{i,1},\ldots, y_{i,p})$, the population quantities of the marginal and joint distributions in \eqref{eq-mi-pop} can be replaced by their sample estimates. That is, the previous $p_{y_{j}}(c_1)$ and $p_{(y_{j}, y_{m})}(c_1, c_2)$ are replaced by the following,
\begin{align*}
    p^{\text{samp}}_{y_{j}}(c_1) 
    =  \frac{1}{n} \sum_{i=1}^n \mathbb I(y_{i,j}=c_1),\quad
    p^{\text{samp}}_{(y_{j}, y_{m})}(c_1, c_2)
    =
     \frac{1}{n} \sum_{i=1}^n\mathbb I (y_{i,j}=c_1, y_{i,m}=c_2).
\end{align*}
Using the sample-based Cramer's V measure, we calculate the Cramer's V for all the pairs of variables when $j$ and $m$ each range from 1 to $p$. For two randomly chosen simulated datasets from the simulations settings $p=30,~ G=6,~ K=3,~ n=1000$ and $p=90,~ G=15,~ K=3,~ n=1000$ described in Section \ref{sec-simu} in the main text, their pairwise Cramer's V plots are displayed in Figure \ref{fig-mi-simu}.

\begin{figure}[h!]
    \centering
    \begin{subfigure}[b]{0.49\textwidth}
    \includegraphics[width=\textwidth]{figures/simu_crv_G6.png}
    \caption{$p=30,~ G=6,~ K=3,~ n=1000$.}
    \end{subfigure}
    \hfill
    \begin{subfigure}[b]{0.49\textwidth}
    \includegraphics[width=\textwidth]{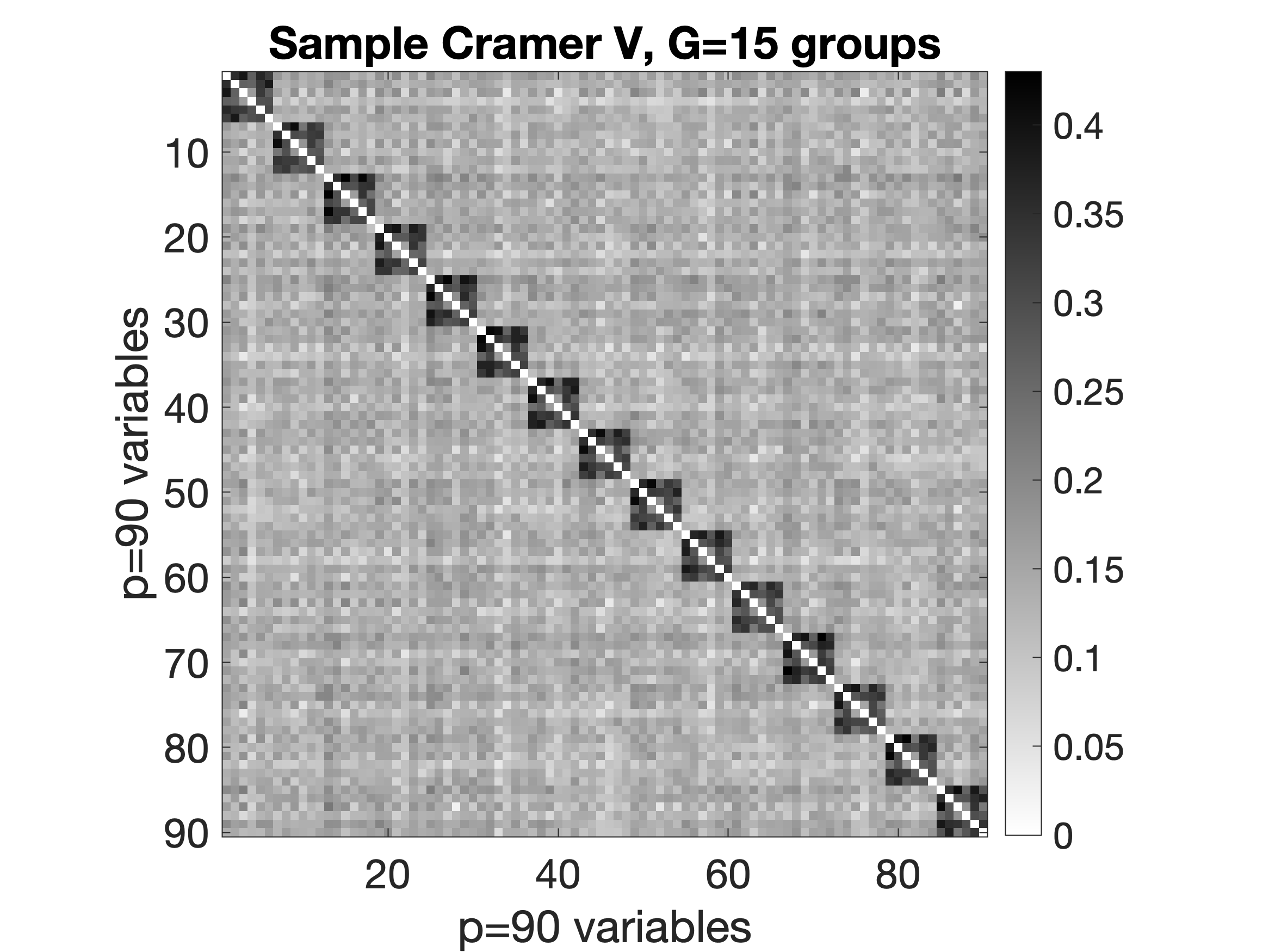}
    \caption{$p=90,~ G=15,~ K=3,~ n=1000$.}
    \end{subfigure}
    
    \caption{Cramer's V of item pairs for two simulated datasets.}
    \label{fig-mi-simu}
\end{figure}

By visual inspection, Figure \ref{fig-mi-simu} shows a block-diagonal structure of the $p\times p$ pairwise Cramer's V matrix for both simulation settings. In each of these settings, the true grouping matrix $\LL$ used to generate data takes the form that the first $p/G$ variables belong to a same group, the second $p/G$ variables belong to another same group, etc. Therefore, Figure \ref{fig-mi-simu} implies in the simulations, the variables belonging to the same group tend to show higher dependence than those variables belonging to different groups.

\end{document}